\documentclass[a4paper,11pt]{article}

\usepackage[utf8]{inputenc}
\usepackage{graphicx}        % standard LaTeX graphics tool
\usepackage{color}
\usepackage{amsfonts}
\usepackage{amsmath}
\usepackage{amssymb}
\usepackage{amsthm}
\reversemarginpar

\usepackage{epsfig,verbatim}
\usepackage{subfigure}
\usepackage{amsmath, amssymb, graphics,amsthm}
\usepackage[dvipsnames]{xcolor}
\usepackage{slashed}            % for slashed characters in math mode
\usepackage{bm}                % for \mathbbm{1} (unit matrix)
\usepackage{units}              % for \unit[AMOUNT]{UNIT}
\usepackage{xspace}             % for 'smart' spaces after mathematics
\usepackage{enumerate}          % for 'enumerate' lists with custom labels
\usepackage{soul}
\usepackage{upgreek}
\usepackage{marginnote}

\usepackage{authblk}

\newcommand{\mathsym}[1]{{}}

\newcommand{\LL}{\mathcal{L}}
\newcommand{\cu}{{\cal{U}}}
\newcommand{\cupo}{{\cal{U}}_{p}}

\newtheorem{theorem}{Theorem}[section]
\newtheorem{proposition}{Proposition}[section]
\newtheorem{lemma}{Lemma}[section]
\theoremstyle{definition}
\newtheorem{definition}{Definition}[section]

\theoremstyle{remark}
\newtheorem{remark}{Remark}[theorem]

\numberwithin{equation}{section}
\theoremstyle{plain}

\newtheorem{corollary}[theorem]{Corollary}

\def\proof{\noindent{\em Proof.\/}\hspace{1mm}}

\def\Sup {\hat{\Sigma}_\varepsilon}
\def\lup{\vec{\hat{l}}_\varepsilon}
\def\nup{\boldsymbol{\hat{n}}_\varepsilon}
\def\Sdown {\Sigma_\varepsilon}
\def\ldown {\vec{l}_\varepsilon}
\def\ndown{\boldsymbol{n}_\varepsilon}

\def \tetrad {\mu}
\def \invtetrad {\nu}
\def \ta {\mathfrak{a}}
\def \tb {\mathfrak{b}}
\def \dimension {m}
\def \abshyp{\sigma}
\def \norm {\xi}
\begin{document}
\title{First order perturbations of hypersurfaces of arbitrary causal character}
\author{Brien C. Nolan$^{1}$ \thanks{ 
		brien.nolan@dcu.ie} , Borja Reina$^{1,2}$\thanks{ 
	 borja.reina@dcu.ie} ,  
Kepa Sousa$^{3}$ \thanks{ 
kepa.sousa@csic.es} \\
$^{1}$Centre for Astrophysics and Relativity, School of Mathematical Sciences,  Dublin City University, Glasnevin,  Dublin 9,  Ireland \\
$^{2}$Departamento de F\'isica Te\'orica e Historia de la Ciencia, University of the Basque Country UPV/EHU, Apartado 644, 48080 Bilbao, Spain\\
$^{3}$Instituto de F\'isica Te\'orica UAM-CSIC
Universidad Aut\'onoma de Madrid, Cantoblanco, 28049 Madrid, Spain
}
\maketitle

\begin{abstract}
	In this work we study the problem of first order perturbations of a general hypersurface, i.e. with arbitrary causal character at each point. We extend the framework by Mars \cite{Mars2005} where this problem was studied to second order for everywhere timelike or spacelike hypersurfaces, and we adapt it to cover the general case. We apply the formalism to the matching of spacetimes across a general hypersurface to first order in perturbation theory.
\end{abstract}

\section{Introduction}
There are many problems in General Relativity that for different reasons have proven to be very difficult to study in the exact theory. For instance this can be due to the lack of symmetries. One possible approach is then to identify a relevant parameter in the problem and resort to perturbation theory. In this context, the background spacetime corresponds to a known exact solution, normally with a high degree of symmetry, and the perturbation of the geometry is encoded in a two covariant symmetric tensor field that satisfies the perturbed field equations.  However, the situation becomes more complicated if, apart from the spacetime itself, there is some special hypersurface sensitive to the perturbation involved in the problem. There are many relevant examples that illustrate this situation, for instance a background spacetime which is an outcome of a matching procedure. In this context, the matching hypersurface is a clear candidate to experience some kind of deformations. Settings of this type have been exploited to tackle a wide variety of problems in the context of relativistic astrophysics, where the global spacetime describing a star can be broken down into two different spacetimes, one accounting for the stellar interior and another one for the vacuum, and both of them are matched across a timelike hypersurface that acts as a boundary separating those two regions. For instance pulsating stars have been studied to first order (see \cite{Price_Thorne} and the related works in the series), or more recently, stellar collapse to second order \cite{Brizuela}, just to list some remarkable works that consider perturbation theory up to different orders.

Nevertheless, a matching is not a requisite for having a hypersurface that is being deformed by some perturbations, as one could consider a wide variety of problems related to hypersurfaces sitting in an ambient spacetime, and for instance, in \cite{VegaPoissonMassey} a formalism is presented to study perturbations of a null hypersurface and then applied to the setting of a black hole immersed in a tidal environment, although it is also suitable for  different problems in the context. Therefore the problem of deforming submanifolds is by itself general and interesting enough to deserve some attention in the context of mathematical relativity. Indeed, there have been many general approaches such as \cite{Battye01}, where perturbations of a brane in an arbitrary bulk were computed; \cite{Capovilla:1994bs}, where perturbations of timelike submanifolds of arbitrary codimension are studied; or the doubly-covariant approach by \cite{Mukohyama00}, where the different sources of gauge freedom are identified.
%and studied in detail. 
 We will focus in the remarkable work \cite{Mars2005}, which is a fully general framework to study first and second order perturbations to hypersurfaces and where the key idea is to formulate the perturbations in an abstract hypersurface detached from the spacetime, so that the perturbations of the first and second fundamental forms can be computed in a very straight and transparent manner (the second order perturbations would be certainly difficult to compute without using the methods therein). As in any approach based on perturbations, there is some inherent freedom in the method, namely the spacetime gauge transformations and hypersurface gauge transformations, but their effect on the perturbations is studied separately. There is, however, one assumption taken explicitly in \cite{Mars2005} regarding the causal character of the hypersurfaces, as they are assumed to be timelike or spacelike everywhere. This method has been applied to linear order to obtain uniqueness results for the Einstein-Strauss model \cite{MarsMenaVera2008} and to second order to revisit the problem of slowly rotating stars \cite{ReinaVera2014}.

The problem of hypersurfaces of changing causal character has been studied in the work  \cite{Mars1993}, and recently revisited  in \cite{Mars_constraints}.
The main difference with respect to the standard timelike or spacelike hypersurfaces,  is that at null points the first fundamental form becomes degenerate, and therefore it does not define a proper geometric structure, so that additional ingredients are needed such as a transverse vector to the hypersurface, used to endow the submanifold with a \textit{rigged} connection \cite{schouten2013ricci}, \cite{Mars1993}. Note that the transverse vector is non-unique, and it adds a gauge freedom in the method. We refer to this vector as a rigging.

But even keeping the causal character fixed, to the best of our knowledge, there is only one framework dealing with perturbations of null hypersurfaces \cite{VegaPoissonMassey}, where the focus is put on the generators of the null hypersurface. In this work we consider general hypersurfaces, and therefore the approach based in the generators is not convenient, so that we start from the framework \cite{Mars2005} and generalize it appropriately to deal with hypersurfaces of general causal character. This upgrade of the method \cite{Mars2005} results satisfactory to deal with perturbations of general hypersurfaces, but also implements the possibility of characterizing the perturbations of everywhere timelike/spacelike or null hypersurfaces with a specification of the transverse vector other than the normal vector, or the null transverse vector respectively. We restrict our analysis to first order perturbations. 

This work is organized as follows. In Section 2 we introduce the setting and the formalism that we will use to work with embedded hypersurfaces, which is based in the works \cite{Mars1993} and \cite{Mars_constraints}. The starting point of these two  approaches is different, but what is relevant for this work is that the hypersurface $\abshyp$ is endowed with a Riemannian structure, whose fundamental building blocks are the first fundamental form $(h)$, a two covariant symmetric tensor that captures the extrinsic properties of the hypersurface $(Y)$, and a scalar $(l^{(2)})$ and a one form $(\boldsymbol{l})$ related to the rigging vector. 
In Section \ref{section_constuction} we make explicit use of the framework \cite{Mars2005} to construct the perturbations of all these geometric objects. The perturbative method relies on a construction of a one-parameter family of the involved geometric objects, in particular we end with a collection of hypersurfaces $\Sigma_\varepsilon$ diffeomorphically related among themselves. Each of them is equipped with a rigging $\vec{l}_\varepsilon$ and we encounter the problem of how to relate (or compare) these vector fields, since they are defined in different tangent spaces. Since we are working locally, we can restrict to a neighbourhood of a point in the background hypersurface and extend this vector fields to this neighbourhood. The necessity of these extensions was already touched on \cite{Mars2005}, but we revisit this question in order to discuss their existence. We provide an example of an extension of the riggings that can be built without restricting the generality of the method, which enables us to use safely Lemma \ref{lemma_marc_perturbations} to produce expressions for the perturbations of $\{h, l^{(2)}, \boldsymbol{l}, Y\}$ which are given in Section \ref{section_perturbations_data} in terms of the \textit{allowed ingredients}: the background elements that define $\abshyp$ listed above, the metric perturbations $g_1$ and a vector field $\vec{Z}$ that describes the deformation of the hypersurface. There is also a vector field -- that we denote by $\vec{\zeta}$ -- that comes up from the direct application of the method to calculate the perturbations of the hypersurface data. This extra ingredient in the method reflects the freedom in the choice of rigging (see point (ii) in the next paragraph) to first order.

The first order perturbations are not invariant under the distinct sources of freedom in the method. In Section \ref{section_freedom} we discuss how the perturbations transform under spacetime and hypersurface gauge transformations. These were already discussed in \cite{Mars2005} and the results therein remain valid for general hypersurfaces. The dependence of the method in the rigging vectors $\vec{l}_\varepsilon$ adds two sources of freedom that we examine, (i) the possibility of taking different extensions of $\vec{l}_\varepsilon$ and (ii) transformations of $\vec{l}_\varepsilon$ into other transverse vectors $\vec{l}_\varepsilon{}'$. We conclude that the perturbations do not depend on the extensions, and we characterize the change of the perturbations under rigging transformations in Lemma \ref{lemma_rigging_transformations_fo}. We show that the perturbations of $\{\boldsymbol{l}, l^{(2)} \}$ are in direct correspondence with the gauge objects that characterize the transformations (ii) to first order.

Finally, we devote Section 6 to apply the results to a matching situation. We consider the proper matching conditions to first order established in \cite{Mars2005} and a series of results that hold in the exact case, such as the uniqueness of the rigging and the independence of the junction condition regarding the extrinsic curvature on the choice of rigging.

\section{General hypersurfaces}
\label{section_introduction_hypersurfaces}
We consider an ambient spacetime, which is a $\dimension+1$ dimensional Lorentzian manifold $\mathcal M$ endowed with a Lorentzian metric $g$. The metric is assumed to be at least $C^2$ everywhere,
 unless otherwise stated.  We use index notation, and in particular spacetime objects will carry Greek indices, which run from $0$ to $\dimension$.  We denote the Levi-Civita covariant derivative by $\nabla$ and the Riemann tensor follows the convention \cite{Wald}
\begin{equation}
\nabla_\alpha \nabla_\beta \omega_\mu- \nabla_\beta \nabla_\alpha \omega_\mu = R_{\alpha \beta \mu}^{\phantom{\alpha \beta \mu}\nu} \omega_\nu,
\end{equation}
 for any one form $\boldsymbol{\omega}$. We also follow the conventions of \cite{Wald} for symmetrization/antisymmetrization of tensor fields.

Let $\abshyp$ be an abstract hypersurface, i.e. a $\dimension$-dimensional manifold detached from the spacetime. Objects related to this manifold will carry Latin indices, whose range is $\{1,\ldots,\dimension\}$. We introduce local coordinates $\{x^\alpha \}$ in $\mathcal M$ and $\{y^a\}$ in $\abshyp$. This abstract hypersurface is embedded in the ambient space $(\mathcal M, g)$ via the $C^3$ mapping $\Phi: \abshyp \rightarrow \Sigma_0  \subset  \mathcal M$. The tangent space $T(\Sigma_0)$ admits a direct sum decomposition in terms of a $m$-dimensional subspace of vectors tangential to $\Sigma_0$ and a one dimensional subspace of vectors transverse to $\Sigma_0$.

A basis of the tangential subspace can be constructed in terms of the 
image of the natural basis $\partial_{y^a}$ at the tangent spaces $T_p \abshyp$, which  through the differential map $d\Phi$ defines the set of vectors 
\begin{equation}
\vec{e}_a := d\Phi\left(\frac{\partial}{\partial y^a}\right) = \frac{\partial \Phi^\alpha}{\partial y^a} \left.\frac{\partial}{\partial x^\alpha}\right|_{\Sigma_0} = e^\alpha_a \left.\frac{\partial}{\partial x^\alpha}\right|_{\Sigma_0}.%, \qquad \mathbf{n} (\vec{e}_a) = 0.
\end{equation}
Note that the set $\{\vec{e}_a\}$ are $\dimension$-independent spacetime vectors defined at $T(\Sigma_0)$. 
We define the normal one form as the unique one form $\boldsymbol{n}$, up to scaling, that satisfies $\boldsymbol{n} (\vec{e}_a) = 0$ for all $a$. The normal one form induces the normal vector via the metric isomorphism. The hypersurface $\Sigma_0$ can be classified attending to its causal character as follows

\begin{definition}
	The hypersurface $\Sigma_0$ is:
	\begin{itemize}
		\item Timelike at a point $p \in \Sigma_0$ if $\boldsymbol{n}(\vec{n})|_p > 0$.
		\item Spacelike at a point $p \in \Sigma_0$ if $\boldsymbol{n}(\vec{n})|_p < 0$.
		\item Null at a point $p \in \Sigma_0$ if $\boldsymbol{n}(\vec{n})|_p = 0$.
	\end{itemize}
\end{definition}

Since the causal character of the hypersurfaces is a property defined pointwise, it could vary from one point to another. In this case we refer to the hypersurface as a \textit{general hypersurface}. If, on the contrary, the causal character remains constant in the hypersurface, then these are called null, timelike or spacelike. We will refer to the last two options as \textit{standard hypersurfaces}.

The transverse subspace of $T(\Sigma_0)$ is spanned by a transverse vector called the rigging $\vec{l}$ \cite{Mars1993}, where transverse means that $\boldsymbol{n}(\vec{l}) \neq 0$ at every point of $\Sigma_0$. It is clear from this definition of transversality that the rigging is highly non-unique. For instance in standard hypersurfaces, the normal vector $\vec{n}$ is transverse everywhere and it is a very convenient choice as a rigging. However, if there are null points, the normal vector becomes tangent there and ceases to complete the basis at $T(\Sigma_0)$. The existence of the rigging is ensured by the following lemma.
\begin{lemma}{\bf (Mars 2013 \cite{Mars_constraints})} \label{lemma_existence_rigging}
	Let $\Sigma_0$ be a hypersurface in $\mathcal M$. A rigging $\vec{l}$ exists if and only if $\Sigma_0$ is orientable.
\end{lemma}

In the rest of the manuscript we adopt the convention $\boldsymbol{n} (\vec{l}) = 1$ for the transversality condition.

Finally, we construct the dual basis with the normal one form and the one forms $\boldsymbol{\omega}^a$ uniquely defined by 
\begin{equation}
\omega^a_\alpha e_b^\alpha = \delta_b^a, \qquad \boldsymbol{\omega}^a(\vec{l}) = 0.
\end{equation}
The dual basis depends on the rigging, which is not unique. We review this freedom in the method at the end of this section.
The decomposition of the unit tensor $\delta_\alpha^\beta$ in this basis results in the following expression
\begin{equation}
e_a^\alpha \omega^a_\beta + l^\alpha n_\beta = \delta_\beta^\alpha. \label{basis_completeness}
\end{equation}
In the particular case of a timelike or spacelike hypersurface, the rigging can be fixed by $\vec l = \norm \vec n$ with $\norm \equiv \boldsymbol{n}(\vec{n})$, and this relation is simply $e_a^\alpha \omega^a_\beta = \delta_\beta^\alpha - \norm n^\alpha n_\beta$.

The vectors $\{e_a^\alpha\}$ and the covectors $\{\omega^a_\alpha\}$ applied to tensor fields act as the distinct differential mappings associated to $\Phi$. Let $A$ and $\mathcal A$ be tensor fields of the appropriate rank defined at points of $\Sigma_0$ in the ambient spacetime and in the abstract hypersurface respectively. The pullback of spacetime covariant objects to $\abshyp$ reads
\begin{equation*}
\mathcal A_{{a_1}\dots {a_r}} := \Phi^* ( A_{{\alpha_1}\cdots {\alpha_r}}) =   e^{\alpha_1}_{a_1} \dots e^{\alpha_r}_{a_r}  A_{{\alpha_1}\cdots {\alpha_r}},
\end{equation*}
and any contravariant tensor in $\abshyp$ can be promoted to the spacetime similarly
\begin{equation*}
A^{{\alpha_1}\cdots {\alpha_r}} := d\Phi (\mathcal A^{{a_1}\dots {a_r}}) =  e^{\alpha_1}_{a_1} \dots e^{\alpha_r}_{a_r} \mathcal A^{{a_1}\dots {a_r}}.
\end{equation*}
The one forms $\{\boldsymbol{\omega}^a\}$ projects contravariant tensors to $\abshyp$ 
\begin{equation*}
\mathcal A^{{a_1}\dots {a_r}} := d\Phi^{-1} ( A^{{\alpha_1}\cdots {\alpha_r}}) =\omega_{\alpha_1}^{a_1} \dots \omega_{\alpha_r}^{a_r} A^{{\alpha_1}\cdots {\alpha_r}},
\end{equation*}
and promotes covariant tensors from $\abshyp$ to the ambient space at points of $\Sigma_0$
\begin{equation*}
A_{\alpha_1 \ldots \alpha_r} := \Phi^{-1}{}^* (\mathcal A_{a_1 \ldots a_r}) = \omega_{\alpha_1}^{a_1} \ldots \omega_{\alpha_r}^{a_r} \mathcal A_{a_1 \ldots a_r}.
\end{equation*}
We will often use the decomposition of a  vector field $A^\alpha$ in the tangent basis $ A^\alpha \equiv \mathcal A l^\alpha +  {\mathcal A}^a  e_a^\alpha$. This defines the following fields in $\abshyp$: $\mathcal A \equiv n_\alpha  A^\alpha$ and $ {\mathcal A}^a \equiv  A^\alpha \omega^a_\alpha$. If the vector field has no rigged component ($\mathcal A=0$) we will often use the notation $ \vec A_\abshyp$ to denote its counterpart in $\abshyp$, i.e. $\vec{A} = d \Phi (\vec{A}_\abshyp)$, which is equivalent to $A^\alpha = \mathcal A^a e_a^\alpha$, or abusing notation $A^\alpha =  A^a e_a^\alpha$. We remark that whenever we use this notation, the subindex $\abshyp$ in the vector field $\vec{A}_\abshyp$ is a simple reminder that the vector field is defined on $T(\abshyp)$ and it should not be confused with a spacetime index.

There are some relevant objects defined through projections in the bases $\{\vec{l}, \vec{e}_a\}$ or $\{\boldsymbol{n}, \boldsymbol{\omega}^a\}$, such as the component  of the normal vector in $T(\abshyp)$: $n^a \equiv n^\alpha \omega_\alpha^a$, or similarly for the rigging one form: $l_a \equiv l_\alpha e_a^\alpha$, and also the symmetric tensors
\begin{eqnarray*}
h_{ab} = e_a^\alpha e_b^\beta \left. g_{\alpha \beta} \right|_{\Sigma_0}, \quad P^{ab} \equiv \omega^a_\alpha \omega^b_\beta \left. g^{\alpha \beta}\right|_{\Sigma_0}.
\end{eqnarray*}
The tensor $h$ is known as the first fundamental form of $\abshyp$. Its signature is not fixed, and it is an induced metric on $\abshyp$ for standard hypersurfaces, but at null points it becomes degenerate \cite{Mars1993}.  In fact, the causal character is a property that we have defined for points that belong to the hypersurface $\Sigma_0$ in $\mathcal M$, but given the embedding $\Phi_0$, which identifies points and tangent spaces of $\Sigma_0$ with those of $\abshyp$, this property is also attached to the abstract hypersurface $\abshyp$ pointwise.  The procedure depends on the embedding, but there are no ambiguities as long as we keep it unchanged.

We consider the derivatives of the normal one form and rigging decomposed in the tangent/dual basis at $\Sigma_0$ 
\begin{eqnarray}
\nabla_{\vec{e}_a}\boldsymbol{n} = - \varphi_a \boldsymbol n + \kappa_{ab} \boldsymbol{\omega}^b, \quad 
\nabla_{\vec{e}_a}\vec{l} = \varphi_a \vec{l} + \Psi_a^b \vec{e}_b,\label{nabla_normal_rigging}
\end{eqnarray}
which define the following objects \cite{Mars1993}
\begin{equation}
\kappa_{ab} := e_a^\alpha e_b^\beta \nabla_\alpha n_\beta,\quad \varphi_a := n_\mu e_a^\nu\nabla_\nu l^\mu, \quad \Psi_b^a := \omega^a_\mu e^\nu_b \nabla_\nu l^\mu. \label{definition_objects}
\end{equation}
The two covariant tensor $\kappa_{ab}$ is the second fundamental form of $\abshyp$, and when the hypersurface is timelike or spacelike everywhere it captures its extrinsic properties.

It is possible to define a torsion free covariant derivative in $\abshyp$ exploiting the rigged structure built at points of $\Sigma_0$. The rigged connection  $\overline{\nabla}$ is constructed as follows: consider any pair of vectors $\vec{X}_\abshyp$ and $\vec{Y}_\abshyp$ from $\abshyp$. These induce vectors in the spacetime via the mapping $d\Phi$, that we denote by $\vec{X}$ and $\vec{Y}$ respectively. The covariant derivative of one along the other is a vector that will have, in general, a rigged component and a component completely tangent to $\Sigma_0$. The latter defines a covariant operator on $\abshyp$
\begin{equation*}
\overline{\nabla}_{\vec{X}_\abshyp} \vec{Y}_\abshyp := d\Phi^{-1} (\nabla_{\vec{X}} \vec{Y}), \text{ or  } X^b \overline \nabla_b Y^a =  \omega^a_\alpha X^\mu \nabla_{\mu} Y^\alpha \text{ using index notation}.
\end{equation*}
Moreover the relation between the covariant derivatives of the ambient space and the embedded hypersurface is ruled by  \cite{Mars1993}
\begin{eqnarray}
\omega_{\mu_1}^{a_1}\cdot \cdot \cdot \omega_{\mu_r}^{a_r} e^{\nu_1}_{b_1} \cdot \cdot \cdot e^{\nu_q}_{b_q} e_c^\gamma \nabla_\gamma  A^{\mu_1 \cdot \cdot \cdot \mu_r}_{\nu_1 \cdot \cdot \cdot \nu_q} &=& \overline \nabla_c \mathcal A^{a_1 \cdot \cdot \cdot a_r}_{b_1 \cdot \cdot \cdot b_q} + \sum_{i=1}^r \mathcal A^{a_1 \cdot \cdot \cdot a_{i-1} \gamma a_{i+1} \cdot \cdot \cdot a_r}_{b_1 \cdot \cdot \cdot b_q} n_\gamma \Psi_c^{a_i} \nonumber\\
&&+\sum_{j=1}^q \mathcal A^{a_1 \cdot \cdot \cdot a_r}_{b_1 \cdot \cdot \cdot b_{j-1}\gamma b_{j+1}\cdot \cdot \cdot b_q} l^\gamma \kappa_{c b_j}, \label{nabla_nablabar}
\end{eqnarray}
where $A$ is a spacetime tensor, $\mathcal A \equiv \Phi^* A$ and the tensor with mixed indices represents the full projection of $A$ to $\abshyp$ except for the index $\gamma$ which is being contracted with $\vec{l}$ or $\boldsymbol{n}$. For instance $\mathcal A^{a_1 \cdot \cdot \cdot a_{i-1} \gamma a_{i+1} \cdot \cdot \cdot a_r}_{b_1 \cdot \cdot \cdot b_q} n_\gamma \equiv \omega_{\mu_1}^{a_1}\cdot \cdot \cdot \omega_{\mu_{i-1}}^{a_{i-1}} n_\gamma \omega_{\mu_{i+1}}^{a_{i+1}}\cdot \cdot \cdot \omega_{\mu_r}^{a_r} e^{\nu_1}_{b_1} \cdot \cdot \cdot e^{\nu_q}_{b_q}   A^{\mu_1 \cdot \cdot \cdot \mu_{i-1}\gamma\mu_{i+1} \cdot \cdot \cdot \mu_r}_{\nu_1 \cdot \cdot \cdot \nu_q}$, and similarly for the contraction of the covariant indices with the rigging vector.

As shown above, this formalism exploits the fact that $\Sigma_0$ is an embedded hypersurface in order to define the relevant geometrical objects in $\abshyp$ in terms of spacetime objects. This strategy has been proven to succeed in order to describe general hypersurfaces and the corresponding junction conditions \cite{Mars1993}, but there are some aspects regarding the perturbations which are better understood from an alternative point of view. The \textit{hypersurface data} approach, introduced in \cite{Mars_constraints}, consist of considering a set of geometric data on $\abshyp$ that characterizes the abstract manifold completely, as well as a series of properties and relations that this data satisfies. The main advantage of this approach is that it is independent of the spacetime, and in fact an embedding is not required \textit{a priori}. Obviously when the data is embedded both approaches are equivalent.

The main ingredients for this approach are the so called \textit{hypersurface metric data}, which is a set consisting on a $\dimension$ dimensional smooth manifold, a symmetric two covariant tensor, a one form and a scalar: $\{\abshyp, h_{ab}, l_a, l^{(2)}\}$ so that the matrix  \footnote{ The \textit{hypersurface metric data} in matrix form reads \cite{Mars_constraints} $	\begin{bmatrix}
l^{(2)}    & l_a \\
l_b & h_{ab}
\end{bmatrix}
$.} formed by these has a Lorentzian signature everywhere on $\abshyp$.  When the data is embedded in a spacetime, these objects become
\begin{equation}
\Phi(\abshyp) = \Sigma_0, \quad h = \Phi^*(g), \quad \boldsymbol{l} = \Phi^* g(\vec{l}, \cdot), \quad  l^{(2)} = \Phi^* g(\vec{l}, \vec{l}),
\end{equation}
which are the relations defining the left hand sides in the spacetime approach. From this hypersurface metric data the objects $\{P^{ab}, n^a, n^{(2)}\}$ are defined as the solution of the equations
\begin{eqnarray}
&&P^{ac}h_{bc} + n^a l_b =\delta_b^a, \quad P^{ab}l_b + l^{(2)}n^a =0,\nonumber\\
&&n^al_a + n^{(2)}l^{(2)}= 1, \quad h_{ab}n^b + n^{(2)}l_a = 0. \label{constraints_definitions}
\end{eqnarray}
When the data is embedded these objects agree with the definitions provided in the \textit{spacetime approach}, i.e. $P^{ab} = \omega^a_\alpha \omega^b_\beta g^{\alpha \beta}$, $ n^a = n^\alpha \omega^a_\alpha$, and $n^{(2)} = g^{\alpha \beta}n_\alpha n_\beta$.

The \textit{hypersurface data} is the set comprising the \textit{hypersurface metric data} plus an additional two covariant symmetric tensor $Y_{ab}$, which for embedded  data is
\begin{equation}
Y = \frac{1}{2}\Phi^* (\LL_{\vec{l}}g).
\end{equation}
We will refer to this object as the \textit{rigged fundamental form}, and it captures the extrinsic properties of a general hypersurface. In fact it is possible to define objects (\ref{definition_objects}) solely in terms of the hypersurface metric data and the tensor $Y$  (see Proposition 1 in \cite{Mars2005}).

One of the main results from this second intrinsic approach that we will use is summarized in the following lemma.
\begin{lemma} {\bf (Mars 2013 \cite{Mars_constraints})} \label{lemma_vector_reconstruction}
	Let $Z_a$ and $W$ be given. There exists a vector $V^a$ satisfying $V^al_a = W$ and $h_{ab}V^b=Z_a$ if and only if
	\begin{equation}
	n^b Z_b + n^{(2)}W=0. \label{constraint_lemma_reconstruction}
	\end{equation}
	Moreover, the solution is unique and reads
	\begin{equation*}
	V^a = P^{ab}Z_b + n^a W.
	\end{equation*}
	
\end{lemma}

Finally, we recall that these methods to characterise an embedded hypersurface depend on a choice of transverse vector, which is highly non-unique. Any two different riggings $\vec{l}$ and $\vec{l}{}'$ are related by
\begin{equation*}
\vec{l}{}' = \lambda (\vec{l} + \vec{v}),
\end{equation*}
so that the gauge freedom is encoded in a function $\lambda$ on $\Sigma_0$ ($\lambda \neq 0$) and in a tangent vector on $\Sigma_0$, $\vec v$.  As a result of this transformation the tangent basis and the cobasis at points of $\Sigma_0$ transform as
\begin{eqnarray*}
e_a'{}^\alpha = e_a^\alpha, \quad n'_\alpha = \lambda^{-1} n_\alpha, \quad \omega'^a_{\alpha} = \omega^a_\alpha -v^a n_\alpha.
\end{eqnarray*}
Let us denote the corresponding vector in $\abshyp$ by $\overline{v}$, so that $\vec v = d \Phi (\overline{v})$. Hence the data transforms as follows
\begin{eqnarray}
h'_{ab} &=& h_{ab},\label{rigging_transformation_1ff}\\
l'{}^{(2)} &=& \lambda^2 (l^{(2)} + 2\overline{v}^a l_a + \overline{v}^{(2)}),\label{rigging_transformation_l2}\\
l'_a &=& \lambda (l_a + h_{ab} \overline{v}^b),\label{rigging_transformation_la}\\
Y'_{ab} &=& \lambda Y_{ab} + \frac{1}{2} (\lambda_{,a} l_b + \lambda_{,b} l_a) + \frac{1}{2} \mathcal L_{\lambda \overline{v}} h, \label{transformation_Y}
\end{eqnarray}
whereas the objects $\{n^{(2)}, n^a, P^{ab} \}$ transform according to
\begin{eqnarray}
n^{(2)}{}' = \frac{n^{(2)}}{\lambda^2}, \quad n^a{}' = \frac{1}{\lambda}(n^a - n^{(2)} \overline{v}^a) , \quad P^{ab}{}' = P^{ab} + n^{(2)} \overline{v}^a \overline{v}^b -2 \overline{v}^{(a} n^{b)}. \label{inverse_data_transformations}
\end{eqnarray}
Part of the freedom in choosing the rigging vector is encoded in the vector $\vec{v}$. Since it is tangent to $\Sigma_0$, we write it as $\vec{v} = v^a \vec{e}_a$ and its associated one form reads
\begin{equation*}
\boldsymbol{v} = W \boldsymbol{n} + Z_a \boldsymbol{\omega}^a, \quad W \equiv l_a v^a, \quad Z_a \equiv h_{ab} v^b.
\end{equation*}
In fact, using the identity on $\abshyp$ we find that it is possible to decompose $v^a$ as follows
\begin{equation*}
v^a = v^c \delta_c^a = (n^a l_c + P^{ab} h_{bc})v^c = n^a (l_c v^c) + P^{ab} (h_{bc}v^c) = W n^a + Z_b P^{ab}.
\end{equation*}
An equivalent point of view would be to encode the rigging transformations in the two fields $W$ and $Z_a$ that replace $\overline{v}$, which according to Lemma \ref{lemma_vector_reconstruction} satisfy the constraint $n^{(2)}W + n^a Z_a = 0$.

In terms of these objects, the change of the hypersurface metric data under a rigging transformation is rendered as
\begin{eqnarray}
l_a' &=& \lambda (l_a + Z_a), \label{rigging_transformed_la}\\
l^{(2)}{}' &=& \lambda^2 (l^{(2)} + 2W + P^{ab}Z_a Z_b - n^{(2)}W^2),\label{rigging_transformed_l2}
\end{eqnarray}

Along the rest of the paper, and as in the previous discussion, a primed $(')$ quantity shall denote that it is referred to a (transformed) rigging $\vec{l}'$, related to $\vec{l}$ by the gauge fields $\lambda$ and $\overline{v}$.

\section{Construction of the perturbations}
\label{section_constuction}

The method of perturbation of spacetimes requires a one-parameter family of spacetimes $(\mathcal M_\varepsilon, \hat{g}_\varepsilon)$, where the element corresponding to $\varepsilon =0$ is singled out as the background spacetime $(\mathcal M_0, \hat g_0)$, or simply $(\mathcal M, g)$. The spacetimes inside the family are identified through the smooth diffeomorphism 
\begin{equation*}
\psi_\varepsilon: \mathcal M \equiv \mathcal M_0 \rightarrow \mathcal M_\varepsilon,
\end{equation*}
where
$\psi_0$ is the identity map.
This mapping defines a one-parameter family of two covariant symmetric tensors $g_\varepsilon$ and the so called metric perturbations $g_1$ by
\begin{equation*}
g_\varepsilon := \psi_\varepsilon^* \hat{g}_\varepsilon, \quad g_1 := \left.\frac{d}{d \varepsilon} g_\varepsilon \right|_{\varepsilon = 0}.
\end{equation*}
We assume that the Einstein equations \footnote{$\chi$ is the gravitational coupling constant.} $\hat{G}_\varepsilon (\hat{g}_\varepsilon) = \chi \hat{T}_\varepsilon$ are satisfied in each of the $(\mathcal M_\varepsilon, \hat{g}_\varepsilon)$ for a given energy momentum tensor $\hat{T}_\varepsilon$. These can be pulled back to $(\mathcal M, g)$ and result into the following equations for the metric perturbation tensor
\begin{eqnarray*}
 -\nabla_\alpha \nabla_\beta {g_1}_\rho^\rho - \square {g_1}_{\alpha \beta} + 2 \nabla_\rho \nabla_{(\alpha}{g_1}_{\beta)}^\rho - {g_1}_{\alpha \beta } R + g_{\alpha \beta} ({g_1}^{\mu \nu} R_{\mu \nu} + \square {g_1}_\rho^\rho - \nabla_\mu \nabla_\nu {g_1}^{\mu \nu}) = 2\chi {T_1}_{\alpha \beta},
\end{eqnarray*}
where $T_1$ is the linearised energy momentum tensor.

From each member $\mathcal M_\varepsilon$ we single out a (orientable) hypersurface $\hat{\Sigma}_\varepsilon$, and require that these are diffeomorphic among themselves, so that these can be identified to an abstract manifold $\abshyp$ through the map $\phi_\varepsilon: \abshyp \rightarrow \hat{\Sigma}_\varepsilon$. This identification is non unique, and it entails another freedom in the method, regarded as the hypersurface gauge freedom: a $\varepsilon$-dependent diffeomorphism $\chi_\varepsilon: \abshyp \rightarrow \abshyp$ previous to the identification among the $\hat{\Sigma}_\varepsilon$ generates a new diffeomorphim $\phi_\varepsilon' = \phi_\varepsilon \circ \chi_\varepsilon: \abshyp \rightarrow \hat{\Sigma}_\varepsilon$.

The combination of the spacetime and hypersurface identifications generates the embeddings $\Phi_\varepsilon \equiv \psi_\varepsilon^{-1} \circ \phi_\varepsilon: \abshyp \rightarrow \Sigma_\varepsilon \subset \mathcal M$. From this point of view, the one-parameter family of hypersurfaces $\Sigma_\varepsilon$ generates the perturbations of the hypersurface $\Sigma_0$. In fact,  the information about the deformation of the hypersurfaces, relative to the gauges used, is encoded in the deformation vector $\vec Z \equiv \partial_\varepsilon \Phi_\varepsilon|_{\varepsilon = 0}$. In general it has a rigged and tangent parts to $\Sigma_0$, encoded in the scalar $Q$ and the vector field $\vec{T}$ respectively, so that $\vec Z = Q \vec l + \vec T$. Note that $Q$ and $\vec{T}$ will depend on the choice of rigging. The embeddings $\Phi_\varepsilon$ define the one-parameter two covariant fields $h_\varepsilon := \Phi_\varepsilon^*(g_\varepsilon)$ in $\abshyp$, i.e. the $\varepsilon-$family of first fundamental forms, whose perturbations are defined by
\begin{equation*}
\delta h_{ab} := \left.\frac{d h_\varepsilon}{d \varepsilon}\right|_{\varepsilon=0}.
\end{equation*}

In order to complete the construction of the one-parameter \textit{hypersurface data}, we need to describe the behaviour of the rigging vectors.

 Our starting point is the spacetimes $\{\mathcal M_\varepsilon, \hat{g}_\varepsilon, \hat{\Sigma}_\varepsilon\}$. We will assume that the $\Sup$ are orientable, so that there exists a family of normal one forms $\nup$. Thus Lemma \ref{lemma_existence_rigging} ensures that for each  hypersurface $\Sup$, there exists a rigging $\lup$, and as in the exact case we impose the normalisation condition $\nup(\lup) = 1$. 
 We use the pushforward of the inverse of $\psi_\varepsilon$ to generate the riggings $\vec{l}_\varepsilon$ in $\Sigma_\varepsilon$ as stated in the following lemma.
\begin{lemma}
	The one forms $\ndown := \psi_\varepsilon^*(\nup)$ and the vectors $\vec{l}_\varepsilon := d\psi_\varepsilon^{-1} (\vec{\hat{l}}_\varepsilon)$ are normal one forms and riggings with respect to $\Sigma_\varepsilon$.
\end{lemma}
\proof
Assume that the action of the normal one forms on a vector field $\hat{X}_\varepsilon$ defined in points of $\Sup$ is a constant with respect to $\varepsilon$, i.e. $\nup (\hat{X}_\varepsilon) = C$. Using the differential maps associated to $\psi_\varepsilon$ we can write the previous operation involving objects in $(\mathcal M, g)$
\begin{equation*}
(\nup(\hat{X}_\varepsilon)) = C \Rightarrow \nup (d \psi_\varepsilon d\psi_\varepsilon^{-1}\hat{X}_\varepsilon) =  \psi_\varepsilon^ * \nup (d\psi_\varepsilon^{-1}\hat{X}_\varepsilon) = C  \Rightarrow \ndown (X_\varepsilon) = C,
\end{equation*} 
where  $X_\varepsilon$ is defined so that $\hat{X}_\varepsilon \equiv d\psi_\varepsilon (X_\varepsilon)$.
If we let $\hat{X}_\varepsilon$ be any vector field tangent to $\Sup$ then $C=0$, and since $\hat{\Sigma}_\varepsilon$ is the image of $\Sigma_\varepsilon$ through the diffeomorphism $\psi_\varepsilon$, tangent vectors of $\Sigma_\varepsilon$ are mapped to tangent vectors of $\hat{\Sigma}_\varepsilon$ and therefore $X_\varepsilon$ is a vector field tangent to $\Sdown$. Hence, we have that $\ndown$ is a family of normal one forms for $\Sdown$.

On the contrary, let $\hat X_\varepsilon = \lup$, normalised so that $C=1$. This same argument shows that $\ndown (\ldown) = 1$, hence $\vec{l}_\varepsilon$ is a rigging for $\Sdown$. 

\hfill $\blacksquare$

\begin{figure}[h!]
	\centering
	%%%%%%%%%%%%%%%%%%%%%%%%%%%%%%
	%%%%%%%%%%%%%%%%%%%%%%%%%%%%%%
	\includegraphics[ %scale=0.40
	width=0.8\textwidth]{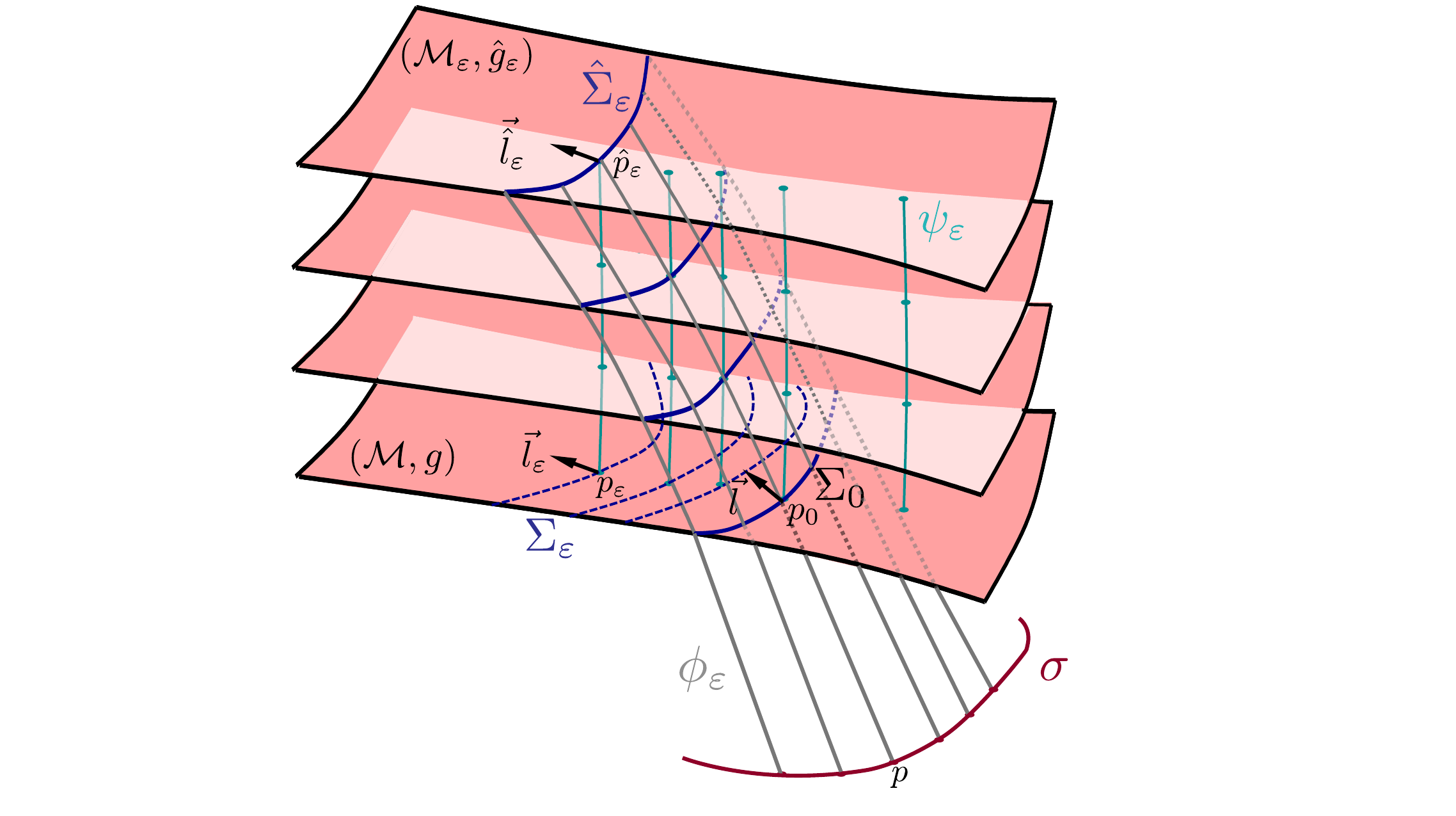}
	%%%%%%%%%%%%%%%%%%%%%%%%%%%%%%%%%
	%%%%%%%%%%%%%%%%%%%%%%%%%%%%%%%%%
	\caption{The geometric approach for the perturbation of hypersurfaces. There is a one parametric family of spacetimes $\{\mathcal{M}_\varepsilon, \hat{g}_\varepsilon \}$, diffeomorphically related through $\psi_\varepsilon$. The hypersurfaces $\hat{\Sigma}_\varepsilon$ are identified with the abstract hypersurface $\abshyp$ via $\phi_\varepsilon$. The riggings $\vec{\hat l}_\varepsilon$ on $\hat{\Sigma}_\varepsilon$ induce the vectors $\vec{l}_\varepsilon$ on $\Sigma_\varepsilon$ through the differential map $d\psi_\varepsilon$. }
	
\label{figure:per_picture_3}
\end{figure}

Having constructed these riggings, we introduce the fields $\boldsymbol{l}_\varepsilon := \Phi_\varepsilon^* g_\varepsilon (\vec{l}_\varepsilon, \cdot )$ and $l^{(2)}_\varepsilon := \Phi_\varepsilon^* g_\varepsilon (\vec{l}_\varepsilon, \vec{l}_\varepsilon )$ and $Y_\varepsilon := (1/2) \Phi_\varepsilon^ * \mathcal L_{\vec{l}_\varepsilon} g_\varepsilon$, which define the perturbations of the remaining objects in the data set
\begin{eqnarray*}
&&\delta \boldsymbol{l} := \left.\frac{d \boldsymbol{l}_\varepsilon}{d \varepsilon}\right|_{\varepsilon=0} =  \left.\frac{d }{d \varepsilon}\Phi_\varepsilon^* g_\varepsilon (\vec{l}_\varepsilon, \cdot )\right|_{\varepsilon=0}, \quad \delta l^{(2)} := \left.\frac{d l^{(2)}_\varepsilon}{d \varepsilon}\right|_{\varepsilon=0} = \left.\frac{d}{d \varepsilon}\Phi_\varepsilon^* g_\varepsilon (\vec{l}_\varepsilon, \vec{l}_\varepsilon )\right|_{\varepsilon=0}, \nonumber\\
&& \delta Y :=  \left.\frac{d Y_\varepsilon}{d \varepsilon}\right|_{\varepsilon=0} = \frac{1}{2} \frac{d}{d \varepsilon} \left.\Phi_\varepsilon^* \LL_{\vec{l}_\varepsilon} g_\varepsilon \right|_{\varepsilon = 0}.
\end{eqnarray*}

Now that we have defined the perturbations, the question is how to compute them in terms of allowed ingredients. By these we refer to (i) the background objects that generate the embedded \textit{hypersurface data} (ii) the perturbations of the metric (iii) the deformation vectors of the hypersurface, and also a perturbation of the rigging vector. The framework \cite{Mars2005} is adequate for this purpose, and the basic result that we will use is a lemma that allows the computation of $\varepsilon$ derivatives of covariant objects in $\abshyp$ in terms of $\varepsilon$ derivatives of spacetime objects.

\begin{lemma}{(adapted from Mars 2005 \cite{Mars2005})}
	Let $A_\varepsilon$ be a $C^2$ one-parameter family of covariant tensor fields on $\mathcal M$, $\Phi_\varepsilon: \abshyp \rightarrow \Sigma_\varepsilon \subset \mathcal M$ a $C^3$ family of embeddings and define $\mathcal{A}_\varepsilon := \Phi_\varepsilon^*(A_\varepsilon)$. Then
	\begin{equation}
	\left. \frac{d \mathcal A_\varepsilon}{d\varepsilon}\right|_{\varepsilon = 0} = \Phi_0^*\left(\mathcal{L}_{\vec{Z}}A_0 \right) + \Phi_0^* \left(\lim_{\varepsilon\rightarrow 0} \frac{d A_\varepsilon}{d \varepsilon}\right), \label{lemma_formula}
	\end{equation}
	where
	\begin{equation*}
	\vec{Z} = \left. \frac{d \Phi_\varepsilon}{d\varepsilon}\right|_{\varepsilon=0}.
	\end{equation*}
	\label{lemma_marc_perturbations}
\end{lemma}
The proof of this lemma highlights an important subtlety which is of considerable significance to the present paper. The proof proceeds by calculating as follows:
\begin{eqnarray}
\left.\frac{d{\mathcal{A}}_\varepsilon}{d\varepsilon}\right|_{\varepsilon=0} &=& \lim_{\varepsilon\to 0}
\frac{1}{\varepsilon}\left(\Phi_\varepsilon^*(A_\varepsilon)-\Phi_0^*(A_0)\right)\nonumber \\
&=& \lim_{\varepsilon\to 0}
\frac{1}{\varepsilon}\left(\Phi_\varepsilon^*(A_\varepsilon)-\Phi_0^*(A_\varepsilon)\right)+\lim_{\varepsilon\to 0}
\frac{1}{\varepsilon}\left(\Phi_0^*(A_\varepsilon)-\Phi_0^*(A_0)\right).\label{lemma_calculation}
\end{eqnarray}
With a little more work (see \cite{Mars2005}), it can be shown that the first (respectively, second) term on the right hand side of (\ref{lemma_calculation}) yields the first (respectively, second) term on the right hand side of (\ref{lemma_formula}). We flag two crucial points that arise in this calculation. First, the term that is added and subtracted (to ultimately yield the two terms of (\ref{lemma_formula})) involves the pull-back from $\Sigma_0$ of the quantity $A_\varepsilon$. In Lemma \ref{lemma_marc_perturbations}, no problem arises, as $A_\varepsilon$ is defined throughout ${\mathcal{M}}$, and so, in particular, is defined on  $\Sigma_0$. However, in the applications below, we will need to apply Lemma \ref{lemma_marc_perturbations} to quantities which are \textit{ab initio} defined only on each $\Sigma_\varepsilon$. Indeed this arises in \cite{Mars2005}, where the lemma is applied to the one-parameter family of second fundamental forms $\kappa_\varepsilon$. Below, we will apply the lemma to calculate the perturbation of the hypersurface data. These are associated with one-parameter families of geometric quantities (${\boldsymbol{l}}_\varepsilon, l_\varepsilon^{(2)}, Y_\varepsilon$) defined on each $\Sigma_\varepsilon$, as opposed to (for example) the one-parameter family of metrics $g_\varepsilon$, which is defined throughout ${\mathcal{M}}$ for each $\varepsilon$. This presents a fundamental problem with the application of the lemma, as the pull-back of e.g.\ $Y_\varepsilon$ from $\Sigma_0$ is not defined. 
We overcome this problem by constructing \textit{extensions} of such quantities to appropriate neighbourhoods of of $\Sigma_0$. Most importantly, we describe the extension of the one-parameter family of riggings $\vec{l}_\varepsilon$ to appropriate neighbourhoods of $\Sigma_0$: see Definition \ref{smooth-ext} and Proposition \ref{proposition_extensions_existence} below. 

The second issue that arises in the calculation that yields (\ref{lemma_formula}) from (\ref{lemma_calculation}) relates to the existence of the limit of the second term on the right hand side of (\ref{lemma_calculation}). This requires a sufficient degree of smoothness of the mapping $\varepsilon\mapsto A_\varepsilon$ at $\varepsilon=0$. In the case of quantities  defined throughout ${\cal{M}}$ (e.g.\ the one-parameter family of metrics, $g_\varepsilon$), this smoothness is an essential part of the definition: without this degree of differentiability, there is no sense of a perturbed metric. However, when calculating perturbations below, we will be working with extensions. Thus our definition of these extensions must take account of the need for this limit to exist. Naturally, this has consequences for Proposition 3.2 in which the existence of appropriately defined extensions is established.

In the case of a general hypersurface, the  application of Lemma \ref{lemma_marc_perturbations} to the objects of interest include 
%The application of Lemma \ref{lemma_formula} to certain objects of interest, including 
in particular $\vec{l}_\varepsilon$, the one-parameter family of riggings, and it requires the construction of an extension of these objects to a neighbourhood ${\cal{U}}$ of $\Sigma_0$ in ${\cal{M}}$. In the rest of this section, we prove the existence of relevant extensions for $\vec{l}_\varepsilon$. We note that a similar method could be applied to construct (for example) the extension of the one-parameter family of normal 1-forms. We need to be precise about what is meant by an extension:

\begin{definition}\label{smooth-ext} Let ${\cal{U}}$ be an open subset of ${\cal{M}}$ such that $\Sigma_\varepsilon\subset {\cal{U}}$ for all $\varepsilon\in I$, where $I$ is an interval containing 0, and for each $\varepsilon \in I$, let $\vec{l}_\varepsilon$ be a rigging on $\Sigma_\varepsilon$ (giving a one-parameter family of riggings). 
	\begin{enumerate}
		\item A \textbf{monotone path} $\mu:I\to {\cal{M}}$ is a $C^1$ curve with the property that $\mu(\varepsilon)\in\Sigma_\varepsilon$ for all $\varepsilon\in I$. 
		\item The one-parameter family of riggings is said to be \textbf{smooth} if in any local coordinate system in which the metric is $C^1$, the functions $\varepsilon\mapsto l^\alpha_\varepsilon|_{\mu(\varepsilon)}$ are $C^1$ on $I$ for all monotone paths $\mu$.
		\item Let ${\cal{U}}_{p}\subset {\cal{U}}$ be a neighbourhood of $p\in\Sigma_0$. The one-parameter family of vector fields $\vec{L}(\varepsilon,x)\in T(\cupo)$, $\varepsilon\in I$ is called an \textbf{extension to $\cupo$} of the smooth one-parameter family of riggings $\vec{l}_\varepsilon \in T(\Sigma_\varepsilon), \varepsilon\in I$ if
		\begin{itemize}
			\item[(a)] $\vec{L}(\varepsilon,x)|_{p_\varepsilon}=\vec{l}_\varepsilon|_{p_\varepsilon}$ for all $\varepsilon\in I$ and $p_\varepsilon\in\Sigma_\varepsilon\cap\cupo$ and:
			\item[(b)] For each $p_0\in\Sigma_0\cap\cupo$, there exists the limit
			\begin{equation}
			\vec{l}_1|_{p_0} = \lim_{\varepsilon\to 0}\frac{1}{\varepsilon}\left(\vec{L}(\varepsilon,x)|_{p_0}-\vec{l}_0|_{p_0}\right), \label{definition_l1}
			\end{equation}
			and this limit defines a continuous vector field on $\Sigma_0\cap\cupo$. 
		\end{itemize}
	\end{enumerate}
\end{definition}

We note that an extension of a smooth one-parameter family of riggings has this property:

\begin{proposition} Let $\cupo$ be a neighbourhood of $p\in\Sigma_0$ and let $\vec{L}(\varepsilon,x)$ be an extension to $\cupo$ of the one-parameter family of riggings $\vec{l}_\varepsilon$. Then for all $p_0\in\Sigma_0\cap\cupo$, 
	\begin{equation}
	\lim_{\varepsilon\to 0} \vec{L}(\varepsilon,x)|_{p_0} = \vec{l}_0|_{p_0}.
	\end{equation}
\end{proposition}

The key outstanding question in relation to extensions is that of their existence, which is settled by the following result. 

\begin{proposition} \label{proposition_extensions_existence}
	Let $\vec{l}_\varepsilon, \varepsilon\in I$ be a smooth one-parameter family of riggings on ${\cal{U}}$. Then for every $p\in\Sigma_0$, there exists an interval $I_p\subset I$ with $0\in I_p$, a neighbourhood $\cupo\subset\cu$ of $p$ and $\vec{L}(\varepsilon,x)\in T(\cupo)$, an extension to $\cupo$ of $\vec{l}_\varepsilon, \varepsilon\in I_p$. 
\end{proposition} 

\proof
The basic idea of the proof is as follows. We consider the congruence of geodesics emanating from $\Sigma_0$, tangent to the rigging $\vec{l}_0$ at $\Sigma_0$. Along individual members of the congruence, we parallel transport $\vec{l}_\varepsilon$ from $\Sigma_\varepsilon$ to $\Sigma_0$. The union (over the geodesic congruence) of these parallel transported vectors yields the extension $\vec{L}(\varepsilon,x)$ on an appropriate neighbourhood of $p$.

Let $p\in\Sigma_0$ and let $O_p\subset \Sigma_0$ be a bounded neighbourhood of $p$ in $\Sigma_0$ with compact closure $\overline{O}_p$.

Let $p_0\in\overline{O}_p$ and let $\gamma_0$ be the unique (affinely parametrised) geodesic with initial point $p_0$ and initial tangent $\vec{l}_0|_{p_0}$. For all $\varepsilon\in I_{p_0}$, a sufficiently small interval containing 0, the geodesic $\gamma_0:I_{p_0}\to {\cal{M}}$ is a monotone path, meeting the hypersurface $\Sigma_\varepsilon$ at a point $p_\varepsilon$ given by $p_\varepsilon=\gamma_0(\tau_\varepsilon)$. By the smoothness properties of the embedding of the $\Sigma_\varepsilon$ and of the solutions of geodesic equations, the mapping $\varepsilon\mapsto \tau_\varepsilon$ is $C^1$ on $I_{p_0}$. 

For $\tau\in[0,\tau_\varepsilon]$, we define the tangent vector $\vec{L}_\varepsilon(\tau;p_0)\in T_{\gamma_0(\tau)}({\cal{M}})$ to be the unique vector parallel transported along $\gamma_0$ and satisfying 
\begin{equation} 
\vec{L}_\varepsilon(\tau_\varepsilon;p_0)=\vec{l}_\varepsilon|_{p_\varepsilon}.
\label{extension-1}
\end{equation}
For notational ease, we will write $\vec{L}_\varepsilon(\tau_\varepsilon;p_0)=\vec{L}_\varepsilon(\tau_\varepsilon)$ for the moment.  
Now let $\vec{\tetrad}_\ta, \ta=0,1,\dots,\dimension$ be an orthonormal set of vectors, parallel transported along $\gamma_0$. Then we can write
\begin{equation}
\vec{L}_\varepsilon(\tau) = L^\ta_\varepsilon \vec{\tetrad}_\ta(\tau), \tau\in[0,\tau_\varepsilon],
\end{equation}
where $L^\ta_\varepsilon$ are $\varepsilon-$dependent constants. These constants are uniquely determined by the smooth rigging: in any local coordinate system, 
\begin{eqnarray*}
	L^\alpha_\varepsilon(\tau_\varepsilon) &=& L^\ta_\varepsilon \tetrad^\alpha_\ta(\tau_\varepsilon) \\
	&=& l^\alpha_\varepsilon|_{p_\varepsilon} \\
	&=& l^\ta_\varepsilon(p_\varepsilon) \tetrad^\alpha_\ta(\tau_\varepsilon),
\end{eqnarray*}
so that $L^\ta_\varepsilon=l^\ta_\varepsilon:= l^\ta_\varepsilon(p_\varepsilon)$, the components of $\vec{l}_\varepsilon$ in the basis $\{\vec{\tetrad}_\ta\}$ at $p_\varepsilon$. Then we can write $L^\alpha_\varepsilon(\tau)=l^\ta_\varepsilon \tetrad^\alpha_\ta(\tau)$, $\tau\in[0,\tau_\varepsilon]$, so that in particular, $L^\alpha_\varepsilon|_{p_0}=L^\alpha_\varepsilon(0)=l^\ta_\varepsilon \tetrad^\alpha_\ta(0)$ and 
\begin{equation}
\lim_{\varepsilon\to 0}\frac{1}{\varepsilon}
\left( L^\alpha_\varepsilon|_{p_0}-l^\alpha_0|_{p_0}\right)
=
\lim_{\varepsilon\to 0}\frac{1}{\varepsilon}
\left((l^\ta_\varepsilon-l^\ta_0)\tetrad^\alpha_\ta(0)\right).\label{limit-tetrad}
\end{equation}

We note that this makes the question of the existence of the limit $\vec{l}_1$ solely a question about the original one-parameter family of riggings. We note further that this \textit{does not} imply that $\vec{l}_1$ on its own is independent of the extension: it says that in the particular extension we are using has this feature \footnote{ This is consistent with the outcome of Lemma \ref{lemma_zeta_independent_extensions} about $\vec{\zeta}\equiv Q\vec{a}+\vec{l}_1$, namely that it is independent of the choice of extension. In the present construction, $\vec{a}=\nabla_{\vec{L}_0}\vec{L}_0=0$ since $\vec{L}_0$ is the initial tangent to the affinely parametrised geodesic $\gamma_0$.}.

To establish the existence of the limit of $(l^\ta_\varepsilon-l^\ta_0)/\varepsilon$, we introduce the 1-forms ${\boldsymbol{\invtetrad}}^\ta, \ta=0,1,\dots,\dimension$ dual to the parallel transported vectors $\vec{\tetrad}_\ta$, satisfying  
\begin{equation}
\tetrad^\alpha_\ta(\tau)\invtetrad^\tb_\alpha(\tau) = \delta^\tb_\ta,\tau\in[0,\tau_\varepsilon].
\end{equation}
Then the tetrad components $l^\ta_\varepsilon$ may be written as 
\begin{eqnarray}
l^\ta_\varepsilon &=& l^\ta_\varepsilon(\tau_\varepsilon) \nonumber\\
&=& \invtetrad^\ta_\alpha(\tau_\varepsilon)l^\alpha_\varepsilon|_{p_\epsilon}.
\label{main-limit}
\end{eqnarray}
Working in a fixed local coordinate system, we then have
\begin{eqnarray}
l^\ta_\varepsilon-l^\ta_0 &=& \invtetrad^\ta_\alpha(\tau_\varepsilon)l^\alpha_\varepsilon|_{p_\epsilon} - \invtetrad^\ta_\alpha(0)l^\alpha_0|_{p_0} \nonumber\\
&=&(\invtetrad^\ta_\alpha(\tau_\varepsilon)-\invtetrad^\ta_\alpha(0))l^\alpha_\varepsilon|_{p_\varepsilon} +
\invtetrad^\ta_\alpha(0)(l^\alpha_\varepsilon|_{p_\varepsilon} -l^\alpha_0|_{p_0}). \label{key-limit}
\end{eqnarray}

By smoothness (i.e.\ differentiability) of the geodesic $\tau\mapsto \gamma_0(\tau)$, of the solutions of the equations of parallel transport along $\gamma_0$ and of the one-parameter family of riggings (in the sense of Definition \ref{smooth-ext}), we see that in this local coordinate system, the limits

\begin{equation}
\lim_{\varepsilon\to0}\frac{1}{\varepsilon}(\invtetrad^\ta_\alpha(\tau_\epsilon)-\invtetrad^\ta_\alpha(0)) 
\end{equation}
and
\begin{equation}
\lim_{\varepsilon\to0} \frac{1}{\varepsilon}\left(l^\alpha_\varepsilon|_{p_\varepsilon}-l^\alpha_0|_{p_0}\right)
\end{equation}
both exist. This is sufficient to guarantee existence of the limit (\ref{key-limit}), proving the result via (\ref{main-limit}), provided that coordinate independence can be verified. This is an issue because the term $\invtetrad^\ta_\alpha(0)l^\alpha_\varepsilon|_{p_\varepsilon}$ involves a contraction of geometric objects defined at different points (without loss of generality, the local coordinate patch we work in covers both points). But under a smooth transformation of these coordinates
\begin{equation}
x^\alpha \to x^{\alpha'}=f^\alpha(x^\beta),
\end{equation}
we have
\begin{eqnarray}
\invtetrad^\ta_{\alpha'}(0)l^{\alpha'}_\varepsilon|_{p_\varepsilon} &=&
\invtetrad^\ta_{\alpha}(0)l^{\beta}_\varepsilon|_{p_\varepsilon}
\frac{\partial x^\alpha}{\partial x^{\alpha'}}|_{p_0}
\frac{\partial x^{\alpha'}}{\partial x^\beta}|_{p_\varepsilon} \nonumber\\
&=& \invtetrad^\ta_{\alpha}(0)l^{\alpha}_\varepsilon|_{p_\varepsilon} + O(\varepsilon),\label{coord-transf}
\end{eqnarray}
by virtue of the smoothness of the coordinate transformation, of the geodesic $\gamma_0$ and of the mapping $\varepsilon\to \tau_\varepsilon$. 
Thus the coordinate transformation introduces (at worse) an order-$\varepsilon$ correction, which does not affect the existence of the limit. With proof of the existence of the limit in hand, the coordinate independent nature of its value is immediate from the definition (\ref{limit-tetrad}).

Reinstituting the relevant notation, we recap and note that for each $p_0\in \overline{O}_p$, we have constructed the vectors $\vec{L}_\varepsilon(\tau;p_0)$ where $\tau\in[0,\tau_\varepsilon]$ and $\varepsilon\in I_{p_0}$. For each $p_0$, these vectors have the appropriate limiting behaviour at $\tau=\tau_\varepsilon$ (corresponding to property 3(a) of an extension in Definition \ref{smooth-ext}) and at $\tau=0$ (corresponding to 3(b)). By minimising over the compact set $\overline{O}_p$, we deduce the existence of an open interval $I_p$ containing the origin such that for all $p_0\in \overline{O}_p$, $\vec{L}_\varepsilon(\tau;p_0)$ is defined for all $\tau\in[0,\tau_\varepsilon]$ and for all $\varepsilon\in I_p$. This provides the interval whose existence is claimed in the statement of the proposition. The open neighbourhood $\cupo$ may be taken to be 
the interior of the set of points swept out by the congruence of geodesics with initial points $p_0\in O_p$ and initial tangents $\vec{l}_0|_{p_0}$, i.e.\ the set $\{\gamma_0(\tau): p_0\in O_p, \tau\in[0,\tau_\varepsilon],\varepsilon\in I_p\}$. Finally, the extension $\vec{L}(\varepsilon,x), x\in \cupo$ is defined by taking the union of the $\vec{L}(\varepsilon;p_0)$ as $p_0$ ranges over $O_p$. This ensures that the extension is defined throughout $\cupo$. 

\hfill $\blacksquare$

Nonetheless, other extensions in the sense of Definition \ref{smooth-ext} might be constructed, especially in particular settings, but we will not impose any particular one throughout our computations. We  work with generic extensions whenever these are needed for the construction of any geometrical object related to the hypersurfaces $\Sigma_\varepsilon$, taking care afterwards that these do not depend on the extensions. The limit (\ref{definition_l1}) defines the perturbation $\vec{l}_1$ as a spacetime vector at points of $\Sigma_0$, so that we consider its decomposition in the tangent basis of $\Sigma_0$, which reads $\vec{l}_1 \equiv \alpha \vec{l} + s^a \vec{e}_a$. Furthermore $\vec L(\varepsilon, x)$ defines in $\mathcal U$ an acceleration vector $\vec A(\varepsilon,x) \equiv \nabla_{\vec{L}(\varepsilon,x)}\vec{L}(\varepsilon,x)$, whose value at points of $\Sigma_0$ defines the acceleration of the rigging vector, i.e. $\vec a \equiv \nabla_{\vec{l}} \vec{l}|_{\Sigma_0} = \vec{A}(\varepsilon \rightarrow 0, x|_{\Sigma_0} )$.

In the following, we make use of Lemma \ref{lemma_marc_perturbations} to compute explicit expressions for the perturbations of the hypersurface metric data, i.e. $\{\delta h, \delta \boldsymbol{l}, {l}^{(2)}, \delta Y\}$.

\section{Perturbations of the \textit{hypersurface data}}
\label{section_perturbations_data}

The formalism introduced so far is sufficient to produce explicit expressions for the first fundamental form and rigging data perturbed to first order.

\begin{proposition} \label{proposition_deltah}
	Let $(\mathcal{M},g)$ be a ($\dimension$+1)-dimensional spacetime equipped with a $C^2$ metric and $\abshyp$ a hypersurface embedded by $\Phi: \abshyp \rightarrow \Sigma_0 \subset \mathcal{M}$, for which a transverse vector $\vec{l}$ has been specified.

	If the metric is perturbed to linear order with $g_1$ and the hypersurface with a vector field $\vec{Z} = Q \vec{l} + \vec{T}$, the first fundamental form is perturbed as
	\begin{eqnarray}
	\delta h_{ab} &=& 2QY_{ab} + 2l_{(a}\vec{e}_{b)}Q + \mathcal{L}_{\vec T{}_\abshyp} h_{ab} + g_1(\vec{e}_a, \vec{e}_b),\label{pert_fff}
	\end{eqnarray}
	%where $\tilde{a}^\alpha := \nabla_{\vec{l}}\; l^\alpha$.
	where $\vec{T}{}_{\abshyp}$ is the vector in $\abshyp$ such that $\vec T \equiv d \Phi (\vec{T}{}_\abshyp)$.
\end{proposition}
\begin{proof}
	The $\varepsilon-$family of embeddings and spacetime metrics define a $\varepsilon-$family of first fundamental forms 
	\begin{equation*}
	h_\varepsilon = \Phi_\varepsilon^* \left(\left.g_\varepsilon \right|_{\Sigma_\varepsilon} \right).
	\end{equation*}
	A direct application of Lemma \ref{lemma_marc_perturbations} provides
	\begin{equation*}
	\left.\frac{d h_\varepsilon}{d \varepsilon}\right|_{\varepsilon=0} = \Phi^*(\LL_{\vec{Z}}g) + \Phi^*(g_1). \label{fff_intermediate_1}
	\end{equation*}
	Applying Corollary \ref{pullbackliemetric} to the first term, expression (\ref{pert_fff}) is obtained. 
	
	\hfill $\blacksquare$
\end{proof}

\begin{proposition} \label{proposition_deltal}
	The $\varepsilon$-family of the rigging data is given by the following one-parameter family of scalars and one forms in $\abshyp$ 
	\begin{equation}
	 l_\varepsilon^{(2)} \equiv \Phi_\varepsilon^* \left(g_\varepsilon|_{\Sigma_\varepsilon} (\vec{l}_\varepsilon,\vec{l}_\varepsilon )\right), \quad  \boldsymbol{ l}_\varepsilon \equiv \Phi_\varepsilon^* \left(g_\varepsilon|_{\Sigma_\varepsilon} (\vec{l}_\varepsilon,\cdot )\right), \label{leSigma}
	\end{equation}
	and in terms of the vector field $\vec{\zeta} := \vec{l}_1 + Q \vec{a}$, their first order perturbations read
	\begin{eqnarray}
	\delta l^{(2)} &=& 2 g(\vec{\zeta}, \vec{l})
	%2 \Phi^*(l_\alpha \zeta^\alpha) 
	+ 2 T^a (l^{(2)} \varphi_a + \Psi_a^b l_b) + g_1(\vec{l},\vec{l}),\label{l12_b}\\
	\delta l_a &=& g(\vec{\zeta}, \vec{e}_a) + 
	%\Phi^* (\boldsymbol{\zeta}) +  
	Q (l^{(2)} \varphi_a + \Psi^b_a l_b) + l^{(2)}\vec{e}_a(Q) + \mathcal{L}_{\vec{T}_\abshyp}l_a + g_1 (\vec{l},\vec{e}_a) .\label{l1a_b}
	\end{eqnarray}
\end{proposition}
\begin{remark}
	These expressions for $\delta l^{(2)}$ and $\delta l_a$ depend on the vector field $\vec\zeta$, which arises as a new ingredient inherent to the perturbative method. This raises the question whether the expressions (\ref{l12_b}) and (\ref{l1a_b}) depend on the extension of the rigging $\vec{L}(\varepsilon, x)$. The answer is negative, as follows from Lemma \ref{lemma_independent_extensions}, and it stands as a satisfying feature of this method for computing the perturbations of the \textit{hypersurface data}.  Also, the vector field $\vec \zeta$ is closely related to the freedom in the choice of rigging in the perturbative setting (see Proposition \ref{proposition_gauge_inversion} and the subsequent discussion) and therefore the perturbations can actually be computed for particular problems.
\end{remark}
\begin{proof}
	We start noting the following identity
	\begin{equation}
	\mathcal{L}_{\vec{L}} \boldsymbol{L}|_{\{\varepsilon \rightarrow 0,\Sigma_0\}} = a_\alpha + \left.\frac{1}{2} \nabla_\alpha \left(L_\mu L^\mu \right)\right|_{\{\varepsilon \rightarrow 0,\Sigma_0\}} , \label{Lll}
	\end{equation}
	where $\boldsymbol{L} \equiv g(\vec{L}, \cdot)$.
	The objects (\ref{leSigma}) can be linearised in terms of the extension $\vec{L}(\varepsilon, x)$ by direct application of Lemma \ref{lemma_marc_perturbations}. For the one form this provides
\begin{eqnarray*}
	\delta l_a &\equiv& \left.\frac{d \boldsymbol{l}_\varepsilon}{d \varepsilon}  \right|_{\varepsilon=0} = g_1(\vec{l},\vec{e}_a) + g(\vec{l}_1, \vec{e}_a) + 
	%\alpha l_a + h_{ab} s^b + 
	Q\Phi_0^* \left(\mathcal L_{\vec{L}}\mathbf{L}|_{\{\varepsilon \rightarrow 0,\Sigma_0\}} \right) + l^{(2)} \vec{e}_a (Q) + \mathcal L_{\vec{T}_{\abshyp}}l_a.
	\label{l1a}
\end{eqnarray*}
	We use the decomposition $\vec{l}_1 = \alpha \vec{l} + s^a \vec{e}_a$, in terms of its rigged and projected components,  and use expression (\ref{Lll}) for the Lie derivative of the extension of the rigging. This introduces a directional derivative along the tangential directions of the norm of the rigging that we expand using the second expression in (\ref{nabla_normal_rigging}). This leads to (\ref{l1a_b}). 
	
	We deal now in a very similar way with the scalar $\delta l^{(2)}$:
	\begin{eqnarray*}
\delta l^{(2)} &\equiv& \left.\frac{d l_\varepsilon^{(2)}}{d \varepsilon}  \right|_{\varepsilon=0} = g_1(\vec{l},\vec{l}) + 2\alpha l^{(2)} + 2s^al_a + Q \vec{l}( L_\mu L^\mu)|_{\{\varepsilon \rightarrow 0,\Sigma_0\}} + \vec{T}_{\abshyp}(l^{(2)}),\label{l12}
	\end{eqnarray*}
We use again expression (\ref{Lll}) for the derivative along the rigging of the covariant form of the extension $\vec{L}$ and expression (\ref{nabla_normal_rigging}) for the tangential derivative of $l^{(2)}$ and we obtain (\ref{l12_b}).

	\hfill $\blacksquare$
\end{proof}

\begin{remark}
	The perturbations of the hypersurface metric data do not have the properties that define the hypersurface metric data, since the matrix $\{\delta l^{(2)}, \delta l_a, \delta h_{ab} \}$ does not have Lorentzian signature in general. Moreover it depends very strongly on different sources of freedom inherent to the perturbative method, and for a particular choice of the gauges involved, this matrix becomes degenerate. This issue is very similar to the spacetime perturbations, where the tensor $g_1$ does not present the properties of a metric, in general.
\end{remark}

So far we have considered a $\varepsilon$-family of \textit{hypersurface data} $\{h_\varepsilon, \mathbf{l}_\varepsilon, l_\varepsilon^{(2)}\}$, whose linearisation through Lemma \ref{lemma_marc_perturbations} provides its first order perturbations $\{\delta h_{ab}, \delta l_a, \delta l^{(2)}\}$ in terms of the background elements (metric and hypersurface) and the perturbation vector $\vec{Z}$. In analogy with the exact case, we define a set of fields $\{P_\varepsilon, n_\varepsilon, n_\varepsilon^{(2)}\}$ on $\abshyp$ as the solutions of the $\varepsilon-$ version of the system (\ref{constraints_definitions}). These objects are defined on $\abshyp$, and therefore we can take $\varepsilon$-derivatives directly. This operation defines the fields $\{\delta P^{ab}, \delta n^a, \delta n^{(2)}\}$ as the solutions of
\begin{eqnarray}
h_{bc} \delta P^{ac} + l_b \delta n^a = -P^{ac}\delta h_{bc} - n^a \delta l_b \equiv A^a_b,\label{eqid_1}\\
l_b\delta P^{ab} + l^{(2)}\delta n^a = -P^{ab}\delta l_b - n^a \delta l^{(2)} \equiv B^a,\label{eqid_2}\\
l_a \delta n^a + l^{(2)} \delta n^{(2)} = -n^a \delta l_a - n^{(2)} \delta l^{(2)} \equiv C \label{eqid_3},\\
h_{ab} \delta n^b + l_a \delta n^{(2)} = -n^{(2)} \delta l_a - n^b \delta h_{ab} \equiv D_a \label{eqid_4}.
\end{eqnarray}
Note that $C n^a + P^{ab}D_b = n^b A^a_b + n^{(2)}B^a$.

\begin{lemma} \label{lemma_inverse_perturbed_data}
	The first order perturbations of the objects $\{P^{ab}, n^a, n^{(2)} \}$ expressed in terms of the first order objects $\{\delta h_{ab}, \delta l_a, \delta l^{(2)}\}$ and the metric data read
	\begin{eqnarray}
	\delta P^{ab} &=& -P^{bc} P^{ad} \delta h_{cd} - (n^a P^{bc} + n^b P^{ac})\delta l_c - n^a n^b \delta l^{(2)}, \label{deltaP}\\
	\delta n^a &=& - n^b P^{ac} \delta h_{bc} - n^{(2)} n^a \delta l^{(2)} - (n^{(2)}P^{ab}+n^a n^b) \delta l_b,\label{delta_na}\\
	\delta n^{(2)} &=& -n^a n^c \delta h_{ac} - n^{(2)} (n^{(2)} \delta l^{(2)} + 2n^c \delta l_c). \label{delta_n2}
	\end{eqnarray}
\end{lemma}
\begin{proof}
	Start with equations (\ref{eqid_3}) and (\ref{eqid_4}) and consider them equations for $\delta n^a$, which is a vector on $\abshyp$. The first terms in both of them are $x \equiv l_a \delta n^a$ and $x_a \equiv h_{ab} \delta n^b$, and according to Lemma \ref{lemma_vector_reconstruction} define a unique vector  $\delta n^a = xn^a + x_b P^{ab}$ on $\abshyp$, provided the constraint $xn^{(2)} + n^a x_a = 0$ holds. The explicit expression for the vector field results in (\ref{delta_na}) and the constraint provides  equation (\ref{delta_n2}).
	
	The same method applied to equations (\ref{eqid_1}) and (\ref{eqid_2}) determine $\delta P^{ab}$, with the constraint being identically satisfied by $\delta n^a$ provided (\ref{delta_na}) holds.
	
	\hfill $\blacksquare$
\end{proof}

\begin{remark}
Although Lemma \ref{lemma_inverse_perturbed_data} defines the objects $\{\delta P_{ab}, \delta n^a, \delta n^{(2)}\}$ in terms of the \textit{allowed ingredients} of the perturbations, we do not claim that these correspond to the linear perturbations of some $\varepsilon$-dependent geometrical spacetime objects, since this statement would require a version of Lemma \ref{lemma_marc_perturbations} suitable for contravariant tensors.
\end{remark}

\begin{proposition} \label{proposition_deltaY}
	The first order perturbation of the rigged fundamental form defined by
	\begin{equation}
	\delta Y := \left.\frac{d}{d \varepsilon} Y_\varepsilon \right|_{\varepsilon=0} = \left.\frac{d}{d \varepsilon} \Phi^*_\varepsilon \mathcal L_{\vec{l}_\varepsilon} g_\varepsilon  \right|_{\varepsilon=0}, \label{definition_deltaY}
	\end{equation}
	in terms of the vector field  $\vec \zeta$ reads
	\begin{eqnarray}
	\delta Y_{ab} &=&\frac{1}{2}\Phi^* \left( \mathcal{L}_{\vec \zeta}  g\right) + \frac{1}{4}\left(\partial_a Q \partial_b l^{(2)} + \partial_bQ \partial_a l^{(2)}\right)+ Q \left( -R_{\alpha \gamma \beta \mu} e_a^\alpha l^\gamma  e_b^\beta l^\mu + g(\nabla_{\vec{e}_a} \vec{l}, \nabla_{\vec{e}_b} \vec{l}) \right) \nonumber\\
	&+&  \mathcal{L}_{\vec{T}_\abshyp} Y_{ab} + \frac{1}{2} \Phi^* \mathcal L_{\vec{l}}g_1  . \label{Yp_1_b}
	\end{eqnarray}

\end{proposition}
\begin{proof}
	
		A direct application of Lemma \ref{lemma_marc_perturbations} to \ref{definition_deltaY} provides the expression
	\begin{equation}
	2 \delta Y =  \Phi^* \mathcal{L}_{\vec{l}_1} g + \Phi^* \mathcal{L}_{\vec{l}} g_1 + \Phi^* \mathcal{L}_{\vec{Z}} \mathcal{L}_{\vec{l}} g. \label{mastereq}
	\end{equation}
	Since the calculations below become rather involved, we will abandon the notation for the extensions that we used for Lemma \ref{lemma_zeta_independent_extensions} or Proposition \ref{proposition_deltal} and use the symbol $\vec{l}$ instead (or $\mathbf l$ for the metrically related one form). Nonetheless, the computations should be understood in the sense of the extension $\vec{L} (\varepsilon, x)$ (and $\boldsymbol{L} (\varepsilon, x)$). 
	
	The first and third terms in (\ref{mastereq}) require the use of extensions of the objects involved, so that the next task is to perform a series of manipulations in order to express them in terms of allowed objects, i.e. background objects, the metric perturbations, the perturbations to the rigging and the deformation vectors.
	We start by commuting the Lie derivatives in the third term of (\ref{mastereq})
	\begin{equation*}
	\Phi^* \mathcal{L}_{\vec{Z}} \mathcal{L}_{\vec{l}} g = \Phi^* \mathcal{L}_{\vec{l}} \mathcal{L}_{\vec{Z}} g + \Phi^* \mathcal{L}_{\mathcal{L}_{\vec{Z}} \vec{l}}  g,
	\end{equation*}
	so that we are able to write (\ref{mastereq}) in a slightly different way, more convenient in order to manipulate the terms that depend on extensions:
	\begin{equation}
	2 \delta Y =  \Phi^* \mathcal{L}_{\vec{l}_1 + \mathcal{L}_{\vec{Z}} \vec{l}} g + \Phi^* \mathcal{L}_{\vec l} \mathcal{L}_{\vec{Z}} g +  \Phi^* \mathcal{L}_{\vec{l}} g_1. \label{mastereq2}
	\end{equation}
	We define the vector $\vec{W} := \vec{l}_1 + \mathcal{L}_{\vec{Z}} \vec{l}$, whose decomposition into the transverse and tangential components to $\Sigma_0$ is $\vec W = W l^\alpha + W^a e_a^\alpha$ with
	\begin{eqnarray}
	W &:=& n_\alpha W^\alpha = \alpha + n_\alpha \mathcal{L}_{\vec{Z}} l^\alpha = \alpha -\vec{l}(Q) + n_\alpha \mathcal{L}_{\vec{T}} l^\alpha, \nonumber\\
	W^a &:=& \omega^a_\alpha W^\alpha = s^a + \omega^a_\alpha \mathcal{L}_{\vec{Z}} l^\alpha = s^a + \omega^a_\alpha \mathcal{L}_{\vec{T}} l^\alpha. \label{Wdecomposition}
	\end{eqnarray}
	This splitting is useful to write the first term in (\ref{mastereq2}), using (\ref{Llmetric}) and (\ref{LTmetric}), as
	\begin{equation*}
	\Phi^* \mathcal{L}_{\vec W} g = 2 W Y_{ab} + l_a \partial_b W + l_b \partial_a W + \mathcal{L}_{\vec W{}_\abshyp}h_{ab}.
	\end{equation*}
	The second term in (\ref{mastereq2}) involves a second derivative, which can be expanded as 
	\begin{equation}
	\Phi^* \mathcal{L}_{\vec l} \mathcal{L}_{\vec{Z}} g = \Phi^* \lbrace \vec{l}(Q) \mathcal L_{\vec{l}}g + Q \mathcal{L}_{\vec l} \mathcal{L}_{\vec l} g + ( \mathcal{ L}_{\vec l} dQ \otimes \mathbf{l} + \mathbf{l} \otimes \mathcal{L}_{\vec l} dQ) + (dQ \otimes \mathcal{L}_{\vec{l}} \mathbf{l} + \mathcal{L}_{\vec{l}} \mathbf{l}  \otimes dQ) + \mathcal{L}_{\vec l} \mathcal{L}_{\vec T}g\rbrace. \label{LlLZg}
	\end{equation}
	We need to analyse every single term in the expression above, since all of them contain transverse derivatives.
	The first one is straightforward
	\begin{equation*}
	\Phi^*  \vec{l}(Q) \mathcal L_{\vec{l}}g  = 2\vec{l}(Q) Y_{ab},
	\end{equation*}
	whereas for the second one we expand the first Lie derivative as a covariant derivative and commute it with the second Lie derivative, using formula (\ref{commute_Lie_cov}), which results into
	\begin{eqnarray*}
	\Phi^* \mathcal{L}_{\vec l} \mathcal{L}_{\vec l} g &=& \Phi^* \mathcal{L}_{\vec l} (\nabla_\alpha l_\beta + symm \; \; \alpha \leftrightarrow \beta)\\
	&=& e_a^\alpha e_b^\beta (\nabla_\alpha \mathcal{L}_{\vec l} l_\beta + \nabla_\beta \mathcal{L}_{\vec l} l_\alpha) + e_a^\alpha e_b^\beta (R_{\mu \alpha \beta}^{\phantom{\mu \alpha \beta}\rho} +R_{\mu \beta \alpha }^{\phantom{\mu \beta \alpha}\rho} ) l^\mu - e_a^\alpha e_b^\beta (\nabla_\alpha \nabla_\beta l^\rho + \nabla_\beta\nabla_\alpha  l^\rho)l_\rho \nonumber\\
	&=&( \Phi^* \mathcal L_{\nabla_{\vec{l}} \vec{l}} g + e_a^\alpha e_b^\beta \nabla_\alpha \nabla_\beta (l_\mu l^\mu)) + 2R_{\mu \beta \alpha \rho}l^\rho l^\mu e_a^\alpha e_b^\beta - e_a^\alpha e_b^\beta(\nabla_\alpha \nabla_\beta (l_\mu l^\mu) - 2 \nabla_\alpha l^\rho \nabla_\beta l_\rho)\nonumber\\
	&=&  \Phi^* \mathcal L_{\nabla_{\vec{l}} \vec{l}} g  + 2R_{\mu \beta \alpha \rho}l^\rho l^\mu e_a^\alpha e_b^\beta + 2 g(\nabla_{\vec{e}_a} \vec{l}, \nabla_{\vec{e}_b} \vec{l}).\nonumber
	\end{eqnarray*}
	The third term in (\ref{LlLZg}) can be arranged taking into account that the Lie derivative and the exterior derivative commute. Hence we have that
	\begin{eqnarray*}
	&& \Phi^* ( \mathcal{ L}_{\vec l} dQ \otimes \mathbf{l} + \mathbf{l} \otimes \mathcal{L}_{\vec l} dQ) = \Phi^* ( d\mathcal{ L}_{\vec l} Q \otimes \mathbf{l} + \mathbf{l} \otimes d\mathcal{L}_{\vec l} Q) = l_a \partial_b \vec{l}(Q) + l_b \partial_a \vec{l}(Q).
	\end{eqnarray*}
	It is convenient to write the fourth term in (\ref{LlLZg}) in a way that makes the acceleration of the rigging explicit, which we do using relation (\ref{Lll})
	\begin{equation*}
	\Phi^* (dQ \otimes \mathcal{L}_{\vec{l}} \mathbf{l} + \mathcal{L}_{\vec{l}} \mathbf{l}  \otimes dQ) = a_b\partial_a Q + a_a\partial_b Q + \frac{1}{2} (\partial_a l^{(2)} \partial_b Q + \partial_b l^{(2)} \partial_a Q ).
	\end{equation*}
	Finally, we commute the derivatives in the fifth term 
	\begin{equation*}
	\Phi^* \mathcal L_{\vec l} \mathcal L_{\vec T}g = \Phi^* \mathcal L_{\vec T} \mathcal L_{\vec l}g + \Phi^* \mathcal L_{[\vec l, \vec T]} g,
	\end{equation*}
	and use the decomposition into the rigged and tangent parts of the vector $[\vec{l}, \vec{T}] \equiv B \vec{l} + B^a \vec{e}_a$ 
	\begin{eqnarray}
	B: = n_\alpha [\vec{l}, \vec{T}]^\alpha = -n_\alpha [\vec{T}, \vec{l}]^\alpha = -n_\alpha \mathcal L_{\vec{T}}l^\alpha, \quad B^a := \omega^a_\alpha  [\vec{l}, \vec{T}]^\alpha = -\omega^a_\alpha [\vec{T}, \vec{l}]^\alpha = - \omega^a_\alpha  \mathcal L_{\vec{T}}l^\alpha, \label{Bdecomposition}
	\end{eqnarray}
	which renders this term as 
	\begin{eqnarray*}
	\Phi^* \mathcal L_{\vec l} \mathcal L_{\vec T}g =  2\mathcal L_{\vec T{}_\abshyp} Y_{ab} + 2B Y_{ab} + l_a \partial_b B + l_b \partial_a B+  \mathcal L_{\vec B{}_\abshyp} h_{ab}  .
	\end{eqnarray*}
	Finally, we gather all these expressions and use them to write (\ref{LlLZg}) as
	\begin{eqnarray*}
	\Phi^* \mathcal{L}_{\vec l} \mathcal{L}_{\vec{Z}} g &=& 2(\vec{l}(Q) + B) Y_{ab} + Q\Phi^* \mathcal L_{\nabla_{\vec{l}} \vec{l}} g  + 2QR_{\mu \beta \alpha \rho}l^\rho l^\mu e_a^\alpha e_b^\beta + 2Q g(\nabla_{\vec{e}_a} \vec{l}, \nabla_{\vec{e}_b} \vec{l})\nonumber\\
	&+& l_a \partial_b (B + \vec{l}(Q)) + l_b \partial_b (B+\vec{l}(Q)) + a_b\partial_a Q + a_a\partial_b Q + \frac{1}{2} (\partial_a l^{(2)} \partial_b Q + \partial_b l^{(2)} \partial_a Q ) \nonumber\\
	&+&  2\mathcal L_{\vec T{}_\abshyp} Y_{ab} + \mathcal L_{\vec B{}_\abshyp} h_{ab},
	\end{eqnarray*}
	which taken into (\ref{mastereq2}) provides the expression
	\begin{eqnarray}
	2 \delta Y_{ab} &=&  2QR_{\mu \beta \alpha \rho}l^\rho l^\mu e_a^\alpha e_b^\beta + 2 Qg(\nabla_{\vec{e}_a} \vec{l}, \nabla_{\vec{e}_b} \vec{l}) + \frac{1}{2} (\partial_a l^{(2)} \partial_b Q + \partial_b l^{(2)} \partial_a Q ) +  2\mathcal L_{\vec T{}_\abshyp} Y_{ab}  +   \Phi^* \mathcal{L}_{\vec{l}} g_{1}\nonumber\\
	&+&\left \lbrace 2(W+\vec{l}(Q) + B) Y_{ab} + Q\Phi^* \mathcal L_{\nabla_{\vec{l}} \vec{l}} g +  l_a \partial_b (W+ B + \vec{l}(Q)) + l_b \partial_a (W+ B+\vec{l}(Q)) \right.\nonumber\\
	&+&\left. a_b\partial_a Q + a_a\partial_b Q +  \mathcal L_{\vec W_{\abshyp} + \vec B{}_\abshyp} h_{ab}\right \rbrace. \label{mastereq3}
	\end{eqnarray}
	We have grouped together inside the braces the terms that require transverse derivatives to be computed. Two relevant combinations show up
	\begin{eqnarray*}
	&& W + \vec{l}(Q) + B = (\alpha -\vec{l}(Q) + n_\mu \mathcal{L}_{\vec{T}} l^\mu) + \vec{l}(Q) - n_\mu \mathcal L_{\vec{T}}l^\mu  = \alpha,\\
	&&\vec W_{\abshyp} + \vec B{}_\abshyp = ( s^a + \omega^a_\mu \mathcal{L}_{\vec{T}} l^\mu)  - \omega^a_\mu  \mathcal L_{\vec{T}}l^\mu = \vec{s}{}_\abshyp,
	\end{eqnarray*}
	where we used the explicit expressions for $\vec{W}$ and $\vec{B}$ given in (\ref{Wdecomposition}) and (\ref{Bdecomposition}). The term in braces in (\ref{mastereq3}) is thus
	\begin{eqnarray*}
	&&\left \lbrace 2(W+\vec{l}(Q) + B) Y_{ab} + Q\Phi^* \mathcal L_{\nabla_{\vec{l}} \vec{l}} g +  l_a \partial_b (W+ B + \vec{l}(Q)) + l_b \partial_a (W+ B+\vec{l}(Q)) \right.\nonumber\\
	&+&\left. a_b\partial_a Q + a_a\partial_b Q +  \mathcal L_{\vec W_{\abshyp} + \vec B{}_\abshyp} h_{ab}\right \rbrace	= 2 \alpha Y_{ab} +l_a \partial_b \alpha + l_b \partial_a \alpha +  \mathcal L_{\vec{s}{}_\abshyp} h_{ab}
	+ Q \Phi^* \mathcal L_{\vec{a}}g + a_b\partial_a Q + a_a\partial_b Q\nonumber\\
	&=& \Phi^* \mathcal L_{\vec{l}_1} g + \Phi^* \mathcal L_{Q\vec{a}} g =  \Phi^* \mathcal L_{\vec{\zeta}} g.
	\end{eqnarray*}
	
	\hfill $\blacksquare$
\end{proof}

The only dependence of $\delta Y$ with the extension of the rigging is encoded in the first term of expression (\ref{Yp_1_b}) which involves $\vec \zeta$.

\section{Freedom in the method}
\label{section_freedom}

There are four sources of freedom inherent to the method of perturbing general hypersurfaces, namely (i) the spacetime gauge (ii) the hypersurface gauge (iii) a local extension of the vector $\vec{l}_\varepsilon$ to an open neighbourhood $\mathcal U$ (iv) the nonuniqueness of the rigging vectors $\vec{l}_\varepsilon$ at $\Sigma_\varepsilon$. 

\subsection{Inherent freedom in perturbation theory}
\label{subsection_inherent_gauge_freedom}
The inherent degrees of freedom corresponding to points (i) and (ii) were already discussed in \cite{Mars2005}, and our treatment for general hypersurfaces is completely analogous. We include a brief summary next for completeness.

The spacetime gauge freedom arises from the non-uniqueness of the diffeomorphism $\psi_\varepsilon$ used to identify the different $\mathcal M_\varepsilon$ among themselves. In fact, a different choice, say $\psi_\varepsilon^{(g)}$, can be seen as the composition of the old identification with a $\varepsilon$-dependent diffeomorphism in the background, i.e. $\Omega_\varepsilon: \mathcal M \rightarrow \mathcal M$. Thus the new identification is $\psi_\varepsilon^{(g)} = \psi_\varepsilon \circ \Omega_\varepsilon$ and it induces a new family of metrics $g_\varepsilon^{(g)} = \psi^{(g)}_\varepsilon{}^*(\hat{g}_\varepsilon) = \Omega_\varepsilon^*(g_\varepsilon)$. In terms of the spacetime gauge vector field $\vec{s}_1 \equiv \partial_\varepsilon \Omega_\varepsilon|_{\varepsilon=0}$, the metric perturbations are related by $g_1 ' = g_1 + \LL_{\vec{s}_1} g$.
%\begin{equation}
%g_1 ' = g_1 + \LL_{\vec{s}_1} g.
%\end{equation}
However, the geometrical objects that characterize the geometry of the hypersurfaces $\Sigma_\varepsilon$ are independent of the spacetime gauge by construction. 

The hypersurface gauge arises from taking a $\varepsilon$-dependent diffeomorphism $\chi_\varepsilon$ in $\abshyp$ previous to the identification with the $\hat \Sigma_\varepsilon$. This generates a new diffeomorphism $\phi_\varepsilon^{(h)} = \phi_\varepsilon \circ \chi_\varepsilon$, with $\chi_0$ being the identity transformation in $\abshyp$. In turn, the embeddings change as follows $\Phi_\varepsilon^{(h)} = \psi_\varepsilon \circ \phi_\varepsilon^{(h)} = \psi_\varepsilon \circ \phi_\varepsilon \circ \chi_\varepsilon = \Phi_\varepsilon \circ \chi_\varepsilon$
The new set of one-parameter fields in $\abshyp$ is given thus by
\begin{eqnarray*}
h^{(h)}_\varepsilon &\equiv& \Phi^{(h)}_\varepsilon{}^* g_\varepsilon = \chi_\varepsilon^* ( \Phi_\varepsilon^* g_\varepsilon) = \chi_\varepsilon^* h_\varepsilon, \\
{\boldsymbol{l}}_\varepsilon^{(h)}  &\equiv& \Phi^{(h)}_\varepsilon{}^* g_\varepsilon(\vec{l}_\varepsilon, \cdot) = \chi_\varepsilon^*  {\boldsymbol{l}}_\varepsilon, \\
{l}_\varepsilon^{(2)}{}^{(h)}  &\equiv& \Phi^{(h)}_\varepsilon{}^* g_\varepsilon(\vec{l}_\varepsilon, \vec{l}_\varepsilon) = \chi_\varepsilon^*  {l}_\varepsilon^{(2)},\\
Y^{(h)}_\varepsilon &=& \chi_\varepsilon^* Y_\varepsilon.
\end{eqnarray*}

We apply Lemma 1 from \cite{Mars2005}, which is just a general version of Lemma  \ref{lemma_marc_perturbations} that rather than starting from a hypersurface $\abshyp$ and an ambient spacetime $\mathcal M$, it applies to a pair of differentiable manifolds $\mathcal N$ and $\mathcal M$ endowed with a differential map $\chi_\varepsilon: \mathcal N \rightarrow M$. We set $\mathcal N = \mathcal M = \abshyp$, so that the diffeomorphism that makes the diagram (6) from \cite{Mars2005} commutative is $\Psi_h^\varepsilon = \chi_{\varepsilon +h} \circ \chi_{\varepsilon}^{-1}$. As long as the diffeomorphism $\chi_0= \mathbb{I}$ in $\abshyp$, we have that
\begin{equation*}
\left. \frac{\partial \Psi^\varepsilon_h}{\partial h}\right|_{\varepsilon=h=0} = \left. \frac{\partial \chi_\varepsilon}{\partial \varepsilon}\right|_{\varepsilon=0} \equiv \vec{u}_\abshyp.
\end{equation*}
This vector field in $\abshyp$ is called the hypersurface gauge vector. The direct application of Lemma 1 from \cite{Mars2005} with $\mathcal M = \mathcal N$ and $\vec{u}_\abshyp \equiv \left.\frac{\partial \chi_\varepsilon}{\partial \varepsilon}\right|_{\varepsilon = 0}$ provides
\begin{eqnarray}
\delta h_{ab}{}^{(h)} = \delta h_{ab} + \LL_{\vec{u}_\abshyp} h_{ab}, \quad \delta l_a^{(h)} = \delta l_a + \LL_{\vec{u}_\abshyp}l_a, \quad \delta l^{(2)}{}^{(h)} = \delta l^{(2)} + \LL_{\vec{u}_\abshyp}l^{(2)}, \quad \delta Y_{ab}^{(h)} = \delta Y_{ab} + \LL_{\vec{u}_\abshyp} Y_{ab}. \nonumber\\ \label{hypersurface_gauge_transformations}
\end{eqnarray}
Moreover, the perturbation vector $\vec{Z}$ is affected by this gauge freedom.
\begin{proposition}{\bf (Mars 2005 \cite{Mars2005})}
	Under a hypersurface gauge transformation on $\abshyp$ defined by a gauge vector $\vec{u}_\abshyp$ to first order, the deformation vector $\vec{Z}$ 
	%at any point $p\in \Sigma_0$ 
	transforms as $Q^{(h)} \rightarrow Q$ and $\vec{T}_\abshyp^{(h)} \rightarrow \vec{T}_\abshyp + \vec{u}_\abshyp$
%	\vec{Z}{}^{(h)} = \vec{Z} + \vec{u}$. 
\end{proposition} 
Note that in the first order perturbations, a hypersurface gauge transformation affects only the tangential component $\vec{T}$ of $\vec{Z}$.

\subsection{Freedom associated to the rigging vector}

We start with the issue of the extensions of $\vec{l}_\varepsilon$ to $\mathcal U$. These appear manifestly through the vector field $\vec{\zeta}$ in the expressions for $\{\delta l^{(2)}, \delta l_a, \delta Y_{ab} \}$ in Propositions \ref{proposition_deltal} and \ref{proposition_deltaY}. Hence, the task here is to characterize the dependence of this vector field on the extensions. A different extension shall be denoted by a tilde, i.e. $\vec{\widetilde{L}}$, and the same rule applies for other quantities that depend on the extensions. 
\begin{lemma}
	The vector field $\vec \zeta := Q \vec{a} + \vec{l}_1$ is independent of the extension of the rigging. \label{lemma_zeta_independent_extensions}
\end{lemma}

\begin{proof} {

		Keeping the riggings $\vec{l}_\varepsilon$ fixed at each $\Sigma_\varepsilon$, let us  consider any other extension $\vec {\widetilde{L}}(\varepsilon,x)$ to $\mathcal U$, which also satisfies the properties from Definition \ref{smooth-ext}. These two extensions differ at most by a vector field $\vec k$ on $\mathcal U$ that necessarily has the properties
		\begin{equation*}
		\text{(i)} \lim_{\varepsilon \rightarrow 0} \frac{\vec k(\varepsilon, x|_{\Sigma_0})}{\varepsilon} = \vec{\widetilde l}_1 - \vec{l}_1 \equiv \vec{k}_1, \quad %\text{(ii)} \; \text{$\vec{k}$ is $C^2$  in $\mathcal U$} , \quad 
		\text{(ii)}\; \vec{k}(\varepsilon, x|_{\Sigma_\varepsilon}) = 0.
		\end{equation*}
		Moreover, the accelerations $\vec {\widetilde A}$  and $\vec{A}$ associated to the two different extensions are related as follows
		\begin{equation*}
		\vec {\widetilde{A}} (\varepsilon, x) \equiv \nabla_{\vec{\widetilde L}} \vec{\widetilde L} = \vec A (\varepsilon, x) + \nabla_{\vec{L} (\varepsilon, x)} \vec k (\varepsilon, x) + \nabla_{\vec k (\varepsilon, x)} (\vec L(\varepsilon, x) + \vec k(\varepsilon, x)).
		\end{equation*}
		A consequence of this relation is that at points of $\Sigma_0$, $\vec{\widetilde a} = \vec{\widetilde A}(\varepsilon \rightarrow 0, x|_{\Sigma_0}) = \vec{a} + (\nabla_{\vec{l}} \vec{k})|_{\{\varepsilon \rightarrow 0, \Sigma_0 \}}$, which immediately leads to $\vec{\widetilde \zeta} = \vec{\zeta} + Q (\nabla_{\vec{l}} \vec{k})|_{\{\varepsilon \rightarrow 0, x|_{\Sigma_0} \}} + \vec{k}_1$. We claim that the vector $\vec{v} \equiv Q(\nabla_{\vec{l}} \vec{k})|_{\{\varepsilon \rightarrow 0, x|_{\Sigma_0} \}} + \vec{k}_1$ is identically zero, so that $\vec \zeta$ is indeed independent of the choice of the extension of $\vec{l}_\varepsilon$. To see this we consider the following objects 
		\begin{eqnarray*}
		\frac{d}{d\varepsilon} \left.\left( \Phi_\varepsilon^* g_\varepsilon(\vec k_\varepsilon, \cdot)\right) \right|_{\varepsilon = 0} &=& \Phi_0^* g_1 \left( \vec{k}(0, x|_{\Sigma_0}), \cdot \right) + \Phi_0^* g\left( \left.\frac{\partial k(\varepsilon, x)}{\partial \varepsilon}\right|_{\Sigma_0}, \cdot\right) \nonumber\\
		&+&\Phi_0^*  \mathcal L_{Q\vec{l}} g(\vec{k}(0, x), \cdot )|_{\Sigma_0} + \Phi_0^*  \mathcal L_{\vec{T}} g(\vec{k}(0, x|_{\Sigma_0}), \cdot ) \nonumber\\
&=& g(\vec{k}_1 , \vec{e}_a)|_{\Sigma_0} + g\left(\mathcal L_{Q\vec{l}}\vec{k}(x, 0)|_{\Sigma_0} , \vec{e}_a\right)  = 	g(\vec v, \vec{e}_a ), \label{geax}
		\end{eqnarray*}
		where the right hand side is found by application of Lemma \ref{lemma_marc_perturbations} and the defining properties of the extensions (and their corresponding difference vector) listed above. Moreover, the left hand side vanishes because of property (ii) of $\vec{k}_\varepsilon$. Following a similar strategy we find that $0=\partial_\varepsilon \left.\left( \Phi_\varepsilon^* g_\varepsilon(\vec k_\varepsilon, \vec{l}_\varepsilon)\right) \right|_{\varepsilon = 0} = \Phi_0^*(g(\vec{l},\vec{v})) $.
		The fact that these inner products between $\vec{v}$ and $\vec{l}$ and $\vec{e}_a$  vanish simultaneously implies that $\vec{v} = 0$. In fact if we write them in terms of the decomposition $\vec v \equiv \mathcal X \vec{l} + \mathcal X^a \vec{e}_a$ we see that
		\begin{equation*}
		\begin{bmatrix}
		h_{ab}   & l_a \\
		l_b & l^{(2)}
		\end{bmatrix}
		\begin{bmatrix}
		\mathcal X^b \\
		\mathcal X
		\end{bmatrix}
		= 0,
		\end{equation*}
		and since the coefficient matrix is non-degenerate, we conclude that the vector $\vec v$ is zero.
		
		\hfill $\blacksquare$
	}
\end{proof}

\begin{corollary}
	The scalar $\delta l^{(2)}$, the one form $\delta l_a$ and the two covariant symmetric tensor $\delta Y_{ab}$ on $\abshyp$ are independent of the extension $\vec{L}_\varepsilon$ chosen to perform the linearisation. \label{lemma_independent_extensions}
\end{corollary}
There is still an explicit dependence on the vector field $\vec{\zeta}$ in the relevant formulas for $\delta l^{(2)}$, $\delta l_a$ and $\delta Y_{ab}$, which is in essence a manifestation of the freedom in the choice of the rigging in the perturbative setting, as we see next.
In the perturbation scheme, we consider  the setting for the riggings developed in Section \ref{section_constuction}. A transformation of the riggings at each $\Sigma_\varepsilon$, including the background, leads to a new family of riggings $\vec{l}_\varepsilon{}'$ so that 
\begin{eqnarray}
&&\vec{l}_\varepsilon' = \lambda_\varepsilon (\vec{l}_\varepsilon + \vec{v}_\varepsilon),\quad \lambda_\varepsilon|_{\Sigma_\varepsilon} \neq 0,\quad \ndown (\vec v_\varepsilon)|_{\Sigma_\varepsilon} = 0. \label{rigging_vector_transformation_family}
\end{eqnarray}
We assume that both riggings $\vec{l}_\varepsilon$ and $\vec{l}_\varepsilon{}'$ are smooth in the sense of point 2 in Definition \ref{smooth-ext}. We single out the path generated by the embeddings $\Phi_\varepsilon$ by fixing $p \in \abshyp$ and letting $\varepsilon$ run, and the decomposition (\ref{rigging_vector_transformation_family}) necessarily requires that the functions $\varepsilon \mapsto \lambda_\varepsilon|_{\Phi_\varepsilon(p)}$ and $\varepsilon \mapsto \vec{v}_\varepsilon|_{\Phi_\varepsilon(p)}$ are $C^1$. Since the embeddings are also $C^1$ the  objects $\overline{\lambda}_\varepsilon := \lambda_\varepsilon \circ \Phi_\varepsilon$ and $ \overline{v}_\varepsilon := d\Phi_\varepsilon^{-1} (\vec{v}_\varepsilon|_{\Phi_\varepsilon(p)})$ are $C^1$ and their limits as $\varepsilon \rightarrow 0$ are $\lambda$ and $\overline{v}$ respectively, i.e. the gauge fields encoding the rigging transformation in the background. 
The property that $\vec{v}_\varepsilon = d\Phi_\varepsilon (\overline{v}_\varepsilon)$, for some $\overline{v}_\varepsilon$ in $T(\abshyp)$, allows us to relate $\{\boldsymbol{l}_\varepsilon{}', l_\varepsilon^{(2)}{}'\}$ with the non transformed rigging data in the following way:
	\begin{eqnarray}
\boldsymbol{l}_\varepsilon{}' &=& \lambda_\varepsilon (\boldsymbol{l}_\varepsilon + h_\varepsilon(\overline{v}_\varepsilon,\cdot)), \label{l_epsilon_form_prime}\\
l_\varepsilon^{(2)}{}' &=& \lambda_\varepsilon^2 (l_\varepsilon^{(2)} + 2 \boldsymbol{l}_\varepsilon (\overline{v}_\varepsilon) + h_\varepsilon(\overline{v}_\varepsilon, \overline{v}_\varepsilon)). \label{l_epsilon_scalar_prime}
\end{eqnarray}
It is natural to decompose this vector in the basis $\{n^a , P^{ab}\}$, which we write as
\begin{equation*}
\overline{v}^a_\varepsilon = W(\varepsilon) n^a + Z_b(\varepsilon)P^{ab}, \quad \text{with } \quad n^{(2)}W(\varepsilon) + n^a Z_a(\varepsilon) = 0.
\end{equation*} 
\begin{lemma} \label{lemma_rigging_transformations_fo}
	Consider hypersurface metric data $\{\abshyp, h_{ab}, l_a, l^{(2)}\}$ and rigging transformation gauge fields $\{\lambda, W, Z_a\}$ that satisfy $n^{(2)}W + n^a Z_a = 0$, so that these uniquely define $\overline v^a \equiv P^{ab}Z_b + n^a W$.
	
	The freedom related to rigging transformations to first order is encoded in the fields $\{\delta \lambda, \delta W, \delta Z_a\}$ satisfying the constraint
	%	\begin{equation}
	%	n^a \delta Z_a + n^{(2)} \delta W = \delta h_{ab} n^a v^b + n^{(2)} v^a \delta l_a \label{linearized_constraint_tangent_vector}
	%	\end{equation}
	\begin{equation}
	n^a \delta Z_a + n^{(2)} \delta W = 0, \label{linearized_constraint_tangent_vector}
	\end{equation}
	so that these uniquely define the vector field
	\begin{eqnarray*}
	\delta \overline v^a \equiv P^{ab}\delta Z_b + n^a \delta W .
	\end{eqnarray*}
	
	The gauge transformed first order perturbations of the rigging data in terms of the background gauge fields $\{\lambda, W, Z\}$ and the first order gauge fields $\{\delta \lambda, \delta W, \delta Z \}$ reads
	\begin{eqnarray}
	h_{ab}' &=& h_{ab},\nonumber\\
	\delta l_a' &=& \delta \lambda (l_a + Z_a) + \lambda (\delta l_a + \delta Z_a + \delta h_{ab} \overline v^b), \label{rigging_transformed_deltala}\\
	\delta l^{(2)}{}' &=& 2 \lambda \delta \lambda (l^{(2)} + 2W + Z_a \overline v^a) + \lambda^2 (\delta l^{(2)} + 2 \delta W + 2 \overline v^a \delta Z_a +  \overline v^a (2\delta l_a + \overline v^b \delta h_{ab})), \label{rigging_transformed_deltal2}\\
	\delta Y_{ab}' &=& \lambda \delta Y_{ab} + \delta \lambda Y_{ab} +\frac{1}{2} \left( \delta l_a \partial_b \lambda + \delta l_b \partial_a \lambda + l_a \partial_b \delta \lambda + l_b \partial_a \delta \lambda\right) + \frac{1}{2} \mathcal L_{\delta \lambda \overline v + \lambda \delta \overline v} h. \label{rigging_transformed_deltalY}
	\end{eqnarray}
\end{lemma}
\begin{proof}
	The first fundamental form does not depend on the rigging, and therefore it does not change under rigging transformations.
	
	The first $\varepsilon$-derivative of $\overline{v}_\varepsilon$ at $\varepsilon = 0$ defines the perturbed fields
	\begin{equation}
	\delta \overline v^a = \delta W n^a + \delta Z_b P^{ab}, \quad \text{with } \quad \delta W n^{(2)} + \delta Z_a n^a = 0, \label{deltav_gaugefields}
	\end{equation}
and by taking $\varepsilon$ derivatives, this time in (\ref{l_epsilon_form_prime}) and (\ref{l_epsilon_scalar_prime}) we obtain the expressions
\begin{eqnarray}
\delta l_a {}' &\equiv& \left.\frac{d}{d \varepsilon}	\boldsymbol{l}_\varepsilon{}'\right|_{\varepsilon=0} = \delta \lambda (l_a + Z_a) + \lambda \left(\delta l_a + \left. \frac{d}{d \varepsilon} h_\varepsilon(\overline{v}_\varepsilon,\cdot)\right|_{\varepsilon=0}\right),\label{l_epsilon_form_prime_1}\\
\delta l^{(2)}{}' &\equiv& \left.\frac{d}{d \varepsilon}	l_\varepsilon^{(2)}{}'\right|_{\varepsilon=0} = 2 \lambda \delta \lambda (l^{(2)} + 2W + Z_a \overline v^a) + \lambda \left( \delta l^{(2)} + 2 \left.  \frac{d}{d \varepsilon}\boldsymbol{l}_\varepsilon (\overline{v}_\varepsilon)   \right|_{\varepsilon=0} + \left.  \frac{d}{d \varepsilon}h_\varepsilon(\overline{v}_\varepsilon,\overline{v}_\varepsilon)   \right|_{\varepsilon=0}\right) \nonumber\\
\label{l_epsilon_scalar_prime_1}
\end{eqnarray}
The three terms involving the vector $\overline{v}_\varepsilon$ can be written in terms of the gauge fields as follows
\begin{eqnarray*}
\left. \frac{d}{d \varepsilon} h_\varepsilon(\overline{v}_\varepsilon,\cdot)\right|_{\varepsilon=0} &=& \delta h_{ab} \overline v^b + h_{ab} \delta \overline v^b = \delta h_{ab} \overline v^b + \delta Z_a,\\
%%%%%%%%%%%%%%%%%%%%%%%%%%%%%%%%
\left.  \frac{d}{d \varepsilon}\boldsymbol{l}_\varepsilon (\overline{v}_\varepsilon)   \right|_{\varepsilon=0} &=& \delta l_a \overline v^a + l_a \delta \overline v^a = \overline v^a \delta l_a + \delta W, \\
%%%%%%%%%%%%%%%%%%%%%%%%%%%%%%%
\left.  \frac{d}{d \varepsilon}h_\varepsilon(\overline{v}_\varepsilon,\overline{v}_\varepsilon)   \right|_{\varepsilon=0} &=& \delta h_{ab} \overline v^a \overline v^b + 2 h_{ab} \overline v^a \delta \overline v^b = \delta h_{ab} \overline v^a \overline v^b + 2 \overline v^a \delta Z_a,
\end{eqnarray*}
where we have made use of the decomposition of $\delta \overline v$ (\ref{deltav_gaugefields}) and the relations for general hypersurfaces (\ref{constraints_definitions}). Taking these expressions back to (\ref{l_epsilon_form_prime_1}) and (\ref{l_epsilon_scalar_prime_1}) we find  (\ref{rigging_transformed_deltala}) and (\ref{rigging_transformed_deltal2}).	Finally we explore the effect of these transformations in the tensor $Y_\varepsilon$, that can be studied considering the tensor $Y_\varepsilon'$ relative to the rigging $\vec{l}_\varepsilon{}'$, related to $\vec{l}_\varepsilon$ through (\ref{rigging_vector_transformation_family}). This consideration allows us to relate $Y_\varepsilon$ and $Y_\varepsilon'$ as follows
	\begin{eqnarray*}
	Y_\varepsilon{}' = \left.\frac{1}{2} \Phi_\varepsilon^* \left( \mathcal L_{\vec{l}_\varepsilon{}'}g_\varepsilon\right)\right|_{\varepsilon=0} = \lambda_\varepsilon Y_\varepsilon + \frac{1}{2} \left(  {\boldsymbol{l}}_\varepsilon \otimes d\lambda_\varepsilon + d \lambda_\varepsilon \otimes  {\boldsymbol{l}}_\varepsilon \right)+ \frac{1}{2}\mathcal L_{\overline{\lambda}_\varepsilon \overline{v}_\varepsilon}h_\varepsilon,
	\end{eqnarray*}
	where we have used the properties $\Phi_\varepsilon^* d\lambda_\varepsilon = d\Phi_\varepsilon^* \lambda_\varepsilon$ and $\Phi_\varepsilon^* \mathcal L_{d\Phi_\varepsilon(\overline \lambda_\varepsilon \overline v_\varepsilon)} g_\varepsilon = \mathcal L_{\overline{\lambda}_\varepsilon \overline{v}_\varepsilon} \Phi_\varepsilon^* g_\varepsilon$.
	
	The first derivative with respect to $\varepsilon$ of this relation provides the expression (\ref{rigging_transformed_deltalY}) for $\delta Y'$ .
	
	\hfill $\blacksquare$
\end{proof}

\begin{remark}
	There is a slightly different, but equivalent, procedure to study the gauge freedom. Recall transformations (\ref{rigging_transformed_la}) and (\ref{rigging_transformed_l2}). Since they involve functions and one forms in $\abshyp$, one could consider their generalization to a one-parameter family of scalars and one forms and take  $\varepsilon$ derivatives directly from there. It can be shown that this method leads to the same expressions from Lemma \ref{lemma_rigging_transformations_fo} after the substitutions $\delta W \rightarrow \delta W + \overline{v}^c \delta l_c$ and $\delta Z_a \rightarrow \delta Z_a + \overline v^c \delta h_{ac}$. In other words, this alternative method would correspond to a decomposition  $\delta \overline{v}^a = \delta P^{ab}Z_b + P^{ab} \delta Z_b + n^a \delta W + \delta n^a W$. 

\end{remark}

\begin{lemma}
	\label{lemma_rigging_transformation_perturbed_inverse_data}
	Let $\{\lambda, W, Z_a\}$ and $\{\delta \lambda, \delta W, \delta Z_a\}$ be background and first order gauge fields respectively, that generate the vector field $\delta v^a \equiv \delta W n^a + \delta Z_b P^{ab}$. Under their action the objects $\{\delta n^{(2)}, \delta n^a, \delta P^{ab}\}$ change according to
	\begin{eqnarray}
\delta n^{(2)}{}' &=& -\frac{2}{\lambda^3} n^{(2)} \delta \lambda + \frac{\delta n^{(2)}}{\lambda^2} = -2\frac{\delta \lambda}{\lambda} n^{(2)}{}' + \frac{\delta n^{(2)}}{\lambda^2}, \label{delta_n2_gauged}\\
\delta n^a{}' &=& -\frac{\delta \lambda}{\lambda} n^a{}' + \frac{1}{\lambda} (\delta n^a - n^{(2)} \delta \overline{v}^a - \overline{v}^a \delta n^{(2)}), \nonumber\\
\delta P^{ab}{}' &=& \delta P^{ab} + 2 n^{(2)} \delta \overline{v}^{(a} \overline{v}^{b)} - 2 n^{(a} \delta \overline{v}^{b)} + \overline{v}^a \overline{v}^b \delta n^{(2)} - 2 \overline{v}^{(a} \delta n^{b)}. \nonumber
	\end{eqnarray}
\end{lemma}

The transformations in this lemma are trivially found by taking derivatives in the background transformations (\ref{inverse_data_transformations}), and another way leading to the same result consist of plugging the gauge transformed expressions from Lemma \ref{lemma_rigging_transformations_fo} into the first order expressions from Lemma \ref{lemma_inverse_perturbed_data}. 

\begin{corollary}
	The sign of $\delta n^{(2)}$ at null points is independent of the rigging transformations.
\end{corollary}

In analogy with the exact theory, we might ask ourselves whether the object $n^{(2)}_\varepsilon$ encodes information about the causal character of the hypersurfaces $\Sigma_\varepsilon$, at least to first order in perturbation theory. Recall that we have endowed a causal character to $\abshyp$ by the identification with the background hypersurface $\Sigma_0$, given the embedding $\Phi_0$.

The function $n_\varepsilon^{(2)}$ is a function defined on $\abshyp$ in terms of the $\varepsilon$-family of \textit{hypersurface metric data} (see the linearized version (\ref{delta_n2})) and therefore it depends on the embeddings $\Phi_\varepsilon$. The sign of $n_\varepsilon^{(2)}$ would store information about the signature of $\Sigma_\varepsilon$ as follows: given a fixed point $p\in \abshyp$ we let it propagate in $\varepsilon$ through $\Phi_\varepsilon$. This produces a collection of points $\{p_\varepsilon \equiv \Phi_\varepsilon (p) \in \Sigma_\varepsilon\}$ whose corresponding causal characters, with respect to the metrics $g_\varepsilon$, are compared. Note that a hypersurface gauge transformation understood as an internal $\varepsilon$-dependent diffeomorphism $\chi_\varepsilon$ in $\abshyp$, changes the former identification and it results into a different embedding $\Phi^{(h)}_\varepsilon = \Phi_\varepsilon \circ \chi_\varepsilon$. The (background) causal character of a point $p\in \abshyp$ is unaffected by these transformations, because we have restricted hypersurface gauge transformations to $\chi_0 = \mathbb{I}$, thus keeping the background embedding $\Phi_0$ fixed. But apart from this constraint, the identification between $\abshyp$ and the $\Sigma_\varepsilon \subset \mathcal M$ is an inherent freedom of the method, thus completely arbitrary. Because $n_\varepsilon^{(2)}$ is sensitive to the identification, the signature change between the different elements in $\{\Sigma_\varepsilon\}$ tracked in this way is relative to the hypersurface gauge, i.e. to the pair $\{\abshyp, \Phi_\varepsilon \}$.

At non-null points in $\abshyp$ the causal character of the hypersurfaces in the sense above will not change as a result of the perturbation. This claim is supported by the fact that at $p \in \abshyp$, $n_\varepsilon^{(2)}$ depends continuously on $\varepsilon$, with limiting value $n^{(2)}$. Thus if $\abshyp$ is timelike (respectively, spacelike) at $p$, then $n^{(2)}|_p >0$ and so $n^{(2)}_\varepsilon|_p >0$ for all sufficiently small $\varepsilon$. 

At null points, in turn,  $n_\varepsilon^{(2)} (p, \varepsilon =0) = 0$ and this function could depart from zero as $\varepsilon$ varies, attaining either positive or negative values. This would be understood as a signature change in the scheme $\{\abshyp, \Phi_\varepsilon\}$. If the null points are isolated then the function $n_\varepsilon^{(2)}$ does not provide any information due to the gauge dependence of the method, but for an open set of null points  $\mathcal O \subset \abshyp$ the perturbation of $n^{(2)}$ is indeed meaningful: if we consider a point $p \in \mathcal O \subset \abshyp$, then the point $\chi_\varepsilon(p)$ will also belong to $\mathcal O$ for $\varepsilon$ small enough. In fact the hypersurface gauge transformation of $\delta n^{(2)}$ at null points results to be
	\begin{equation}
\delta n^{(2)}{}^{(h)} = \delta n^{(2)} + \mathcal L_{\vec u_\abshyp} n^{(2)}, \label{delta_n2_hypersurfacegaugetransformation}
\end{equation} 
which in this case simplifies just to $\delta n^{(2)}{}^{(h)} = \delta n^{(2)}$. Therefore the sign of $\delta n^{(2)}$ is completely gauge independent (under rigging and hypersurface gauge transformations) inside open sets of null points in $\abshyp$ (if any), providing a notion signature change for this case. This observation may result interesting in order to study first order perturbations of null hypersurfaces, as we shall see in Appendix \ref{appendix_constant_signature}.

Still, we are left with expressions of the perturbations of the \textit{hypersurface data} that depend on the vector $\vec{\zeta}$, which does not have a clear interpretation in practical problems.
By this, we mean that there are situations where we know, or assume, a particular behaviour of the family $\Sigma_\varepsilon$. For instance a common approach in the literature consist of exploiting the fact that their causal character is constant (see the discussion in Appendix \ref{appendix_constant_signature}).
 This piece of information that we may have a priori is difficult to incorporate directly in $\vec \zeta$, and instead, it might be easier to deal with the data. As we show next, it is possible to fully encode $\vec \zeta$ in the objects $\delta l^{(2)}$ and $\delta l_a$. In fact, expressions (\ref{l12_b}) and (\ref{l1a_b}) contain the decomposition of the vector field $\vec{s}_\abshyp$ in terms of $\{n^a, P^{ab}\}$ through the combinations $\mathcal S := s^al_a$ and $\mathcal S_a := h_{ab}s^b$. These read explicitly
 \begin{eqnarray}
 \mathcal S &=& \frac{\delta l^{(2)}}{2} - (\alpha l^{(2)} + Q l_\mu a^\mu) - \frac{1}{2} g_1 (\vec l, \vec l) - \frac{1}{2}\vec{T}_\abshyp (l^{(2)}), \nonumber\\
 \mathcal S_a &=& \delta l_a - (\alpha l_a + Q e_a^\mu a_\mu) - g_1 (\vec l, \vec e_a) - \frac{Q}{2}\partial_a l^{(2)} - l^{(2)} \partial_a Q - \LL_{\vec T_{\abshyp}}l_a.\nonumber
 \end{eqnarray}
 It is a matter of applying Lemma \ref{lemma_vector_reconstruction} in order to reconstruct the vector $\vec{s}_\abshyp$. On the one hand, from the constraint $n^{(2)} \mathcal S + n^a \mathcal S_a = 0$ we obtain 
 \begin{eqnarray}
&&n^a \delta l_a + \frac{n^{(2)}}{2} \delta l^{(2)} - \alpha  -Q a_\mu n^\mu   - \frac{n^{(2)}}{2} g_1 (\vec{l},\vec{l}) -n^a g_1(\vec{l},\vec{e}_a) - l^{(2)}n^a \partial_a Q - n^a \LL_{\vec{T}_\abshyp}l_a \nonumber\\
 && - \frac{n^{(2)}}{2} \LL_{\vec{T}_\abshyp} l^{(2)} -\frac{Q}{2} n^a \vec{e}_a (l^{(2)}) =0,\nonumber
 \end{eqnarray}
 where we used the decomposition of the normal vector $n^\mu = n^{(2)} l^\mu + n^a e_a^\mu$ and definitions (\ref{constraints_definitions}). This result leads to
\begin{eqnarray}
\boldsymbol{n}(\vec{\zeta}) = \alpha + Q (n_\mu a^\mu)&=& n^a \delta l_a + \frac{n^{(2)}}{2} \delta l^{(2)} - \frac{n^{(2)}}{2} g_1 (\vec{l},\vec{l}) -n^a g_1(\vec{l},\vec{e}_a) - l^{(2)}n^a \partial_a Q - n^a \LL_{\vec{T}_\abshyp}l_a \nonumber\\
& -& \frac{n^{(2)}}{2} \LL_{\vec{T}_\abshyp} l^{(2)} -\frac{Q}{2} n^a \vec{e}_a (l^{(2)})  \label{L_deltas2}. \label{zeta_normal}
\end{eqnarray}
On the other hand,  the vector itself is found to be
\begin{eqnarray}
s_\abshyp^a = P^{ab} \mathcal S_b + n^a \mathcal S &=& P^{ab} \delta l_b + \frac{n^a}{2} \delta l^{(2)} -\alpha (n^a l^{(2)} + P^{ab} l_b) - Q(n^a l_\mu a^\mu +P^{ab}a_\mu e_b^\mu) -\frac{n^a}{2} g_1 (\vec{l},\vec{l}) \nonumber\\
&& - P^{ab} g_1(\vec{l},\vec{e}_b) - \frac{n^a}{2} \vec{T}{}_\abshyp (l^{(2)}) - P^{ab} \LL_{\vec{T}_\abshyp}l_b - l^{(2)} P^{ab}\partial_b Q - \frac{Q}{2}P^{ab}\vec{e}_b (l^{(2)}).\nonumber
\end{eqnarray}
We use the identity $P^{ab}e_b^\mu = g^{\rho \mu}|_{\Sigma_0} \omega^a_\rho - n^a l^\mu$ and definitions (\ref{constraints_definitions}) to find the remaining projection of $\vec \zeta$ as follows
\begin{eqnarray}
\boldsymbol{\omega}^a (\vec{\zeta}) = \vec{s}{}_\abshyp + Q(\omega^a_\mu a^\mu) &=& P^{ab} \mathcal{S}_b + \mathcal{S}n^a + Q(\omega^a_\mu a^\mu) = P^{ab} \delta l_b + \frac{n^a}{2} \delta l^{(2)} - P^{ab} g_1(\vec{l},\vec{e}_b) \nonumber \\
&-& \frac{n^a}{2} g_1 (\vec{l},\vec{l}) - l^{(2)} P^{ab}\partial_b Q - P^{ab} \LL_{\vec{T}_\abshyp}l_b - \frac{n^a}{2} \vec{T}{}_\abshyp (l^{(2)}) - \frac{Q}{2}P^{ab}\vec{e}_b (l^{(2)}). \nonumber \\\label{s_deltas}
\end{eqnarray}

These two expressions above relate algebraically the rigged and tangential components of $\vec{\zeta}$ with the first order perturbations of the rigging data $\{\delta l_a, \delta l^{(2)} \}$. Therefore the problem of characterizing the vector field $\vec{\zeta}$ has been put into the perturbations of the rigging data, whose dependence on the gauge freedom is well understood.  
\begin{proposition} \label{proposition_gauge_inversion}
	Consider a rigging transformation driven by the fields $\{\lambda=1, W=Z_a = 0, \delta \lambda, \delta W, \delta Z_a\}$. The equations for the transformed first order rigging data $\{\delta l_a{}', \delta l^{(2)}{}'\}$ in terms of $\{\delta \lambda, \delta W, \delta Z_a \}$ are invertible, so that
	\begin{eqnarray*}
	\delta \lambda &=& \frac{n^{(2)}}{2} \Delta \delta l^{(2)} + n^a \Delta \delta l_a, \label{delta_lambda_gauge_removal} \\
	\delta W &=& \frac{1}{2}(1-n^{(2)}l^{(2)}) \Delta \delta l^{(2)} - l^{(2)} n^a \Delta \delta l_a, \label{delta_W_gauge_removal}\\
	\delta Z_a &=& (\delta_a^b-n^bl_a)\Delta \delta l_b - \frac{n^{(2)}}{2}l_a \Delta \delta l^{(2)}, \label{delta_Za_gauge_removal}
	\end{eqnarray*}
	where $\Delta \delta l_a \equiv \delta l_a' - \delta l_a$ and $\Delta \delta l^{(2)} \equiv \delta l^{(2)}{}' -  \delta l^{(2)}$.
\end{proposition}
\begin{proof}
	
	Starting from Lemma \ref{lemma_rigging_transformations_fo} we apply a pure first order rigging transformation to find the expressions
	\begin{equation*}
	\Delta \delta l^{(2)}  = 2 l^{(2)} \delta \lambda + 2 \delta W, \quad \Delta \delta l_a = l_a \delta \lambda + \delta Z_a.
	\end{equation*}
	
	The constraint $n^{(2)} \delta W + n^a \delta Z_a = 0$ solves for $\delta \lambda$ and direct substitution of the result in the equations above provides the expressions for $\delta W$ and $\delta Z_a$.
	
	\hfill $\blacksquare$
\end{proof}
The consequence of this result is that there is a direct correspondence between the first order perturbations of the rigging data $\{\delta l^{(2)}, \delta l_a\}$ and the gauge fields $\{\delta \lambda, \delta W, \delta Z_a \}$ constrained by (\ref{linearized_constraint_tangent_vector}). Therefore, it is possible to specify these perturbations conveniently. 
\begin{corollary}
		There exists a first order rigging transformation such that $\delta l^{(2)} = \delta l_a = 0$, regardless of the causal character of $\abshyp$.
\end{corollary}

This is a very convenient gauge in order to perform calculations in specific problems. Note however that in this gauge the vector $\vec{\zeta}$ is nonzero, and it can be understood as a first order perturbation of the rigging vector for a suitable choice of the extension (for instance a geodesic one). Nevertheless if we sit on the gauge $\delta l_a = \delta l^{(2)} = 0$, it is a matter of applying the formulas of Proposition \ref{proposition_gauge_inversion} to identify the gauge fields that would take us to a gauge where $\delta l_a{}'$ and $\delta l^{(2)}{}'$ are nonvanishing but $\vec{\zeta} = 0$. These are found to be $\delta \lambda = - n_\alpha \zeta^\alpha$, $\delta W = -l_a \omega^a_\alpha \zeta^\alpha$ and $\delta Z_a = -h_{ab} \omega^b_\alpha \zeta^\alpha$, and plugged into relation  \ref{rigging_transformed_deltalY} we find the  the relation $\delta Y' = \delta Y - \frac{1}{2}\Phi^*\mathcal L_{\vec{\zeta}} g$.

\section{Matching conditions}
There are many works covering this topic in the literature, see for instance the classical references \cite{Israel1966} for a presentation of timelike shells in the matching context or \cite{ClarkeDray} for null hypersurfaces. However we follow the formalism \cite{Mars1993}, \cite{Mars_constraints} where the theory for the matching conditions across hypersurfaces of arbitrary causal character was developed.
 
\subsection{Review of the exact matching conditions} 
 
 We consider two spacetimes with $C^3$ oriented boundary $(\mathcal M^+, g^+, \partial \mathcal M^+)$ and $(\mathcal M^-, g^-, \partial \mathcal M^-)$. The matching procedure allows for the generation of a matched spacetime $(\mathcal M, g)$, defined as the disjoint union of the spacetimes $\pm$. It contains a hypersurface $\Sigma_0$ which separates it into the two regions $\pm$ and has a metric $g$ which is well defined everywhere, is continuous at points of $\Sigma_0$ and agrees with $g^\pm$ in the respective regions.
 
 This procedure requires, first of all, that the boundaries $\partial \mathcal M^+$ and $\partial \mathcal M^-$ must be diffeomorphic to each other, and in particular to an abstract hypersurface $\abshyp$, so that if we consider the pair of embeddings $\Phi^\pm:\abshyp \rightarrow \partial \mathcal M^\pm$, this requisite becomes $\Phi^+ (\abshyp) = \Phi^- (\abshyp)$. If this condition is fulfilled we will refer simply to the embedding $\Phi: \abshyp \rightarrow \Sigma_0 \subset \mathcal M$. The boundaries have been identified pointwise, but this is not enough in order to have a well defined geometry at points of $\Sigma_0 \subset \mathcal M$. This is achieved by identifying the tangent spaces at points of $\Phi^\pm (\abshyp)$, and it requires two steps. Consider the two, a priori different, inherited first fundamental forms $h^\pm := \Phi^\pm{}^* (g^\pm)$. The preliminary matching conditions demand that
 \begin{equation*}
 h^+ = h^- \equiv h, \label{matching_conditions_fff}
 \end{equation*}
 and they ensure that the subspace of vectors which are tangent to $\Sigma_0 \subset \mathcal M$ is well defined. 
In addition, the matching procedure is completed by selecting vectors $\vec l{}^\pm$ transverse to $\Phi^\pm (\abshyp)\subset \mathcal M^\pm$ at every point and identifying them, which translates into the following equations on $\abshyp$
\begin{equation}
\Phi^+{}^* (g^+ (\vec{l}{}^+, \cdot)) = \Phi^-{}^* (g^- (\vec{l}{}^-, \cdot)), \qquad \Phi^+{}^* (g^+ (\vec{l}{}^+, \vec{l}{}^+)) = \Phi^-{}^* (g^- (\vec{l}{}^-, \vec{l}{}^-)). \label{matching_conditions_riggings}
\end{equation}
 Also, the relative orientation of the rigging vectors must be appropriate so that these can be completely identified. Hence, the subspaces of vectors transverse to $\Phi^\pm (\abshyp)$ have been identified, so that the full tangent space at points of $\Sigma_0 \subset \mathcal M$ is well defined.  This spacetime approach for the exact matching conditions suffices for our purpose of formulating them in perturbation theory, and it is the point of view we take in Section \ref{subsection_perturbed_matching}.
 
 The inherent freedom in the method to characterise the hypersurface due to the arbitrariness in the choice of the rigging vector is relevant in the context of matching of spacetimes, and it has been studied extensively in \cite{lorentzian2007}. Nonetheless, the whole discussion therein is based on the spacetime perspective, which is not convenient for the particular issue of describing the freedom in the rigging to first order in perturbation theory.  In the first part of the section we revisit this matter in exact matchings and reproduce some conclusions from  \cite{lorentzian2007} from a point of view which is closer to the data approach. The matching conditions formulated within this approach can be found in Theorem 3 in \cite{Mars2005}, where the main underlying idea is that there is \textit{hypersurface metric data} $\{\abshyp, h_{ab}, l^{(2)}, l_a\}$ that can be embedded into the two different spacetimes $(\mathcal M^\pm, g^\pm, \partial \mathcal M^\pm)$ through some embeddings $\Phi^\pm$ and riggings $\vec{l}{}^\pm$ satisfying that $\Phi^\pm (\abshyp) = \partial M^\pm$ and that the orientation of such riggings is compatible.

For clarity we will consider that the preliminary matching equations are $h_{ab}^+ = h_{ab}^-$, we refer to $\{l_a^+, l^{(2)}{}^+\}= \{l_a^-, l^{(2)}{}^-\}$ as the rigging compatibility conditions, and the (full) matching conditions are $\{h_{ab}^+, l_a^+, l^{(2)}{}^+\}= \{ h_{ab}^-, l_a^-, l^{(2)}{}^-\}$. If the matching conditions are satisfied, we will refer to these \textit{hypersurface metric data} sets simply by $  \{\abshyp, h_{ab}, l_a, l^{(2)}\} $.
We also note that the matching conditions immediately imply that $n^{(2)}{}^+ = n^{(2)}{}^-$, $n^a{}^+ = n^a{}^-$ and $P^{ab}{}^+ = P^{ab}{}^-$.

The matching conditions introduced so far are purely geometric, in the sense that they are necessary in order to have a well defined geometry at points of $\Sigma_0$ and in particular, they ensure that the metric is continuous across $\Sigma_0$. For this reason, tensor fields which are constructed taking derivatives of the metric tensor must be defined in this context as tensor distributions \cite{Mars1993}. As a consequence of the matching conditions explained above, the Riemann, Ricci and Einstein tensor distributions acquire a singular part with support in the matching hypersurface. Therefore, the Einstein field equations in the distributional sense induce a singular part in the energy momentum tensor distribution with support in the matching hypersurface. This is known as a shell in the literature. 
These singular parts are related to the extrinsic properties of $\Sigma_0$ and a detailed analysis provides the following result:
\begin{theorem}{{\bf (adapted from Theorems 6 and 7 in Mars and Senovilla 1993 \cite{Mars1993})}}
	\footnote{The original version of the theorems in \cite{Mars1993} are stated in terms of the two covariant tensor field $\mathcal H_{ab} \equiv e_a^\alpha e_b^\beta \nabla_\alpha l_\beta $, which is not symmetric by itself but the object $[\mathcal H_{ab}]$ can be shown to be symmetric. Taking into account that by definition $Y_{ab} = (1/2)(\mathcal H_{ab} + \mathcal H_{ba})$, we see that $[Y_{ab}] = [\mathcal H_{ab}]$.}
		\begin{enumerate}
			\item 	At a point $p \in \Sigma_0$, the singular part of the {\bf Riemann} tensor distribution vanishes if and only if $[Y_{ab}] = 0$. 
			
			\item At a point $p \in \Sigma_0$ where the hypersurface is not null, the singular part of the {\bf Ricci} tensor distribution vanishes if and only if $[Y_{ab}] = 0$. 
			
			At a point $p \in \Sigma_0$ where the hypersurface is null, the singular part of the {\bf Ricci} tensor distribution vanishes if and only if $n^a[Y_{ab}] = 0$ and $P^{ab}[Y_{ab}] = 0$. 
			
			\item The singular part of the {\bf energy momentum} tensor distribution vanishes if and only if so does the singular part of the Ricci tensor distribution.
			
		\end{enumerate}
\label{theorem_singular_distributions_background}
\end{theorem}

It is out of the scope of this paper to extend a distributional approach to perturbation theory, but the approach \cite{Mars_constraints} puts forward an alternative method to construct the shell that depends exclusively on the \textit{hypersurface data} (or derived objects), without invoking a distributional approach. It is shown then that this energy momentum tensor not only satisfies plausible properties expected of a shell, but also that agrees with the well known cases of (everywhere) timelike/spacelike and null shells previously studied in the literature. Its expression in terms of the \textit{hypersurface data} (or derived objects) is \cite{Mars_constraints}
\begin{equation}
\tau^{ab} := \left \lbrace (n^a P^{bc} + n^b P^{ac}) n^d -(n^{(2)}P^{ac}P^{bd} + P^{ab} n^c n^d) + (n^{(2)}P^{ab} -n^a n^b)P^{cd} \right \rbrace [Y_{cd}]. \label{definition_shell}
\end{equation}
In order to avoid a trivial $\tau^{ab}$ we assume from now on that the dimension of $\abshyp$ satisfies $\dimension > 1$.
It is often required that the matching procedure removes this shell, which is accomplished if and only if \cite{Mars_constraints}
\begin{equation}
n^{(2)} P^{ab} [Y_{ac}] = 0, \quad P^{ab}[Y_{ab}] = 0, \quad n^a [Y_{ab}] = 0. \label{eq_junction_exact}
\end{equation}
These conditions hold for any causal character, but note that when particularised to null and non-null points they become $(2)$ in Theorem \ref{theorem_singular_distributions_background}.
This set of conditions is usually known as \textit{junction conditions}, and this type of matchings are called \textit{proper matchings}. 

We revisit a result from \cite{lorentzian2007} about the behaviour of the matching conditions under rigging transformations, providing a different proof based in the \textit{ hypersurface data} approach and better adapted to the study of the perturbations in the next part of the section.

\begin{lemma} {\bf{(Mars, Senovilla, Vera 2007 \cite{lorentzian2007})}}\label{lemma_rigging_compatibility}
	Consider a pair of hypersurface metric data sets $\{\abshyp, h_{ab}^\pm, l_a^\pm, l^{(2)}{}^\pm\}$ satisfying the matching conditions, and a corresponding pair of  rigging transformations driven by $\{\lambda^\pm, \overline v^a{}^\pm\}$. Then
	\begin{itemize}
	\item At null points, the rigging compatibility conditions hold for $\{l_a{}'{}^\pm, l^{(2)}{}'{}^\pm\}$ if and only if $\lambda^+ = \lambda^-$ and $\overline v^a{}^+ = \overline v^a{}^-$. 
	
	\item At non-null points the rigging compatibility conditions for $\{l_a{}'{}^\pm, l^{(2)}{}'{}^\pm\}$, supplemented with a condition ensuring compatibility in the orientations of the riggings, holds if and only if $\lambda^+ = \lambda^-$ and $\overline v^a{}^+ = \overline v^a{}^-$.
	\end{itemize}
\end{lemma}
\begin{proof}
	Let us define $\tilde v^a := \lambda \overline v^a$ and consider the difference of the gauge transformed hypersurface metric data given in  (\ref{rigging_transformed_la}) and (\ref{rigging_transformed_l2})
	\begin{eqnarray}
	%\left[h_{ab}'\right] &=& [h_{ab}], \\
	\left[l_a' \right] &=& [\lambda] l_a + [\lambda Z_a], \label{riggings_1}\\
	\left[l^{(2)}{}'\right] &=& [\lambda^2] l^{(2)} + 2[\lambda^2 W] + P^{ab}[\lambda^2 Z_a Z_b] - n^{(2)} [\lambda^2 W^2], \label{riggings_2}
	\end{eqnarray}
	as well as the constraint $n^{(2)} W^\pm + n^a Z_a^\pm = 0$, which is valid for both sets of data.
	We use the identities for the difference of products between any two objects in $\abshyp$, $a$ and $b$
	\begin{equation}
	[ab] \equiv a^+ b^+ - a^- b^- = [a][b] + a^-[b] + b^-[a], \quad [a^2] = [a]([a]+2a^-). \label{brackts_products}
	\end{equation}
	First of all consider that $[\lambda] = [W] = [Z_a]=0$. It is clear from the two expressions above that $\left[l_a' \right] = \left[l^{(2)}{}'\right] = 0$.
	
	For the reverse implication, we impose $\left[l_a' \right] = \left[l^{(2)}{}'\right] = 0$, contract (\ref{riggings_1}) with $n^a$
	\begin{equation}
	[\lambda](1-n^{(2)}l^{(2)}) - n^{(2)}[\lambda W]= 0. \label{riggings1_1}
	\end{equation}
	We distinguish between null and non-null points, and start with the first possibility. 
	It is clear from the relation above that $[\lambda] = 0$. Then (\ref{riggings_1}) leads to $[\lambda Z_a] = \lambda[Z_a] = 0$, and since $\lambda \neq 0$ we conclude that $[Z_a] = 0$. The remaining equation (\ref{riggings_2}) yields $[W] = 0$. 
	
	Next we consider non-null points, and (\ref{riggings1_1}) and (\ref{riggings_1}) provide 
	\begin{equation}
	[\lambda W] = \frac{1}{n^{(2)}} \left(1-n^{(2)}l^{(2)} \right) [\lambda], \quad [\lambda Z_a] = -l_a [\lambda]. \label{riggings_1_2}
	\end{equation}
	We write the equation (\ref{riggings_2})$=0$ in terms of $[\lambda]$, $[\lambda W]$ and $[\lambda Z_a] $  using the identities (\ref{brackts_products}), and after a somewhat long calculation making use the two relations above we obtain that $[\lambda] (\lambda^+ + \lambda^-) = 0$. We rule out the second solution $\lambda^+ = -\lambda^-$, because it would change the relative orientation between the riggings $\vec{l}{}^+$ and $\vec{l}{}^-$ and therefore we are left with $[\lambda] = 0$. Equations (\ref{riggings_1_2}) lead directly to $[W] = [Z_a] = 0$.

		\hfill $\blacksquare$
\end{proof}

The result that we have obtained recovers Lemmas 2 and 3 from \cite{lorentzian2007}. In fact, given one of the riggings, say $\vec{l}^+$, the rigging compatibility conditions have a unique solution at null points, and a unique solution with the adequate orientation at non-null points.  This uniqueness property of the riggings has implications on the rigged fundamental form. Its change under a rigging transformation, addressed in (\ref{transformation_Y}),  directly leads to the following theorem.   
\begin{theorem} {\bf(Mars and Senovilla \cite{Mars1993})} \label{theorem_junction_Y_riggings}
	If the matching conditions are satisfied, the junction condition $[Y_{ab}]=0$ does not depend on the freedom in the choice of $\{l_a, l^{(2)}\}$.
\end{theorem}
This statement about the independence of the junction conditions under rigging transformations is valid also for null points, so that if $n^a[Y_{ab}] = 0$ and $P^{ab}[Y_{ab}] = 0$ hold at a null point, then $n^a{}'[Y_{ab}{}']=0$ and $P^{ab}{}'[Y_{ab}{}'] = 0$, as a consequence of the expressions (\ref{transformation_Y}) and (\ref{inverse_data_transformations}) that relate the primed and non-primed objects, the matching conditions and the compatibility conditions for the rigging transformations.

\subsection{Matching conditions to first order in perturbation theory}
\label{subsection_perturbed_matching}
The general argument to formulate the perturbed matching conditions works as follows. We assume that there is a matching scheme in the  pair of $\varepsilon$-families of spacetimes, satisfying the matching conditions at each $\varepsilon$ (we will refer to this as a matching \textit{upstairs}). Using the different mappings involved in the problem, this information is encoded in some objects at the hypersurfaces $\Sigma_\varepsilon^\pm$ in the corresponding background spacetimes. At this point, the formalism of perturbing hypersurfaces can be invoked and produces a set of perturbed matching conditions in the background matching hypersurface. We develop this procedure hereafter, and for clarity we include a diagram (Figure \ref{figure:per_picture_4}) displaying the setting.

We assume two families of spacetimes with boundary 
$(\mathcal{M}_\varepsilon^\pm,\hat{g}_\varepsilon^\pm,\hat{\Sigma}_\varepsilon^\pm)$, 
 matched across their respective boundaries $\hat{\Sigma}_\varepsilon^\pm$ for each $\varepsilon$,
so that the mappings $\hat{\Phi}_\varepsilon^\pm: \hat{\abshyp}_\varepsilon \rightarrow \hat{\Sigma}_\varepsilon^\pm$  diffeomorphically relate the boundaries with the abstract hypersurfaces $\hat{\abshyp}_\varepsilon$ where the corresponding first fundamental forms and rigging data  from each side are equated, $\hat{h}_\varepsilon^+=\hat{h}_\varepsilon^-$, $\vec{\hat l}_\varepsilon{}^+ = \vec{\hat l}_\varepsilon^-$.  As usual, $\varepsilon = 0$ singles out the background configuration, so that $\abshyp$ is the matching hypersurface of the background where $\hat{\abshyp}_0\equiv \abshyp$ and $h= \hat{h}_0^+=\hat{h}_0^-$.  This is the most natural construction for the discussion, although it is not absolutely necessary since the hypersurfaces $\abshyp$ and $\hat \abshyp_\varepsilon$ are diffeomorphic.

The  following step is then to construct the objects $h_\varepsilon^\pm$, $\boldsymbol{l}_\varepsilon^\pm$ and $l^{(2)}_\varepsilon{}^\pm$ on $\abshyp$,  from those objects that have been defined  on $\hat{\abshyp}_\varepsilon$. 
Because there is a matching \textit{upstairs}, we consider the embeddings $\hat \Phi_\varepsilon^+: \hat \abshyp_\varepsilon \rightarrow \hat \Sigma_\varepsilon^+$ and, and the perturbative setting involves the diffeomorphisms $\psi_\varepsilon^+: \mathcal M^+ \rightarrow \mathcal M_\varepsilon^+$ and $\phi_\varepsilon: \abshyp \rightarrow \hat{\abshyp}_\varepsilon$. For convenience we also define the two following mappings $\phi_\varepsilon^+ := \hat \Phi_\varepsilon^+ \circ \phi_\varepsilon: \abshyp \rightarrow \hat{\Sigma}_\varepsilon^+$ and $\Phi_\varepsilon^+ := \psi_\varepsilon^+{}^{-1} \circ (\hat{\Phi}_\varepsilon^+ \circ \phi_\varepsilon): \abshyp \rightarrow \Sigma_\varepsilon^+$. Take, for example, the tensorial objects corresponding to the different first fundamental forms, $\hat{h}_\varepsilon^+ \equiv \hat{\Phi}_\varepsilon^+ (\hat{g}_\varepsilon^+)$ and $h_\varepsilon^+ \equiv \Phi_\varepsilon^+ (g_\varepsilon^+)$. These are related  by the diffeomorphism $\phi_\varepsilon$ as follows
\begin{eqnarray*}
	\phi_\varepsilon^* (\hat{h}_\varepsilon^+) &=& \phi_\varepsilon^* \circ \hat{\Phi}^+_\varepsilon{}^* \left(\hat{g}_\varepsilon^+ \right) = (\phi_\varepsilon^* \circ \hat{\Phi}^+_\varepsilon{}^* \circ \psi_\varepsilon^{-1}{}^*)\left(g_\varepsilon^+ \right) \nonumber\\
	&=& (\psi_\varepsilon^{-1} \circ \hat{\Phi}^+_\varepsilon \circ \phi_\varepsilon)^*(g_\varepsilon^+)= \Phi_\varepsilon^+{}^* (g_\varepsilon^+) = h_\varepsilon^+.
\end{eqnarray*}
This same procedure applied to the ``$-$'' spacetime yields $\phi_\varepsilon^* (\hat{h}_\varepsilon^-)= h_\varepsilon^-$. Therefore the matching condition \textit{upstairs} $\hat{h}_\varepsilon^+ = \hat{h}_\varepsilon^-$ implies the condition $h_\varepsilon^+ = h_\varepsilon^-$ in the abstract hypersurface $\abshyp$ by construction. This argument extends analogously for the remaining $\varepsilon-$family of scalars and tensor fields involved in the matching. 
Hence, the matching conditions for each $\varepsilon$ consist of imposing 
\begin{equation}
h^+_\varepsilon = h^-_\varepsilon, \qquad \boldsymbol{l}_\varepsilon^+ = \boldsymbol{l}_\varepsilon^-, \qquad l^{(2)}_\varepsilon{}^+ = l^{(2)}_\varepsilon{}^-.
%\qquad  Y^+_\varepsilon = Y^-_\varepsilon. 
\label{matchingconditions_epsilon}
\end{equation}
The first $\varepsilon$ derivatives of (\ref{matchingconditions_epsilon}), evaluated at $\varepsilon = 0$, provide the perturbed matching conditions.

\begin{figure}[h]

	\centering
	%%%%%%%%%%%%%%%%%%%%%%%%%%%%%%
	%%%%%%%%%%%%%%%%%%%%%%%%%%%%%%
	\includegraphics[%scale=0.20, 
	width=1.0\textwidth]{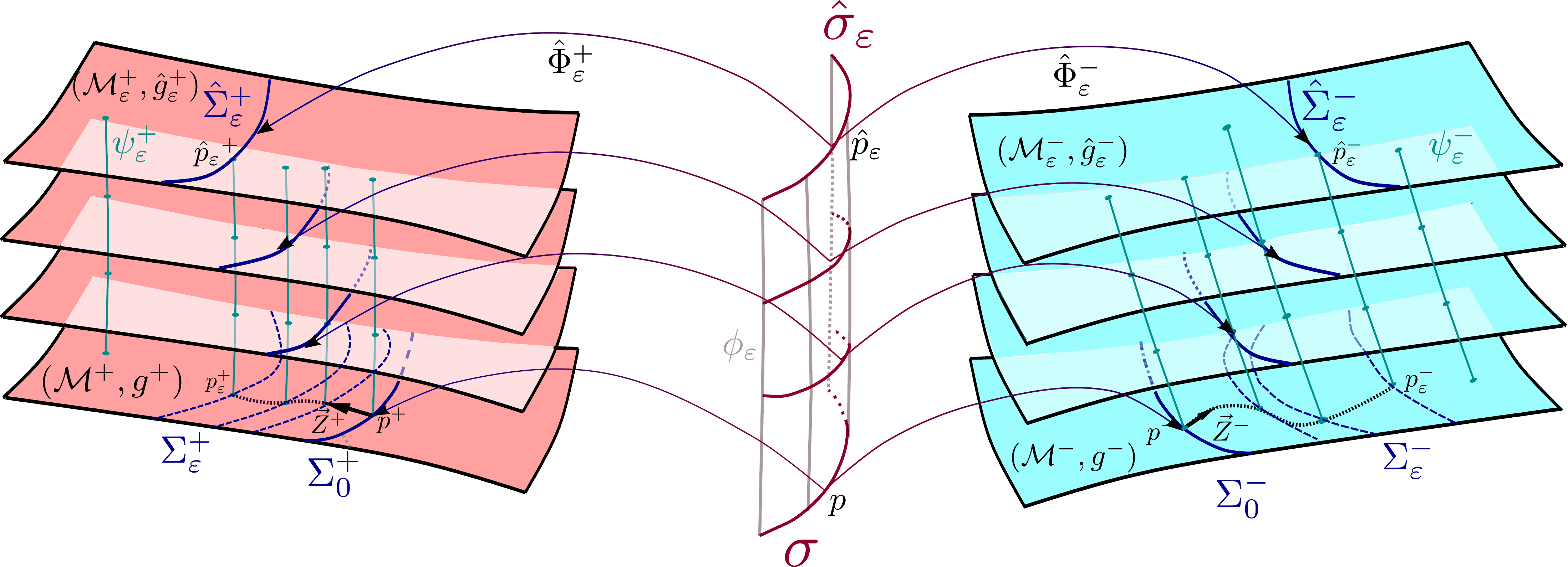}
	%%%%%%%%%%%%%%%%%%%%%%%%%%%%%%%%%
	%%%%%%%%%%%%%%%%%%%%%%%%%%%%%%%%%
	\caption{This diagram illustrates the setting for the perturbed matching of spacetimes. the basic mappings $\hat{\Phi}_\varepsilon^\pm$, $\psi_\varepsilon^\pm$ and $\phi_\varepsilon$ are depicted. The perturbed matching conditions (\ref{matchingconditions_epsilon}) are formulated in $\abshyp$.
	}
\label{figure:per_picture_4}
\end{figure}

\begin{theorem}{\bf(Matching conditions to first order}, generalization of Theorem 1 in \cite{Mars2005} {\bf )}
	%{Generalization of  \cite{Mars2005}.}
	\label{theorem:perturbed_matching}

	Let $(\mathcal{M},g)$ be the background spacetime arising from the matching of two spacetimes $(\mathcal{M}^+,g^+, \Sigma_0^+)$ and $(\mathcal{M}^-,g^-, \Sigma_0^-)$ across their boundaries $\Sigma_0^\pm$, diffeomorphic among themselves and also diffeomorphic to an abstract hypersurface $\abshyp$ so that there exist embeddings  $\Phi_0^\pm:\abshyp \rightarrow \Sigma_0^\pm \subset \mathcal{M}^\pm$. The riggings $\vec{l}^\pm$ for $\Sigma_0^\pm$ have been identified so that $\vec{l}\equiv \vec{l}^\pm$.
	
	Let $g_1^\pm$ be first order metric perturbations in $\mathcal{M}^\pm$, and $\vec{Z}{}^\pm = Q^\pm \vec{l} + \vec{T}{}^\pm$ deformation vectors for the hypersurfaces $\Sigma_0^\pm$. These induce first order perturbations of the hypersurface metric data $\{\delta h_{ab}^\pm, \delta l_a^\pm, \delta l^{(2)}{}^\pm \}$ as described in Propositions \ref{proposition_deltah} and \ref{proposition_deltal}, after the particularisations   $Q \rightarrow Q^\pm$, $\vec{T}_\abshyp \rightarrow \vec{T}_\abshyp^\pm$, $g_1 \rightarrow {g_1}^\pm$, $\vec{\zeta} \rightarrow \vec{\zeta}^\pm$, $g \rightarrow g^\pm$ therein.
	
	Then the first order perturbed matching conditions are fulfilled if and only if there exist two scalars $Q^\pm$ and two vectors $\vec{T} ^\pm_\abshyp$ on $\abshyp$ for which 
	\begin{equation}
	\delta h_{ab}^+ = \delta h_{ab}^-, \qquad \delta l_a^+ = \delta l_a^-, \qquad \delta l^{(2)}{}^+ = \delta l^{(2)}{}^-.
	%,\qquad \delta Y_{ab}^+= \delta Y_{ab}^-, 
	\label{matching_fo}
	\end{equation}
\end{theorem}

The matching conditions to first order imply that $\delta n^{(2)}{}^+ = \delta n^{(2)}{}^-$, $\delta n^a{}^+ = \delta n^a{}^-$ and $\delta P^{ab}{}^+ = \delta P^{ab}{}^-$.

We extend to first order in perturbation theory the result on the uniqueness on the riggings arising from the rigging compatibility conditions and stated in Lemma \ref{lemma_rigging_compatibility}.
 
\begin{proposition}{\bf(Rigging independence of the matching conditions to first order)}
	Consider a pair of sets of \textit{hypersurface metric data} satisfying the matching conditions and 
	gauge fields $\{\lambda, W, Z_a \}$ that preserve the rigging compatibility conditions.

	Consider first order perturbations of the rigging data satisfying $\{\delta l_a^+ = \delta l_a^-, \delta l^{(2)}{}^+ = \delta  l^{(2)}{}^- \}$ and gauge fields $\{\delta \lambda^\pm, \delta W^\pm, \delta Z_a^\pm \}$. The transformed first order perturbations of the rigging data are compatible, i.e. $\{\delta l^{(2)}{}'{}^+ = \delta l^{(2)}{}'{}^-, \delta l_a'{}^+ = \delta l_a'{}^- \}$, if and only if $ \delta \lambda^+ = \delta \lambda^-$, $\delta W^+ = \delta W^-$ and $\delta Z_a^+ = \delta Z_a^-$.
\end{proposition}
\begin{proof}
	
	As concluded in Lemma \ref{lemma_rigging_compatibility}, the compatibility in the transformed background riggings results into $[\lambda] = [W] = [Z_a] = 0$. We thus consider the three following equations which arise from taking differences in (\ref{linearized_constraint_tangent_vector}), (\ref{rigging_transformed_deltala}) and (\ref{rigging_transformed_deltal2}) and applying the matching conditions to first order and the compatibility of the background and first order non-primed objects  
	\begin{eqnarray}
	&&\left[\delta l_a' \right] = (l_a +Z_a) [\delta \lambda] + \lambda [\delta Z_a], \label{difference_delta_la_prime}\\
	&&\left[\delta l^{(2)}{}' \right] = 2 \lambda [\delta \lambda] (l^{(2)} + 2W + Z_a v^a) + 2\lambda^2 ([\delta W] + v^a[\delta Z_a]),  \label{difference_delta_l2_prime}\\
	&&n^a [\delta Z_a] + n^{(2)} [\delta W] = 0. \label{difference_constraint_transformation}
	\end{eqnarray}
	It is clear from these expressions that first order gauge transformations satisfying $[\delta \lambda] = [\delta W] = [\delta Z_a] = 0$  lead to compatible first order transformed riggings.
	The reverse implication requires some further manipulations. We contract $n^a$ with (\ref{difference_delta_la_prime}) and use (\ref{difference_constraint_transformation}) to find
	\begin{equation}
	n^a [\delta l_a'] = (1- n^{(2)}(l^{(2)} + W)) [\delta \lambda] - \lambda n^{(2)} [\delta W]. \label{difference_delta_la_prime_na}
	\end{equation}
	On the one hand, at null points it becomes $[\delta \lambda] = 0$, which plugged into the rest of the equations provides $[\delta Z_a] = [\delta W] = 0$.
	
	On the other hand, at non null points, the combination of  (\ref{difference_delta_la_prime_na}), (\ref{difference_delta_la_prime}) and (\ref{difference_constraint_transformation}) allows us to express the remaining differences in terms of $[\delta \lambda]$ as follows
	\begin{eqnarray*}
	&& \left[ \delta W \right] = \frac{1- n^{(2)}(l^{(2)}+ W)}{\lambda n^{(2)}} [\delta \lambda], \quad [\delta Z_a] = -\frac{l_a + Z_a}{\lambda}[\delta \lambda],
	\end{eqnarray*}
	which inserted into (\ref{difference_delta_l2_prime}) results in $[\delta \lambda]=0$, so that the compatibility of the first order transformations follows.
	
		\hfill $\blacksquare$
\end{proof}

	At this point the matching conditions to first order have been characterized, and we will assume for the rest of the paper that 
	the pair of spacetimes $(\mathcal M^\pm, g^\pm, \Sigma_0^\pm)$ satisfy the full matching conditions up to first order, i.e. the \textit{hypersurface data} set is given by $\{\abshyp, h_{ab}, l_a, l^{(2)}, Y_{ab} \}$ 
	in the background, and the first order perturbations fulfil $\{\delta h_{ab}^+ = \delta h_{ab}^-, \delta l_a^+ = \delta l_a^-, \delta l^{(2)}{}^+ = \delta l^{(2)}{}^-\}$. 	
	Also, when rigging transformations driven by the gauge fields $\{\lambda^\pm, W^\pm, Z_a^\pm, \delta \lambda^\pm, \delta W^\pm, \delta Z_a^\pm \}$ are involved, we assume that these preserve the matching conditions for the rigging up to first order, i.e. $\delta l_a'{}^+ = \delta l_a'{}^-$ and $\delta l^{(2)}{}'{}^+ = \delta l^{(2)}{}'{}^-$ ({\bf assumptions (*) }).
	
	 We extend the analysis of the perturbed matching of spacetimes and define a set of \textit{junction conditions} to first order. 
	One possibility is to formulate them in terms of the Riemann tensor, so that it does not have a singular part in analogy with case (1) in Theorem \ref{theorem_singular_distributions_background}. The extension of this result to the one-parameter family of Riemann tensors in the (upstairs) spacetimes  $(\mathcal{M}_\varepsilon^\pm,\hat{g}_\varepsilon^\pm,\hat{\Sigma}_\varepsilon^\pm)$ provides $\hat{Y}_\varepsilon^+=\hat{Y}_\varepsilon^-$ and the same analysis that we carried out in order to formulate the perturbed matching conditions in Theorem \ref{theorem:perturbed_matching} also applies in this case providing  $[\delta Y_{ab}]=0$.
	
	 However, we proceed as in case (3) from Theorem \ref{theorem_singular_distributions_background} in the exact theory and establish the notion of \textit{perturbed junction conditions} in terms of the shell energy momentum tensor. Therefore we promote the shell energy momentum tensor $\tau^{ab}$ introduced in (\ref{definition_shell}) to a one-parameter family of tensor fields $\tau^{ab}_\varepsilon$ by making the substitutions $\{P^{ab}\rightarrow P^{ab}(\varepsilon), \; n^a \rightarrow n^a(\varepsilon), \; n^{(2)} \rightarrow n^{(2)}(\varepsilon), \; Y_{ab} \rightarrow Y_{ab} (\varepsilon)\}$. The first order shell follows from taking $\varepsilon$-derivatives at $\varepsilon =0$. Assuming the \textit{background junction conditions} (\ref{eq_junction_exact}) hold we obtain that

\begin{eqnarray}
\delta \tau^{ab} &=& -n^a n^b \left(\delta P^{cd}[Y_{cd}] + P^{cd}[\delta Y_{cd}]\right) + \left(n^a P^{bc} + n^b P^{ac}\right) \left(\delta n^d [Y_{cd}] + n^d [\delta Y_{cd}]\right) \nonumber\\
&-& P^{ac}P^{bd} \left(\delta n^{(2)}[Y_{cd}] + n^{(2)}[\delta Y_{cd}]\right) - P^{ab} \left(-n^{(2)} (\delta P^{cd} [Y_{cd}] + P^{cd} [\delta Y_{cd}]) + n^c n^d [\delta Y_{cd}]\right).\nonumber \\
\end{eqnarray}
The energy momentum tensor on the shell depends on the rigging. This property arises very clearly from the distributional approach, where the complete singular part in the energy momentum tensor distribution for the whole spacetime is $\tau_{\alpha \beta} \delta^{\Sigma_0}$. The object $\delta^{\Sigma_0}$ is the scalar Dirac delta distribution with support on $\Sigma_0$ and it depends on the rigging via the normal one form. Since $\tau_{\alpha \beta} \delta^{\Sigma_0}$ is intrinsically defined, the tensor field $\tau_{\alpha \beta}$ also depends on the rigging (see \cite{lorentzian2007} for a detailed discussion).
It transforms as follows
\begin{equation*}
\tau^{ab}{}' = \frac{\tau^{ab}}{\lambda} \Rightarrow \tau^{ab}(\varepsilon){}' = \frac{\tau^{ab}(\varepsilon)}{\lambda (\varepsilon)} \Rightarrow \delta \tau^{ab}{}' = \frac{\delta \tau^{ab}}{\lambda} - \frac{\delta \lambda}{\lambda^2}\tau^{ab}.
\end{equation*}
Hence, provided that the \textit{junction conditions} hold for the background ensuring that $\tau^{ab} = 0$, it make sense to study under which conditions $\delta \tau^{ab} = 0$, since it is a gauge independent equation.

\begin{proposition}{ {\bf(Junction conditions to first order)}}
	
	Assume that the background \textit{junction conditions} are satisfied, so that $\tau_{ab}=0$ and the dimension of $\abshyp$ is $\dimension >1$. The tensor $\delta \tau_{ab}$ vanishes:	
	\begin{itemize}
		\item at null points if and only if 
		\begin{equation}
		\delta n^b [Y_{ab}] + n^b [\delta Y_{ab}] = 0, \quad \delta P^{ab} [Y_{ab}] + P^{ab}[\delta Y_{ab}] = 0, \quad \delta n^{(2)} P^{ac}[Y_{bc}] = 0. \label{eq_deltatau_null}
		\end{equation}
		\item at non-null points if and only if $[\delta Y_{ab}]=0$.
	\end{itemize}
\label{proposition_first_order_junction_conditions}
\end{proposition}
\begin{proof}
	We put forward the following definitions
	\begin{eqnarray}
	W^b &\equiv& \left( P^{bc}\delta n^d -\delta P^{cd} n^b \right)[ Y_{cd}] + \left( P^{bc} n^d - P^{cd} n^b \right)[\delta Y_{cd}] , \label{eq_deltatau_1}\\
	Z_c^b &\equiv& \left(n^b \delta n^d - \delta n^{(2)} P^{bd} \right)[ Y_{cd}] + \left( n^b n^d - n^{(2)} P^{bd} \right)[\delta Y_{cd}] \nonumber\\
	&&+ \delta_c^b \left( n^{(2)} \delta P^{ef}[Y_{ef}] - \left(n^e n^f - n^{(2)} P^{ef} \right) [\delta Y_{ef}]\right), \label{eq_deltatau_2}
	\end{eqnarray}
	which allow us to write $\delta \tau^{ab} = W^b n^a + P^{ac}Z_c^b$. It can be checked with the aid of the \textit{junction conditions} (\ref{eq_junction_exact}) that $n^{(2)} W^b + n^c Z_c^b =0$. Hence $\delta \tau^{ab}=0$ if $W^b=Z_c^b=0$. It is useful to define $Q^{ab} \equiv n^{(2)}P^{ab} -n^a n^b$ and consider the following list of derived objects
	\begin{eqnarray}
	h_{ab}W^b &=&\left( \delta n^d [Y_{ad}] + n^d [\delta Y_{ad}]\right) + l_a \left \lbrace n^{(2)} \delta P^{cd}[Y_{cd}] + Q^{cd}[\delta Y_{cd}]\right \rbrace,\label{eq_deltatau_3}\\
	%%%%%%%%%%%%%%%%%%%%%%%%%%
	l_b W^b &=&-\left(\delta P^{cd} [Y_{cd}] + P^{cd} [\delta Y_{cd}]\right) + l^{(2)} \left( n^{(2)} \delta P^{cd} [Y_{cd}] + Q^{cd} [\delta Y_{cd}]\right), \label{eq_deltatau_4}\\
	%%%%%%%%%%%%%%%%%%%%%%%%%%
	Z^b_b &=& (\dimension-1) Q^{bd}[\delta Y_{bd}] + (\dimension) n^{(2)} \delta P^{bd} [Y_{bd}]. \label{eq_deltatau_5}
	\end{eqnarray}
	We start by (\ref{eq_deltatau_5}). Its second term vanishes at null points, because $n^{(2)} = 0$ there, and also at non-null points, because $[Y_{ab}] = 0$ there. Also, as in the background case, we are assuming $\dimension-1 \neq 0$, so that we obtain that $Q^{ab}[\delta Y_{ab}] = 0$.
	
	We focus now on null points exclusively, so that $Q^{ab}[\delta Y_{ab}] = 0 \Rightarrow n^an^b [\delta Y_{ab}] = 0$. Equating (\ref{eq_deltatau_3}) and (\ref{eq_deltatau_4}) to zero establishes the first and second conditions in (\ref{eq_deltatau_null}). The conditions obtained so far suffice for $W^b=0$. The remaining conditions, the third one in (\ref{eq_deltatau_null}), follow from imposing that $Z_c^b=0$.
	
	At non-null points we follow the same strategy, but the fact that $[Y_{ab}]=0$ simplifies the procedure. The equations arising from (\ref{eq_deltatau_3})-(\ref{eq_deltatau_5}) provide $n^a [\delta Y_{ab}] = P^{ab}[\delta Y_{ab}] = 0$. Back to $Z_c^b = 0$ we find $P^{ac}[\delta Y_{bc}] = 0$ and therefore $[\delta Y_{ab}] = 0$.
	
	Sufficiency is proven by direct substitution of the conditions listed in the proposition into the expression given for $\delta \tau^{ab}$.
\hfill $\blacksquare$
\end{proof}

\begin{proposition}{{\bf (Rigging independence of the junction conditions to first order at null points)}}

	Consider the {\bf assumptions (*)}. 	
	Assume that there are null points at $\abshyp$, where the background junction conditions are  satisfied. Then  the first order \textit{junction conditions} $\delta n^b [Y_{ab}] + n^b [\delta Y_{ab}] = 0$,  $\delta P^{ab} [Y_{ab}] + P^{ab}[\delta Y_{ab}] = 0$ and $\delta n^{(2)} P^{ac}[Y_{bc}] = 0$ imply that  $\delta n^b{}' [Y_{ab}{}'] + n^b{}' [\delta Y_{ab}{}'] = 0$,  $\delta P^{ab}{}' [Y_{ab}{}'] + P^{ab}{}'[\delta Y_{ab}{}'] = 0$ and $\delta n^{(2)}{}' P^{ac}{}'[Y_{bc}{}'] = 0$ at null points.
	\label{proposition_perturbedjunction_independent_rigging_1}
	\end{proposition} 
\begin{proof}
 	Take the transformations (\ref{inverse_data_transformations}) and those in Lemma \ref{lemma_rigging_transformation_perturbed_inverse_data}, set $n^{(2)} =0$ there, and make use of the background \textit{junction conditions} for null points. After some intermediate calculations following the previous steps, we get the following expressions for the primed \textit{junction conditions} stated in the proposition 
 	\begin{eqnarray}
 	\delta n^b{}' [Y_{ab}{}'] + n^b{}' [\delta Y_{ab}{}'] &=& \delta n^b [Y_{ab}] + n^b [\delta Y_{ab}] -\delta n^{(2)}\overline{v}^b[Y_{ab}] \nonumber \\
 	&=& \left(\delta n^b [Y_{ab}] + n^b [\delta Y_{ab}] \right)- \left(\delta n^{(2)}P^{bc}[Y_{ab}] \right)Z_c,\nonumber\\
 	%%%%%%%%%%%%%%%%%%%%%%%%%
 	\delta P^{ab}{}' [Y_{ab}{}'] + P^{ab}{}'[\delta Y_{ab}{}'] &=& \lambda \left \lbrace \left( \delta P^{ab} [Y_{ab}] + P^{ab}[\delta Y_{ab}] \right) -2 \overline{v}^a \left( \delta n^b [Y_{ab}] + n^b [\delta Y_{ab}] \right) \right.\nonumber\\
 	%%%%%%%%%%%%%%%%%%%%%%%%%
 	&&\left.  + \overline{v}^a \overline{v}^b \delta n^{(2)} [Y_{ab}]\right \rbrace = \lambda \left \lbrace \left( \delta P^{ab} [Y_{ab}] + P^{ab}[\delta Y_{ab}] \right) \right.\nonumber\\
 	&&\left. -2 \overline{v}^a \left( \delta n^b [Y_{ab}] + n^b [\delta Y_{ab}] \right)   + \left(\delta n^{(2)} P^{ac}   [Y_{ab}]\right) P^{bd} Z_c Z_d\right \rbrace, \nonumber \\
 	%%%%%%%%%%%%%%%%%%%%%%%%
 	\delta n^{(2)}{}' P^{ac}{}'[Y_{bc}{}'] &=&\frac{1}{\lambda} \left(\delta_d^a - n^a Z_d \right) \left(\delta n^{(2)} P^{cd} [Y_{bc}] \right).\nonumber
 	\end{eqnarray}
 	Imposing the \textit{first order junction conditions} for null points in the expressions above, the result follows.
	\hfill $\blacksquare$
\end{proof}
\begin{proposition} {{\bf (Rigging independence of the junction conditions to first order at non-null points)}}
	
	Consider the {\bf assumptions (*)}. 	
	Assume that there are non-null points in $\abshyp$, where the background junction conditions are satisfied. 	Then  the first order \textit{junction conditions}  $[\delta Y_{ab}] = 0$ imply that $[\delta Y'_{ab}] = 0$. .
\end{proposition}
\begin{proof}
	Take the difference of the gauge transformed rigged fundamental form (\ref{rigging_transformed_deltalY}) as seen from the different $+$ and $-$ sides. Consider also that $ [\partial_b \delta \lambda] =  \partial_b [\delta \lambda] = 0$. Then it is straightforward to see that every single term in the difference $[\delta Y_{ab}']$ vanishes.
	
	\hfill $\blacksquare$
\end{proof}

The \textit{first order junction conditions} formulated from the analysis of the tensor $\delta \tau^{ab}$ are therefore independent of rigging transformations. In addition they are hypersurface gauge independent. This fact follows from the derivatives of the conditions (\ref{eq_junction_exact}) which provide
\begin{eqnarray}
&&n^a\mathcal L_{\vec{u}_\abshyp} [Y_{ab}] = 0, \quad P^{ac}\mathcal L_{\vec{u}_\abshyp} [Y_{ab}] = 0  \quad \text{at non-null points},\nonumber\\
&& P^{ac}[Y_{ab}] \mathcal L_{\vec{u}_\abshyp} n^{(2)} = 0, \quad \mathcal L_{\vec{u}_\abshyp} P^{ab}[Y_{ab}] = 0,  \quad \mathcal L_{\vec{u}_\abshyp} n^a[Y_{ab}] = 0 \quad \text{at null points}.
\end{eqnarray}
The two equalities in the first row above imply that $\mathcal L_{\vec{u}_\abshyp} [Y_{ab}] = 0$ at non-null points which proves the hypersurface gauge invariance of the condition $[\delta Y_{ab}] = 0$. 

The gauge invariance of the \textit{junction conditions} at null points follows from the identities in the second row. Also in this case an argument based in Lemma 2.2 in \cite{Stewart49} would provide the desired result, since it is possible to write all the \textit{first order junction conditions} in the form $\delta t$ for some background tensor field $t$ in $\abshyp$, which vanishes everywhere. For instance at null points
\begin{eqnarray}
\delta n^{(2)} P^{ac}[Y_{bc}] &=& \delta \left( n^{(2)} P^{ac}[Y_{bc}] \right) - n^{(2)} \delta \left( P^{ac}[Y_{bc}]\right)\nonumber\\
&=& \delta \left( n^{(2)} P^{ac}[Y_{bc}] \right) ,
\end{eqnarray}
which is unaffected by hypersurface gauge transformations since $ n^{(2)} P^{ac}[Y_{bc}] =0$ is a \textit{junction condition} everywhere (recall (\ref{eq_junction_exact})). The same argument holds for the rest of the conditions.

\section*{Acknowledgements}

We are grateful to Marc Mars, Ra\"ul Vera, Abraham Harte, Arnaud Mortier, Ko Sanders and Peter Taylor for valuable discussions and suggestions. Figures \ref{figure:per_picture_3} and \ref{figure:per_picture_4} were made modifying a diagram provided to us by Ra\"ul Vera.

 Work supported by the Spanish MINECO and FEDER (FIS2017-85076-P), by the Basque Government (IT-956-16, POS-2016-1-0075, IT-979-16), by the ERC Advanced Grant 339169 ``Selfcompletion'' and by the Spanish Research Agency (Agencia Estatal de Investigación) through the grant IFT Centro de Excelencia Severo Ochoa SEV-2016-0597 and the project FPA2015-65480-P.
 
 BR and BN derived the results presented in this paper. KS independently derived the expressions of the hypersurface data to first order in perturbation theory and obtained some of the gauge redundancies of the geometric data presented in sections \ref{section_perturbations_data} and \ref{section_freedom}. The manuscript was prepared by BR and BN, and revised and finalised by all three authors.

\appendix
\section{Constant signature hypersurfaces }
\label{appendix_constant_signature}
 As discussed in Section \ref{section_freedom}, when the causal character of the $\Sigma_\varepsilon$ does not vary with $\varepsilon$, the rigging data is specified directly from the $\Sigma_\varepsilon$ demanding that $ l^{(2)}_\varepsilon = C$, $ (\boldsymbol{l}_\varepsilon) = \boldsymbol C$, where $C$ and $\boldsymbol C$ are a constant and a constant one covariant tensor (with respect to $\varepsilon$) in $\abshyp$ . This is implicitly used to fix the freedom in the first order perturbations of the rigging data, so that $\delta l^{(2)} = \delta l_a = 0$.

\subsection*{Standard hypersurfaces}
\label{section_standard}
Since we have followed the methods in \cite{Mars2005} to construct the perturbations, we show in this section that our expressions agree with those given therein when we restrict ourselves to standard hypersurfaces in the gauge where the normal vector is chosen as the rigging.
% The key idea for this is fixing the rigging vectors of $\Sigma_\varepsilon$ so that they follow the corresponding normal vectors. 
Indeed, choosing the background rigging as $\vec{l}=\norm \vec{n}$, with $\norm = 1 (-1)$ for timelike (spacelike) hypersurfaces leads immediately to $l^{(2)} = \norm$, $l_a = 0$, $n^a =0$ and $P^{ac} h_{bc} = \delta_b^a$. In fact $h$ is not degenerate and we will denote $P^{ab} \equiv h^{ab}$. Regarding the objects that are built from taking derivatives of the basis we find that $\varphi_a = 0$, $\Psi_a^b = \norm \kappa_a^b$, $Y_{ab} = \norm \kappa_{ab}$ and $\vec{a} \equiv \nabla_{\vec{l}} \vec{l} = \nabla_{\vec{n}} \vec{n}$.
The hypersurface deformation vector reads
\begin{equation}
\vec Z = \norm Q \vec n + \vec T. \label{tostandard}
\end{equation} 
Note that this decomposition  differs slightly from the decomposition from \cite{Mars2005}, where the deformation vector is expanded as $\vec{Z} = Q \vec{n} + \vec{T}$. Therefore at the time of comparing our expressions with those in \cite{Mars2005} we will have to make the substitution $\norm Q \rightarrow Q$.
Under these choices the perturbation of the first fundamental form becomes
\begin{eqnarray}
\delta h_{ab} &=& 2\norm Q\kappa_{ab} + \mathcal{L}_{\vec T{}_\abshyp} h_{ab} + g_1 (\vec{e}_a, \vec{e}_b).\label{pert_fff_timelike}
\end{eqnarray}
This is to be compared with the corresponding expression in (26) in \cite{Mars2005}.

Since the whole family $\{\Sigma_\varepsilon\}$ is timelike/spacelike, the riggings are chosen so that $\vec{l}_\varepsilon = \norm \vec{n}_\varepsilon$, and therefore $C=\norm$, $C_a = 0$, leading to $\delta l^{(2)} = \delta l_a = 0$. Note also that this choice entails that, as a vector, $\vec{l}_1 = \norm \vec{n}_1$. After using (\ref{tostandard}) and the expressions immediately above, the perturbations of the rigging data (\ref{l12_b}) and (\ref{l1a_b}) become
\begin{eqnarray*}
\delta l^{(2)}= 2 \norm \Phi^*(n_\alpha \zeta^\alpha) + g_1(\vec{n},\vec{n}) = 0,\quad \delta l_a =   \Phi^* (\mathcal \zeta) + \norm \vec{e}_a (Q) + \norm g_1(\vec{n},\vec{e}_a) = 0.
\end{eqnarray*}
Using the explicit definition of $\vec \zeta \equiv  \vec{l}_1 + Q \nabla_{\vec{l}} \vec{l} = \norm \vec{n}_1 + Q \nabla_{\vec{n}}\vec{n}$, we find $\Phi^*(n_\alpha \zeta^\alpha) = \norm n_\alpha n_1^\alpha + Q n_\alpha a^\alpha$, which in combination with the first of the two equations above allows us to get
\begin{equation}
n_\alpha n_1^\alpha = -\norm Q n_\alpha a^\alpha - \frac{1}{2}  g_1(\vec{n},\vec{n}). \label{eq:nn1}
\end{equation}
Considering $\delta l_a=0$ and repeating the same procedure with the tangent part of $n_1$ one gets
\begin{equation}
(n_1)_\alpha e_a^\alpha =  -\norm  Q a_\alpha e_a^\alpha -  \vec{e}_a(Q) - g_1(\vec{n},\vec{e}_a) \label{eq:en1}
\end{equation}
Thus the formulas (\ref{l12_b}) and (\ref{l1a_b}) allow us to obtain the decomposition of the perturbed normal vector $\vec{n}_1$ in the basis $\{\vec{n},\vec{e}_a\}$. In fact a vector $\vec{v}$ and its metrically related one form $\boldsymbol{v}$ defined at points of $\Sigma_0$ admit decompositions in the basis adapted to $\Sigma_0$ (as specified above), that are related as follows
\begin{eqnarray}
v^\alpha \equiv \norm V n^\alpha + V^a e_a^\alpha, \quad v_\alpha \equiv \norm U n_\alpha + U_a \omega^a_\alpha,\quad \text{with} \;\; V=U, \quad \text{and} \;\; V^a = h^{ab}U_b.
\end{eqnarray}
Therefore we can reconstruct the perturbation of the normal vector entirely
\begin{equation}
n_1^\alpha = - \frac{\norm}{2} g_1(\vec{n}, \vec{n}) n^\alpha  - h^{ab}g_1(\vec{n},\vec{e}_b)e_a^\alpha -\norm Q a^\alpha - h^{ab} \partial_b Q e_a^\alpha .\label{n1vector}
\end{equation}
If we borrow the definitions $Y' \equiv g_1(\vec{n},\vec{n})$ and $\tau_a\equiv  g_1(\vec{e}_a, \vec{n})$ from \cite{Mars2005} and consider the redefinition in $Q$ we immediately obtain expression (31) from \cite{Mars2005}. Also, we note that we have an explicit expression for the vector $\vec \zeta =  - \frac{1}{2} g_1(\vec{n}, \vec{n}) n^\alpha  - \norm h^{ab}g_1(\vec{n},\vec{e}_b)e_a^\alpha - \norm h^{ab} \partial_b Q e_a^\alpha$ that does not contain an acceleration term in its right hand side. This is the step where the accelerations disappear from the calculation of the perturbations of the second fundamental form, as we see next.

Due to the rigging choice of the $\varepsilon$-family of hypersurfaces, the second fundamental forms fulfil $\kappa_\varepsilon = \norm Y_\varepsilon$, and therefore we can extract the first order perturbations $\delta \kappa$ from our expression for $\delta Y$ appropriately particularised to the conventions above for standard hypersurfaces. The formula (\ref{Yp_1_b}) for the rigged extrinsic curvature yields 
\begin{eqnarray}
\delta \kappa_{ab} &=& \norm \delta Y_{ab} = \frac{1}{2}\norm \Phi^* \mathcal{L}_{\vec \zeta}g 
 +\norm Q(-R_{\alpha \gamma \beta \mu} n^\gamma n^\mu e_a^\alpha e_b^\beta + \kappa_a^c \kappa_b^d h_{cd}) + \mathcal{L}_{\vec{T}{}_\abshyp} \kappa_{ab} + \frac{1}{2} \Phi^*\mathcal{L}_{\vec n}g_1.\label{recovering_extrinsic_curvature}
\end{eqnarray}
We need to rewrite some of the terms in order to complete the comparison with \cite{Mars2005}.
From our previous discussion about the perturbation of the normal vector it is immediate to obtain that 
\begin{equation*}
\frac{\norm}{2}\Phi^* \mathcal{L}_{\vec \zeta}g = -\frac{1}{2} \Phi^* \mathcal L_{\frac{\norm Y'}{2} n^\mu + \tau^\mu} g -  \overline{\nabla}_a \overline{\nabla}_b Q.
\end{equation*}

The other term that deserves attention is the last one in (\ref{recovering_extrinsic_curvature}). For this we apply Lemma \ref{lemma:S} with $A=g_1$ and $\vec{X}= \vec{n}$ to find that
\begin{eqnarray}
\mathcal L_{\vec{n}} g_1 &=& -2 n_\mu S^\mu_{\alpha \beta} + \mathcal L_{\norm Y'n^\mu + \tau^\mu} g_{\alpha \beta}. \nonumber
\end{eqnarray}

Plugging these results into (\ref{recovering_extrinsic_curvature}) we finally get
\begin{eqnarray}
\delta \kappa_{ab} &=& \mathcal{L}_{\vec{T}{}_\abshyp} \kappa_{ab}  -  \overline{\nabla}_a \overline{\nabla}_b Q  + \norm Q(-R_{\alpha \gamma \beta \mu} n^\gamma n^\mu e_a^\alpha e_b^\beta + \kappa_{ac} \kappa_b^c ) + \frac{1}{2} \Phi^* \mathcal L_{\frac{\norm Y'}{2} \vec{n} } g   - n_\mu S^\mu_{\alpha \beta}e_a^\alpha e_b^\beta, \nonumber
\end{eqnarray}
which recovers the corresponding expression in \cite{Mars2005} after taking into account the change $Q \rightarrow \norm Q$ and that $\frac{1}{2} \Phi^* \mathcal L_{\frac{\norm Y'}{2} \vec{n} } g = \frac{\norm Y'}{2} \kappa$.
\subsection*{Null hypersurfaces}
\label{section_null}
We compare our method with the geometrical methods introduced in \cite{VegaPoissonMassey} to study perturbations of null hypersurfaces, applied to the deformation of a black hole due to a tidal field.  Firstly, we state some general properties about null hypersurfaces needed to carry out the analysis. The intrinsic coordinates in $\abshyp$ are $y^a = \{\lambda, \alpha^A\}$, where the upper-case Latin indices are coordinates in the spherical slices of the hypersurface. These admit a two dimensional Riemannian metric $\gamma_{AB}$ whose compatible connection we denote by $\Gamma_{AB}^C$. This hypersurface is embedded into the spacetime via $\Phi_0$ and the null normal vector is thus defined by $\vec{n} = \partial_\lambda \Phi^\alpha$. The basis is completed with the vectors $\vec{e}_A = \partial_A \Phi_0$ and a null rigging, which is chosen so that it satisfies the condition $l_\alpha e_A^\alpha = 0$, plus its defining condition $l_\alpha n^\alpha = 1$. Note that the sign in the latter differs from the convention taken in \cite{VegaPoissonMassey}, but the effect of this choice is just an overall sign in $Y$ and $\delta Y$.
It is straightforward to formulate the tangential derivatives of the basis elements
\begin{eqnarray*}
&&n^\alpha \nabla_\alpha n^\beta = \overline \kappa n^\alpha, \quad e_A^\alpha \nabla_\alpha n^\beta = \omega_A n^\beta + B_{A}^{\phantom A B}e_B^\beta,\\
&& n^\alpha \nabla_\alpha e_A^\beta = \omega_A n^\beta + B_{A}^{\phantom A B}e_B^\beta, \quad e_A^\alpha \nabla_\alpha e_B^\beta =-B_{AB} l^\beta - \mathcal{K}_{AB} n^\beta + \Gamma_{AB}^C e_C^\beta,\\
&&n^\alpha \nabla_\alpha l^\beta = -\overline \kappa l^\beta - \omega^B e_B^\beta, \quad e_A^\alpha \nabla_\alpha l^\beta = -\omega_A l^\beta + \mathcal{K}_A^{\phantom{A}B}e_B^\beta,
\end{eqnarray*}
where the operations of raising and lowering upper-case Latin indices were relative to the metric $\gamma_{AB}$.
A quick comparison between these previous expressions and the rigged fundamental form allows us to establish
\begin{eqnarray*}
Y_{ab} = \begin{bmatrix}
-\overline \kappa   & -\omega_A \\
-\omega_B & \mathcal{K}_{AB}
\end{bmatrix}, \quad
\Psi_{a}^{\phantom{a}b} = \begin{bmatrix}
0   & -\omega_A \\
0 & \mathcal{K}_{A}^{\phantom{A}B}
\end{bmatrix}, \quad 
\varphi_a =  \begin{bmatrix}
-\overline \kappa\\
-\omega_A
\end{bmatrix}.
\end{eqnarray*}

Apart from these conventions, the ambient geometry is taken to be Schwarzschild, expressed in the ingoing Eddington Finkelstein coordinates
\begin{equation*}
g= -f dv^2 + (drdv + dvdr) + r^2 d \Omega^2, \quad f:= \left(1-\frac{2M}{r}\right),
\end{equation*}
where in this setting the two dimensional metric $\gamma_{AB} = r^2 \Omega_{AB}$, and $d \Omega^2 = \Omega_{AB}dx^A dx^B$ is the metric on the unit sphere. Its associated covariant derivative will be denoted by $D_A$ in the remaining part of this appendix.

 The embedding is simply $\Phi_0 = \{v=\lambda, \; r= 2M, \; x^A = y^A \}$.
Then the tangent basis satisfying the requisites from the general setting introduced above is $\{\vec{l}, \vec{e}_a \}$ is given by $\vec{l} = \partial_r$, $\vec{e}_1 \equiv \vec{n} = \partial_v$ and $\vec{e}_A = \partial_A \Phi_0$ and the dual basis $\{\boldsymbol{n}, \mathbf{\omega}^a\}$ is $\boldsymbol{n} = -f dv + dr$, $\omega^1 = dv$ and $\omega^A$ satisfies $\mathbf{\omega}^A (\vec{e}_B) = \delta^A_B$.
The relevant objects for the data approach are thus $l^{(2)} = n^{(2)} = 0$, $l_a = d\lambda$, $n^a =\partial_\lambda $ and
\begin{eqnarray*}
h_{ab} = (2M)^2\begin{bmatrix}
0   & 0 \\
0 & \Omega_{AB}
\end{bmatrix}, \quad
P^{ab} = (2M)^{-2}\begin{bmatrix}
0   & 0 \\
0 & \Omega^{AB}
\end{bmatrix}, \quad
Y_{ab} = \begin{bmatrix}
-\frac{1}{4M}   & 0 \\
0 & 2M\Omega_{AB}
\end{bmatrix} .
\end{eqnarray*}
Because of the conventions we are using, our expression $\mathcal K_{AB} = 2M \Omega_{AB}$ differs from \cite{VegaPoissonMassey} in a sign. This difference will propagate also to its first order perturbations.

The $\varepsilon-$family of embeddings is considered to be $\Phi_{\varepsilon} = \{v= \lambda, \; r = 2M (1+\varepsilon B(\lambda, y^A)), \; x^A = y^A + \varepsilon \Xi^A (\lambda, y^A) \}$. Note that into our setting, this is translated into a deformation vector of the form $\vec{Z} = Q \vec{l} + \vec{T} = 2MB \vec{l} + \Xi^C \vec{e}_C$. Furthermore, the spacetime metric is perturbed via $g_1$, which is a general metric perturbation tensor for which no symmetries are assumed a priori, and no spacetime gauge fixing conditions are imposed. This tensor, called $p_{\alpha \beta}$ in \cite{VegaPoissonMassey} is expanded in scalar, vector and tensor spherical harmonics therein, exploiting thus the spherical symmetry of the background. However, in the calculations carried out in this appendix we will skip the decomposition in harmonics and work with $g_1$ and its projections into the basis $\{\vec{l}, \vec{n}, \vec{e}_A\}$ generically.
 
We are ready to apply formula (\ref{pert_fff}) for the first fundamental form, whose $A-B$ components are found to be
\begin{eqnarray*}
\delta h_{AB} &=& 2Q Y_{AB} + \mathcal{L}_{\vec{T}_\abshyp} h_{AB} + g_1 (\vec{e}_A, \vec{e}_B) = (2M)^2 \left(2B \Omega_{AB} + \Omega_{CB}D_A \Xi^C + \Omega_{AC} D_B \Xi^C\right) + g_1 (\vec{e}_A, \vec{e}_B) .
\end{eqnarray*}
This result agrees with the first fundamental form (3.35) in \cite{VegaPoissonMassey}.

Regarding the rigging vector, we know that it is null and that $l_a \equiv l_\alpha e_a^\alpha = \delta_a^n$, that is, $l_a$ is constant. These two properties can be promoted to a general behaviour of the $\varepsilon-$family of riggings, which amounts to selecting a gauge where $\delta l_a \equiv \partial_\varepsilon \boldsymbol{l}_\varepsilon|_{\varepsilon = 0} = 0$, supplemented with the null condition $\delta l^{(2)} \equiv \partial_\varepsilon l^{(2)}_\varepsilon|_{\varepsilon = 0} = 0$. A direct application of fomulas (\ref{zeta_normal}) and (\ref{s_deltas}) tells us that
\begin{eqnarray*}
\alpha + Q(n_\mu a^\mu) = -n^a g_1 (\vec{l},\vec{e}_a), \quad \vec{s}_{\abshyp} + Q(\omega_\mu^a a^\mu) = -n^a\left( \frac{1}{2} g_1 (\vec{l},\vec{l})\right) - P^{ab} g_1(\vec{l},\vec{e}_b),
\end{eqnarray*}
where we used that $\LL_{\vec T_{\abshyp}}l_a = T^c (\partial_c l_a - \partial_a l_c) + \partial_a (T^c l_c) = 0$, since $T^c l_c = 0$ and $l_c$ is constant.
Hence the assumptions taken on the perturbed rigging data completely determine the vector field $\vec{\zeta}$ whose components are in full agreement with the perturbation of the transverse vector $N$ given in \cite{VegaPoissonMassey} (see (3.28)-(3.30) therein).
Now we can use expression (\ref{Yp_1_b}) for $\delta Y$ and after some computations, we can read its components. For instance we see that
\begin{eqnarray*}
\delta \overline \kappa &\equiv& -\delta Y_{\lambda \lambda}= -\frac{1}{4M}\left ( 2B + {g_1}{}_{vr} + 2M (\partial_r {g_1}{}_{vv} - 2 \partial_v {g_1}{}_{vr})\right),
\end{eqnarray*}
which is (3.52)-(3.53) from \cite{VegaPoissonMassey}. 

Moreover, the perturbed hypersurface is required to remain null. In fact, this choice is translated into an equation for the (perturbed) null generators of the hypersurface, which in turn constrains the functions that perturb the embedding. In our method, this should be translated in restrictions of the deformation vector $\vec{Z}$, and in particular in its rigged component $Q$, which controls how it is deformed as a set of points in the ambient spacetime. Therefore, following the discussion from Section \ref{section_freedom}, we find a similar condition asking for $\delta n^{(2)} = 0$. A direct application of (\ref{delta_n2}) to this setting yields
\begin{equation*}
g_1{}_{vv} - \frac{Q}{2M} + 2\partial_v Q|_{\Sigma_0}=0 \Rightarrow B - 4M \partial_v B = g_1{}_{vv}.
%g_1{}_{vv} - B + 4M\partial_v B|_{\Sigma_0}=0.
\end{equation*}
We can ask for further conditions to the perturbed embeddings, in particular we demand that the perturbed normal does not acquire components along the directions tangential to the spheres, i.e. $\delta n^A=0$. This provides
\begin{equation*}
(g_1{}_{vA} +  \partial_A Q + 4M^2 \Omega_{AB} \partial_v T^B) = 0 \Rightarrow \partial_v \Xi_A = -(2M)^{-2} \left( g_1{}_{vA} + 2M\partial_A B\right),
\end{equation*}
where we have defined $\Xi_A \equiv \Omega_{AB} \Xi^B$,  in order to ease the comparison with \cite{VegaPoissonMassey}.
These three conditions on the deformation vector $\vec{Z}$ expanded into spherical harmonics agree with (3.25)-(3.27) in \cite{VegaPoissonMassey}, and plugged into the expression for $\delta Y$ provide the remaining components (see (3.54)-(3.60) in \cite{VegaPoissonMassey})
\begin{eqnarray*}
\delta \omega_A &\equiv& -\delta Y_{\lambda A} = \frac{1}{2} \left( \frac{1}{M} (g_1{}_{vA} + 2M \partial_A B) + \partial_A g_1{}_{vr} - \partial_r g_1{}_{vA} + \partial_v g_1{}_{rA} \right),\\
\delta \mathcal K_{AB} &\equiv& \delta Y_{AB} = 2M(B - g_1{}_{vr})\Omega_{AB} + 2M (D_A \Xi_B + D_B \Xi_A) -\frac{1}{2}(D_A g_1{}_{rB} + D_B g_1{}_{rA}) + \frac{1}{2} \partial_r g_1{}_{AB}. \nonumber\\
\end{eqnarray*}

%%%%%%%%% APPENDIX B

\section{Useful formulas and another derivation of $\delta Y$ }
In this appendix we include some of the formulas needed in the application of Lemma \ref{lemma_marc_perturbations} to compute different perturbations. We use the following formulas in order to commute different derivative operators \cite{geroch2013differential}
\begin{eqnarray}
&&\mathcal L_{\vec{u}} \nabla_\alpha v_\beta - \nabla_\alpha \mathcal L_{\vec{u}} v_\beta = (R_{\mu \alpha \beta}^{\phantom{\mu \alpha \beta} \rho} u^\mu - \nabla_\alpha \nabla_\beta u^\rho)v_\rho, \label{commute_Lie_cov}\\
&&\left( \LL_{\vec{u}} \LL_{\vec{v}} - \LL_{\vec{v}} \LL_{\vec{u}} \right)  A = \LL_{[\vec u, \vec v]}  A,
\end{eqnarray}
for any pair of vectors $\vec{u}$ and $\vec{v}$ and a tensor $ A$.

The following lemma is useful to compute $\delta Y$, not only in the calculations from Section \ref{section_perturbations_data} but also in an alternative way to compute it that we provide later in this same appendix.

\begin{lemma}{}\label{lemma_pullbackLie}
	Let $\abshyp$ be an embedded submanifold of $\mathcal{M}$ with embedding $\Phi:\abshyp \rightarrow \Sigma_0 \subset \mathcal{M}$ and $\mathcal A$ any covariant tensor defined on a neighbourhood $\mathcal{U}$ of $\Sigma_0$.
	\begin{enumerate}
		\item Let $\vec{V}$ be a vector field on a neighbourhood $\mathcal{U}$ of $\Phi(\abshyp)$ tangent to this hypersurface and define $\vec{V} \equiv d\Phi (\vec{V}_\abshyp)$. Then
		\begin{equation}
		\Phi^*\left(\mathcal{L}_{\vec{V}}\mathcal A \right) =  \mathcal{L}_{\vec{V}_\abshyp}\left(\Phi^* \mathcal A \right). \label{eq:lemma_5_parallel}
		\end{equation}
		\item Let $\vec{V}$ be any vector field on a neighbourhood $\mathcal{U}$ of $\Phi(\abshyp)$ proportional to the selected rigging vector, i.e. $\vec{V} = F \vec{l}$. Then
		\begin{eqnarray}
		\Phi^*\left(\mathcal{L}_{\vec{V}} \mathcal A \right) &=& F \Phi^*\mathcal{L}_{\vec{l}} \mathcal A + \sum_{k=1}^r \mathcal A_{a_1 \dots \mu \dots a_r} l^\mu \vec{e}_{a_k}(F) \\
		&=& F \left(\Phi^* (\nabla_{\vec{l}}\mathcal A) + \sum_{k=1}^r l^\mu \mathcal A_{a_1 \dots \mu \dots a_r} \varphi_{a_k} + \sum_{k=1}^r \mathcal A_{a_1 \dots b \dots a_r} \Psi_{a_k}^b\right)\nonumber\\
		&+& \sum_{k=1}^r \mathcal A_{a_1 \dots \mu \dots a_r} l^\mu \vec{e}_{a_k}(F) \label{eq:lemma_5_rigging}
		%\\&=& e_a^\alpha e_b^\beta \left( \mathcal L_{F_1 %A^\mu_{\phantom \mu \nu} l^\nu}g_{\alpha \beta} - 2 F_1 %l_\mu S(A)^\mu_{\phantom \mu \alpha \beta}\right) \label{eq:lemma_5_rigging_S}
		\end{eqnarray}
		\item Let $\vec{V}$ be an arbitrary vector field on a neighbourhood $\mathcal{U}$ of $\Phi(\abshyp)$, admitting the decomposition $\vec{V} = F \vec{l} + \vec{T}$, with $\vec{T} \equiv d\Phi (\vec{T}_\abshyp)$. Then
		\begin{eqnarray}
		\Phi^*\left(\mathcal{L}_{\vec{V}} \mathcal A \right) &=&(F \circ \Phi)  \Phi^* \mathcal{L}_{\vec{l}}\left( \mathcal A \right) + \mathcal{L}_{\vec{T}_\abshyp}\left(\Phi^* \mathcal A \right) + \sum_{k=1}^r l^\mu \mathcal A_{a_1 \dots \mu \dots a_r} \vec{e}_{a_k}(F)\\
		&=&(F \circ \Phi) \left(\Phi^* (\nabla_{\vec{l}} \mathcal A) + \sum_{k=1}^r l^\mu \mathcal A_{a_1 \dots \mu \dots a_r} \varphi_{a_k} + \sum_{k=1}^r \mathcal A_{a_1 \dots b \dots a_r} \Psi_{a_k}^b \right) \nonumber \\
		&+ & \mathcal{L}_{\vec{T}_\abshyp}\left(\Phi^* \mathcal A \right) + \sum_{k=1}^r l^\mu \mathcal A_{a_1 \dots \mu \dots a_r} \vec{e}_{a_k}(F) \label{eq:lemma_5_general}
		\end{eqnarray}
		
	\end{enumerate}

\end{lemma}
\proof
The lemma follows from a direct computation. The left hand side in (\ref{eq:lemma_5_parallel}) can be expanded as
\begin{eqnarray}
\Phi^*\left(\mathcal{L}_{\vec{V}}A \right) =e^{\alpha_1}_{a_1} \cdot \cdot \cdot e^{\alpha_r}_{a_r} \lbrace (V^\mu \nabla_\mu A_{\alpha_1 \cdot \cdot \cdot \alpha_r}) + \sum_{k=1}^r A_{\alpha_1 \dots \alpha_{k-1}\mu\alpha_{k+1}\dots \alpha_r} \nabla_{\alpha_k}V^\mu\rbrace. \label{eq:lemma_5_parallel_1}
\end{eqnarray}
For the first summand we can apply formula (\ref{nabla_nablabar}) directly, while for the second we can expand the derivative in the basis  $\{\vec{l},\vec{e}_a\}$ as
\begin{equation*}
\nabla_{\vec{e}_a} V^\mu = (\nabla_{\vec{e}_a} V^\mu)^{\vec{l}} + (\nabla_{\vec{e}_a} V^\mu)^{||},
\end{equation*}
where the superscript $\vec{l}$ denotes the  component along the rigging and the $||$ is used for the completely tangent part. A straightforward calculation using formulas (\ref{nabla_normal_rigging}) and (\ref{nabla_nablabar}) provides
\begin{equation}
\nabla_{\vec{e}_a} V^\mu = -(\kappa_{ac}V^c) l^\mu + (\overline{\nabla}_a V^c)e_c^\mu.
\end{equation}
Inserting this decomposition into (\ref{eq:lemma_5_parallel_1}) the result (\ref{eq:lemma_5_parallel}) follows.
Equation (\ref{eq:lemma_5_rigging}) is obtained after expanding the Lie derivative in terms of the spacetime covariant derivative and applying (\ref{nabla_normal_rigging}). 

Finally, taking into account that the Lie derivative satisfies the following property
\begin{eqnarray*}
\LL_{\vec{v}} A_{\alpha_1 \dots \alpha_r} = \LL_{R\vec{l} + \vec{T}} A_{\alpha_1 \dots \alpha_r} &=& R \LL_{\vec{l}} A_{\alpha_1 \dots \alpha_r} +\sum_{k=1}^r l^\gamma A_{\alpha_1\dots \alpha_{k-1}\gamma\alpha_{k+1} \dots \alpha_r} \nabla_{\alpha_k} R \nonumber\\
&&+\LL_{\vec{T}} A_{\alpha_1 \dots \alpha_r},
\end{eqnarray*}
and applying results (\ref{eq:lemma_5_parallel}) and (\ref{eq:lemma_5_rigging}), we obtain (\ref{eq:lemma_5_general}).

\hfill $\blacksquare$

\begin{corollary}
	The pullback to $\abshyp$ of the Lie derivative of the metric along a transverse vector $F\vec{l}$ ($F$ is a function in $\mathcal{M}$) reads
	\begin{equation}
	\Phi^* \mathcal{L}_{F \vec{l}} g = \Phi^* \lbrace F \mathcal{L}_{ \vec{l}} g + dF \otimes \mathbf{l} + \mathbf{l} \otimes dF \rbrace = 2F Y_{ab} + l_a \partial_b F  + l_b \partial_a F , \label{Llmetric}
	\end{equation}
	where $\mathbf{l} \equiv g(\vec{l},\cdot)$. The same operation, but along a tangent vector $\vec{T}$ results into
	\begin{equation}
	\Phi^* \mathcal{L}_{ \vec{T}} g =   \mathcal{L}_{d\Phi \vec{T}_\abshyp} \Phi^* g = \mathcal{L}_{T_\abshyp} h_{ab} , \label{LTmetric}
	\end{equation}
	where $\vec{T} \equiv d\Phi \vec{T}_\abshyp$.
	\label{pullbackliemetric}
\end{corollary}

We provide a different computation of $\delta Y$. The general approach is completely equivalent to the one used in Proposition \ref{proposition_deltaY}, but some of the intermediate calculations are different, more similar to those in \cite{Mars2005} so that this stands mostly as a check. Before computing $\delta Y$, it is useful to state some useful lemmas.

\begin{lemma} \label{lemma:S} {\bf (Mars 2005 \cite{Mars2005})}\label{lemma:S_object} Let $\vec X$ be an arbitrary vector field and $A_{\alpha \beta}$ any symmetric two covariant tensor. Define 
	\begin{equation*}
	S(A)^\mu_{\phantom{\mu}\alpha \beta} \equiv \frac{1}{2} \left( \nabla_\alpha A^\mu_{\phantom{\mu}\beta} + \nabla_\beta A^{\phantom{\alpha}\mu}_{\alpha}-\nabla^\mu A_{\alpha \beta}\right), \quad H^\alpha \equiv A^{\alpha}_{\phantom{\alpha}\mu}X^\mu.
	\end{equation*}
	The following identity holds
	\begin{equation*}
	\left(\mathcal L_{\vec X} A\right)_{\alpha \beta} +2 X_\mu S(A)^\mu_{\phantom{\mu}\alpha \beta} = \left( \mathcal L_{\vec H}g\right)_{\alpha \beta}
	\end{equation*}
	
\end{lemma}
%{\color{red} Does $A$ need to be symmetric??}
\begin{lemma}{\bf (Mars 2005 \cite{Mars2005})}
	The second Lie derivative of a two covariant symmetric tensor field $A_{\alpha \beta}$ with respect to a vector $\vec{x}$ reads
	\begin{eqnarray}
	\LL^2_{\vec{x}} A_{\alpha \beta} &:=& \LL_{\vec{x}} \LL_{\vec{x}} A_{\alpha \beta}  \nonumber\\
	&=& \LL_{\nabla_{\vec{x}}\vec{x}} A_{\alpha \beta} + x^\mu x^\rho\nabla_\rho \nabla_\mu A_{\alpha \beta} + R_{\alpha \rho \eta}^\mu x^\eta x^\rho A_{\mu \beta} + R_{\beta \rho \eta}^\mu x^\eta x^\rho A_{\alpha\mu} \nonumber\\
	&&+2(\nabla_\alpha x^\mu)(x^\nu \nabla_\nu A_{\alpha \beta}) + 2 (\nabla_\beta x^\mu)(x^\nu \nabla_\nu A_{\alpha \mu}) + 2 A_{\mu \nu} (\nabla_\alpha x^\mu)(\nabla_\beta x^\nu)\nonumber
	\end{eqnarray}
	If $A_{\alpha \beta} = g_{\alpha \beta}$ then
	\begin{eqnarray}
	\LL^2_{\vec{x}} g_{\alpha \beta} &=& \LL_{\nabla_{\vec{x}}\vec{x}} g_{\alpha \beta} -2 R_{\alpha \gamma \beta \mu}x^\gamma x^\mu + 2(\nabla_\alpha x_\mu)(\nabla_\beta x^\mu)
	\end{eqnarray}
\end{lemma}
As in \cite{Mars2005} we apply the previous lemma to a vector field $\vec{x}$ which is proportional to the rigging vector. The following expression generalizes (23) from \cite{Mars2005}.
\begin{corollary} \label{corollary_doubleLie_metric}
	Let $A$ be the metric tensor $g$ on $\mathcal M$ and $F_1$, $F_2$ two arbitrary functions defined on a neighbourhood $U$ of $\Sigma_0 \equiv \Phi(\abshyp)$, and $\vec{a}$ the acceleration of the rigging vector.
	Then
	\begin{eqnarray*}
	\left.\LL_{F_1\vec{l}}\LL_{F_2\vec{l}}g_{\alpha \beta}\right|_{\Sigma_0} &=& F_1 \vec{l}(F_2) \mathcal{L}_{\vec{l}}g + \mathcal{L}_{F_1F_2 \vec a}g - 2F_1 a_{(\alpha} \nabla_{\beta)} F_2  \nonumber\\
	&+&2F_1F_2 \left \lbrace -R_{\alpha \gamma \beta \mu}l^\gamma l^\mu +a_\mu a^\mu n_\alpha n_\beta + 2n_{(\alpha}\omega^a_{\beta)}\left(\frac{1}{2}\vec{l}(l^\mu l_\mu)\varphi_a + a_b \Psi^b_a\right)\right.\nonumber\\
	&+&\left.\omega^a_\alpha \omega^b_\beta (\Psi^c_a \Psi^d_b h_{cd} + (l^\mu l_\mu) \varphi_a \varphi_b + \varphi_al_c \Psi^c_b + \varphi_b l_c \Psi^c_a)\right\rbrace\nonumber\\
	&+& F_2 (\nabla_{(\alpha} F_1) ( \nabla_{\beta)} (l^\mu l_\mu))\nonumber\\
	& +&2F_1 (\nabla_{(\alpha}F_2)( \LL_{\vec{l}}l_{\beta)}) + 2l_\mu l^\mu (\nabla_{(\alpha}F_1) (\nabla_{\beta)}F_2)\nonumber\\
	&+&2l_{(\alpha}\LL_{F_1 \vec{l}}\nabla_{\beta)}F_2,
	\end{eqnarray*}
	where the right hand side makes sense at points of $\Sigma_0$.
\end{corollary}

We use Corollary \ref{corollary_doubleLie_metric}, setting $F_1=Q$ and $F_2=1$, to compute the second derivative of the metric tensor at points of the hypersurface $\Sigma_0$
	\begin{eqnarray}
	\left. \LL_{Q\vec{l}}\LL_{\vec{l}}g_{\alpha \beta} \right|_{\Sigma_0} &=&  \LL_{Q \vec{a}}g +  (\nabla_{(\alpha} Q) ( \nabla_{\beta)} l^\mu l_\mu)\nonumber\\
	&+&2Q \left \lbrace -R_{\alpha \gamma \beta \mu}l^\gamma l^\mu +a_\mu a^\mu n_\alpha n_\beta + 2n_{(\alpha}\omega^a_{\beta)}(a_\nu l^\nu \varphi_a + a_\nu e^\nu_b \Psi^b_a)\right.\nonumber\\
	&+&\left.\omega^a_\alpha \omega^b_\beta (\Psi^c_a \Psi^d_b h_{cd} + l_\mu l^\mu \varphi_a \varphi_b+ \varphi_al_c \Psi^c_b + \varphi_b l_c \Psi^c_a)\right\rbrace .\nonumber
	\end{eqnarray}
	With this formula, and using the properties of the Lie derivative  and subcase 1 of Lemma \ref{lemma_pullbackLie} for derivatives along tangent vectors we can work out the third term in (\ref{mastereq}) for $\delta Y$, which reads
	\begin{eqnarray*}
	\Phi^* \left( \mathcal{L}_{\vec{Z}} \mathcal{L}_{\vec{l}} g\right) &=& \Phi^* \left( \mathcal{L}_{Q\vec{l}+\vec{T}} \mathcal{L}_{\vec{l}} g\right) = \Phi^* \left( \mathcal{L}_{Q\vec{l}} \mathcal{L}_{\vec{l}} g\right) + \Phi^* \left( \mathcal{L}_{\vec{T}} \mathcal{L}_{\vec{l}} g\right)\nonumber\\
	&=&\Phi^* \left( \mathcal{L}_{Q\vec{a}}  g\right) + \vec{e}_{(a}Q\vec{e}_{b)}l^{(2)} + 2Q \left( -R_{\alpha \gamma \beta \mu}l^\gamma l^\mu e_a^\alpha e_b^\beta + \Psi^c_a \Psi^d_b h_{cd} \right.\nonumber\\
	&&\left.  + l^{(2)} \varphi_a \varphi_b+ \varphi_al_c \Psi^c_b + \varphi_b l_c \Psi^c_a \right) + 2 \mathcal{L}_{\vec{T}_\abshyp} Y.
	\end{eqnarray*}
	The remaining two terms in (\ref{mastereq}) can be expanded applying Lemma \ref{lemma_pullbackLie}, resulting in
	\begin{eqnarray}
	\Phi^* \left( \mathcal{L}_{\vec{l}_1}g\right) + \Phi^* \left( \mathcal{L}_{\vec{l}}g_1\right) &=& 2\alpha Y + l_a \vec{e}_b (\alpha) + l_b \vec{e}_a (\alpha) + \mathcal{L}_{\vec{s}_\abshyp}h \nonumber \\
	&+& \Phi^* \left(\nabla_{\vec{l}}g_1\right) + g_1(\vec{l},\vec{e}_a) \varphi_b + g_1(\vec{l},\vec{e}_b) \varphi_a + g_1(\vec{e}_a,\vec{e}_c)  \Psi^c_b + g_1(\vec{e}_b,\vec{e}_c)  \Psi^c_a.\nonumber
	\end{eqnarray}
	
	Putting all these results together, we get the following expression
	\begin{eqnarray}
	\delta Y_{ab} &=&\frac{1}{2}\Phi^* \left( \mathcal{L}_{Q\vec{a}}  g\right) + \frac{1}{2}\vec{e}_{(a}Q\vec{e}_{b)}l^{(2)} + Q \left( -R_{\alpha \gamma \beta \mu}l^\gamma l^\mu e_a^\alpha e_b^\beta + \Psi^c_a \Psi^d_b h_{cd} \right.\nonumber\\
	&+&\left.   l^{(2)} \varphi_a \varphi_b+ \varphi_al_c \Psi^c_b + \varphi_b l_c \Psi^c_a \right) +  \mathcal{L}_{\vec{T}_\abshyp} Y + \alpha Y + l_{(a} \vec{e}_{b)} (\alpha)  + \frac{1}{2} \mathcal{L}_{\vec{s}_\abshyp}h +\nonumber \\
	&+&  \frac{1}{2}\left \lbrace \Phi^* \left(\nabla_{\vec{l}}g_1\right) + g_1(\vec{l},\vec{e}_a) \varphi_b + g_1(\vec{l},\vec{e}_b) \varphi_a + g_1(\vec{e}_a,\vec{e}_c)  \Psi^c_b + g_1(\vec{e}_b,\vec{e}_c)  \Psi^c_a \right \rbrace, \label{Yp_1}
	\end{eqnarray}
	that can be written in terms of the vector $\vec{\zeta}$ to produce (\ref{Yp_1_b}).

\bibliographystyle{review_bib}

\bibliography{references}

\end{document}